\documentclass[final,onefignum,onetabnum,11pt]{siamonline220329}

\usepackage{lipsum}
\usepackage{amsfonts}
\usepackage{epstopdf}
\usepackage{algorithmic}
\usepackage{multicol}
\usepackage{vwcol}
\usepackage{tcolorbox}
\usepackage{xr}
\usepackage{pdfpages}
\usepackage{amsmath,amssymb}
\usepackage[leqno]{amsmath}
\usepackage{graphicx}
\usepackage{subcaption}
\usepackage{placeins}
\usepackage{bm}

\newcommand{\Ra}[1]{{\color{black} {#1}}}
\newcommand{\Rb}[1]{{\color{black} {#1}}}
\newcommand{\Rc}[1]{{\color{black} {#1}}}

\newcommand{\Ncomb}[4]{\Ra{|\set{Z}_{#1}^{#2}\cap \set{Z}_{#3}^{#4}|}}
\newcommand{\Nsing}[2]{\Ra{|\set{Z}_{#1}^{#2}|}}

\makeatletter
\newcommand{\leqnomode}{\tagsleft@true}
\newcommand{\reqnomode}{\tagsleft@false}
\makeatother

\usepackage{enumitem}
\setlist[enumerate]{leftmargin=.5in}
\setlist[itemize]{leftmargin=.5in}

\newsiamremark{remark}{Remark}
\newsiamremark{hypothesis}{Hypothesis}
\crefname{hypothesis}{Hypothesis}{Hypotheses}
\newsiamthm{claim}{Claim}

\headers{Covariance Expressions for Multi-Output Estimation}{T. O. Dixon, J. E. Warner, G. F. Bomarito, and A. A. Gorodetsky}

\title{Covariance Expressions for Multi-Fidelity Sampling 
 with Multi-Output, Multi-Statistic Estimators: Application to Approximate Control Variates \thanks{Submitted to the editors on October 12, 2023.
\funding{This work was funded by the NASA Project: Entry Systems Modeling under grant no. 80NSSC22K1007.}}}

\author{Thomas O. Dixon \thanks{Department of Aerospace Engineering, University of Michigan, Ann Arbor, MI 
  (\email{tdixono@umich.edu}, \email{goroda@umich.edu}).}
\and James E. Warner \thanks{Durability, Damage Tolerance, and Reliability Branch, NASA Langley Research Center, Hampton, VA  
  (\email{james.e.warner@nasa.gov}, \email{geoffrey.f.bomarito@nasa.gov}).}
  \and Geoffrey F. Bomarito \footnotemark[3]
\and Alex A. Gorodetsky \footnotemark[2] }

\usepackage{amsopn}

\newcommand{\reals}{\mathbb{R}}
\newcommand{\mat}[1]{\bm{\mathbf{#1}}}
\renewcommand{\vec}[1]{\bm{\mathbf{#1}}}
\newcommand{\var}[1]{\mathbb{V}ar\left[#1\right]}
\newcommand{\cov}[1]{\mathbb{C}ov\left[#1\right]}
\newcommand{\ee}[1]{\mathbb{E}\left[#1\right]}
\newcommand{\un}[1]{\underline{#1}}

\newcommand{\cv}{\Tilde{\bold Q}}
\newcommand{\set}[1]{\mathcal{#1}}

\newcommand\scalemath[2]{\scalebox{#1}{\mbox{\ensuremath{\displaystyle #2}}}}

\newtheorem{prop}{Proposition}[section]

\begin{document}

\maketitle

\begin{abstract}
We provide a collection of results on covariance expressions between Monte Carlo based multi-output mean, variance, and Sobol main effect variance estimators from an ensemble of models. These covariances can be used within multi-fidelity uncertainty quantification strategies that seek to reduce the estimator variance of high-fidelity Monte Carlo estimators with an ensemble of low-fidelity models. Such covariance expressions are required within approaches like the approximate control variate and multi-level best linear unbiased estimator. While the literature provides these expressions for some single-output cases such as mean and variance, our results are relevant to both multiple function outputs and multiple statistics across any sampling strategy. Following the description of these results, we use them within an approximate control variate scheme to show that leveraging multiple outputs can dramatically reduce estimator variance compared to single-output approaches. Synthetic examples are used to highlight the effects of optimal sample allocation and pilot sample estimation. A flight-trajectory simulation of entry, descent, and landing is used to demonstrate multi-output estimation in practical applications.
\end{abstract}



\section{Introduction}

Estimating statistics of simulation models is of primary concern in uncertainty quantification. However, sampling strategies for estimation are often plagued by slow convergence. For example, the variance of a Monte Carlo (MC) mean estimator is proportional to the inverse of the number of model evaluations, requiring an order of magnitude more samples per digit of accuracy. As a result, the large number of sample evaluations required for accurate estimation becomes prohibitive when the underlying model is computationally burdensome. In this paper, we consider variance reduction techniques that reduce this cost by leveraging ensembles of correlated multi-output models for multiple statistics at once.

We focus on multi-fidelity sampling strategies that extract information from models of varying fidelities to reduce the variance of a baseline estimator without introducing bias. These lower fidelity models can take a hierarchical form, for example arising from a hierarchy of discretizations of a finite-element PDE approximation \cite{MLMC1, Grid}, or they may be unstructured and include simulations with different physics and/or surrogates \cite{BLMF, Traj}. In the context of multi-fidelity variance reduction, we focus on control-variate (CV) methods \cite{CV1, CV, MVCV}. Examples of CV methods include the multi-level MC (MLMC) estimator \cite{MLMC, MLMC1}; the multi-fidelity MC (MFMC) \cite{MFMC1}, and more generally appproximate control variates (ACV)~\cite{ACV}. While MLMC and MFMC require a distinct sampling structure of the ensemble of models, potentially limiting achievable variance reduction, the ACV method provides a general framework for distributing samples amongst models. More recently, the multi-level best linear unbiased estimator (MLBLUE) provides an alternate method to allocate samples based on estimator and model groupings \cite{MLBLUE1}, but can also be interpreted under the ACV framework \cite{GroupACV}. 

Effectively leveraging multiple fidelities of models requires knowledge of the covariance between all models involved. As such, all of the above approaches require the prior knowledge of the covariance between the ensemble of high and low-fidelity estimators. These estimator covariances are intimately tied to the statistics being estimated. A majority of the literature focuses on mean estimation of scalar-valued functions \cite{CV, CV1, ACV, Param, MFMC1, MLMC}. Some works on other statistics such as the variance~\cite{MFMC,HDMFMC,MLBLUE_E}, Sobol indices \cite{MFMC, CVSobol,HDMFMC,MLMCSobol}, and quantiles \cite{CVQE}, also exist, but focus on single-statistic estimators. Note that MLMC does not require estimator covariances by making the strong assumption of perfect correlation amongst models, and, as a result, can yield sub-optimal choices of CV weights when the models are not perfectly correlated \cite{ACV}. 

In the case of mean estimation, the covariance between MC estimators of each model is easily related to the covariances of the underlying models themselves \cite{ACV, Param}. In practice, these model covariances are generally unknown, but estimated via some pilot sampling procedure. Pilot sample estimation can be performed with a fixed number of samples or through more adaptive or robust schemes \cite{CVR}. For example, an exploration-exploitation approach can be taken to minimize the total cost of model evaluations by determining when to stop estimating the model covariances \cite{BLMF}. Another approach directly estimates the covariance of the estimators by creating an ensemble of ACV estimators, each with a different set of samples~\cite{EACV1}. For other statistics, such as probablility, quantile, or Sobol index estimation, the covariance between estimators is generally unavailable \cite{MFMC,CVQE,CVSobol,HDMFMC,MLMCSobol}. For variance and Sobol index estimation, \cite{MFMC} finds the optimal weights for mean estimation and applies them to high-order statistic estimation. In \cite{CVSobol} and \cite{MLMCSobol}, perfect correlation between estimators is assumed for Sobol index estimation, disregarding the estimator covariance requirement, but resulting in sub-optimal CV estimation. Finally, \cite{HDMFMC} numerically estimates the covariance between estimators directly for variance and Sobol indices to find the optimal CV weights. One of the principal aims of this paper is to introduce the analytic covariances for these additional estimators to improve CV efficiency.

A second issue that we consider is models with multiple outputs --- the majority of the above approaches are applied to models with single outputs. Extending these approaches to vector-valued functions requires additional covariance expressions. Current state-of-the-art estimation techniques that use multiple quantities of interest (QoIs) construct one estimator for each QoI \cite{Traj, HDMFMC, MLBLUE, MLBLUE_E}. While creating individual estimators is a simple technique, the correlations between the QoIs are lost, which leads to limited variance reduction. Indeed, in the context of classical CVs, Rubinstein and Marcus \cite{MVCV} show that the correlation between model outputs can be extracted by including vector-valued functions in a single estimator to further reduce the estimator variance. We extend these results to the ACV context. Additionally, estimating multiple statistics in a single estimator can lead to further sources of correlation, which can be extracted to reduce the variance of both statistics' estimators. We newly introduce approaches to leverage multi-statistic information here.

In the context of multi-output mean estimation, a recent approach using the MLBLUE estimator was introduced to indeed extract model output correlations for vector-valued mean estimation \cite{MLBLUE_E}. A covariance matrix estimation approach was also introduced, but lacked the capability of extracting correlations between model outputs in this case. Similarly to \cite{MLBLUE}, independent MLBLUE estimators were stacked into a matrix for vector-valued estimation. The approaches in this paper are applicable to further extending these MLBLUE results to take advantage of the correlations between model outputs for covariance matrix estimation. Similarly, this work introduces multi-statistic estimators for mean, variance, and Sobol indices which can further be applied to MLBLUE estimation.

We now summarize our contributions.
First, we derive estimator covariances for multiple statistics and vector-valued functions for several important cases of interest that can be utilized in the majority of multi-fidelity sampling strategies. Propositions \ref{meancov} and \ref{var_cov} provide the covariance between mean estimators and variance estimators, respectively, for vector-valued functions. Proposition \ref{MeanCov} provides the covariance between the mean and variance estimators for simultaneous mean and variance estimation.

Second, we derive estimator covariances for all the main effect variances of scalar-valued functions for use in Sobol indices for global sensitivity analysis. The covariance between main effect variance estimators of similar/different indices is seen in Proposition \ref{SobolEstCov2}. Similarly, Proposition \ref{VarSobolCov} provides the covariance between the variance and main effect variance estimators since the total variance of the model is required for Sobol index estimation. These covariances allow multiple Sobol indices to be estimated simultaneously, providing a thorough sensitivity analysis across multiple inputs. 

Finally, while these results can be adapted to several schemes, we utilize them to introduce the multi-output ACV (MOACV) estimator. This estimator can simultaneously estimate multiple statistics for vector-valued functions. We provide a number of empirical results that demonstrate that the MOACV estimator outperforms individual ACV estimators. As part of these results, we demonstrate that the newly derived estimator covariances for mean estimation do not require substantially more pilot samples than traditional ACV estimation. Finally, the MOACV estimator is tested on a realistic application of trajectory estimation for entry, descent, and landing (EDL) . The numerical results demonstrate significant further variance reduction compared to existing results. 

The rest of this paper is structured as follows. Section \ref{sec:background} introduces MC sampling and the multi-output ACV theory. Section \ref{sec:Body} provides the introduced estimator covariances and how to apply them to the ACV techniques. The results in Section \ref{sec:numerical} demonstrate the MOACV capabilities on analytical examples. Finally, Section \ref{sec:application} applies the MOACV estimator to the EDL application.

\section{Background}
\label{sec:background}
In this section, we introduce notation, the core sampling-based estimators, and multi-fidelity variance reduction approaches.

\subsection{Notation}
The following notation is used throughout the manuscript. 
Matrices and vectors are denoted by bold-faced Roman letters. Each element of a matrix $\mat{F} \in \reals^{A \times B}$ is denoted as $F_{ab}$ for $\{a,b\} \in \{0, 1,\ldots, A-1 \}\times \{0, 1,\ldots, B-1\}$. Similarly each element of a vector  $\vec{g}\in \reals^A$ is denoted by $g_a$ for $a\in \{0, 1,\ldots, A-1 \}$. We denote a matrix of ones with size $A\times B$ as $\bold 1_{A\times B}$. Generally, block matrices use an underline to denote that the block structure is important. If $\un{\mat{F}}$ is an $A \times B$ block matrix, then its blocks are denoted by $\mat{F}_{ab}.$ 

The Kronecker product between vectors $\vec{X}, \vec{Y} \in \reals^D$ is treated as a flattened outer product $\vec{X} \otimes \vec{Y} = \textrm{vec}(\vec{X} \vec{Y}^T).$ 
The element-wise product between two vectors or matrices is written as  $\vec{X} \circ \vec{Y}$. The square of these two operations uses the following shorthand for both vectors and matrices $\vec{X}^{\otimes 2} \equiv \vec{X}\otimes\vec{X} $ and $ \bold X^{\circ 2} \equiv \bold X \circ \bold X$, respectively. 
Sets are denoted via upper case calligraphic letters such as $\set{Z}.$ \Ra{We denote the size of $\set{Z}$ to be $\Nsing{}{}$ and the size of the intersection between two sets, $\set{X}$ and $\set{Y}$ as $|\set{X}\cap \set{Y}|$.}

\Ra{Let the covariance between two vectors be defined as $\cov{\vec{X},\vec{Y}} = \mathbb{E}[(\mathbf{X}-\mathbb{E}[\mathbf{X}])(\mathbf{Y}-\mathbb{E}[\mathbf{Y}])^T]$, where the variance is a special case, $\mathbb{V}ar[\mathbf{X}] = \mathbb{C}ov[\mathbf{X},\mathbf{X}]$.}

\subsection{Monte Carlo Estimators}
\label{sec:MCests}

\Ra{Consider a probability space ($\Omega$, $\set{F}$, $\mathbb{P}$), and let the random variable $\vec{z} \in \reals^I$ denote an $\reals^I$-valued random vector having law $\mathbb{P}_{\vec{z}}$ defined on this probability-space. Furthermore, let $\Ra{\vec{f}}: \reals^I \to \reals^D$ be a function with input $\vec{z}$, such that $\Ra{\vec{f}}(\vec{z})$ becomes an $\reals^D$-valued random vector with finite mean $\mu \equiv \ee{\Ra{\vec{f}(\vec{z})}} \in \reals^D$ and finite covariance $V \equiv \var{\Ra{\vec{f}} \Ra{(\vec{z})}} \in \reals^{D \times D}$. The $N$-sample MC estimator of the mean is defined using a set of $N$ independent and identically distributed random variables $\set{Z} = \{\vec{z}^{(s)} ; s=1,\ldots,N \}$, each with law $\mathbb{P}_{\vec{z}}$, according to}
\begin{align}
\vec{Q}_{\mu}(\set{Z}) = \frac{1}{N}\sum_{s=1}^N \Ra{\vec{f}}(\vec{z}^{(s)}),
\label{eq:meanMC}
\end{align}
\Ra{so that $\vec{Q}_{\mu}(\set{Z})$ is an $\reals^D$-valued random vector.}
This estimator is unbiased so that \Ra{$\ee{\vec{Q}_{\mu}(\set{Z})} = \mu$,} and it has variance  
 \Ra{$\var{\mat{Q}_{\mu}\Ra{(\set{Z})}} = V / N$.}
A MC estimator for the covariance is
\begin{align}
    \mat{Q}_{V}(\set{Z}) &= \frac{1}{N-1}\sum^N_{s=1} \left(\Ra{\vec{f}}(\vec{z}^{(s)}) - \mat{Q}_{\mu}(\set{Z}) \right)^{\otimes 2} 
 = \frac{1}{2N(N-1)}\sum^N_{s=1}\sum^N_{t=1} \left( \Ra{\vec{f}}(\vec{z}^{(s)})-\Ra{\vec{f}}(\vec{z}^{(t)}) \right)^{\otimes 2},
    \label{eq:varianceMC}
\end{align}
where $\mat{Q}_{V}\Ra{(\set{Z})}$ is a flattened estimate of the covariance matrix, \Ra{becoming a $\reals^{D^2}$-valued random vector.} Its variance $\var{\mat{Q}_{V}\Ra{(\set{Z})}} \in \reals^{D^2 \times D^2}$ is
\begin{align}
\scalemath{1}{
        \var{\mat{Q}_{V}\Ra{(\set{Z})}}} &= \scalemath{1}{
    \frac{1}{N(N-1)}\left[ \var{\Ra{\vec{f}}\Ra{(\vec{z})}}^{\otimes 2} + \left( \bold 1_D^T \otimes \var{\Ra{\vec{f}}\Ra{(\vec{z})}} \otimes \bold 1_D \right) \circ \left(\bold 1_D \otimes \var{\Ra{\vec{f}}\Ra{(\vec{z})}} \otimes \bold 1_D^T \right)\right]} \nonumber\\
    &\quad \scalemath{1}{+ \frac{1}{N}\mathbb{V}ar\left[\left(\Ra{\vec{f}}\Ra{(\vec{z})}- \ee{\Ra{\vec{f}}\Ra{(\vec{z})}}\right)^{\otimes 2}\right]},
    \label{eq:varMCvar}
\end{align}
\Ra{which} follows from Proposition \ref{var_cov}. 

Finally, we consider MC estimators for main effect Sobol sensitivity indices. To this end, the ANOVA decomposition~\cite{ANOVA} of the variance of \Ra{a scalar-valued} $f\Ra{(\vec{z})}$ is
\begin{align}
    \mathbb{V}ar[f\Ra{(\vec{z})}] =  \sum_{u=1}^I V_u +  \sum_{u,v; u>v}^I V_{uv} + \sum_{u,v,w; u>v>w}^I V_{uvw} + \cdots ,
\end{align}
where $V_u = \mathbb{V}ar_{\Ra{z}_{u}}[\mathbb{E}_{\Ra{\vec{z}}_{\sim u}}[f(\Ra{\vec{z}})|\Ra{z}_{u}]]$ \Ra{such that} $\Ra{z}_{u}$ is the $u$-th input variable \Ra{and $\vec{z}_{\sim u}$ is a vector of all input variables, excluding the $u$-th input.} The ANOVA decomposition separates the variance into terms attributed to the function's inputs. One sensitivity measure is the global sensitivity index, or Sobol index \cite{Sobol}, $s_{u_1\cdots u_I} = \frac{V_{u_1,\ldots,u_I}}{V}$
which is the percentage of variance attributed to the corresponding term of the ANOVA decomposition. 

In this paper, we  focus on the main effect sensitivity indices $s_u=\frac{V_u}{V}.$ The Sobol estimator for the main effect can be obtained using two sets of \Ra{independent and identically distributed random variables}: $\Ra{\set{X}} = \{ \Ra{\vec{x}}^{(s)}; s=1,\ldots,N \}$, and $\set{Y}_u = \{\vec{y}_u^{(s)}; s=1,\ldots,N \}$, where \Ra{$\set{X}$ and }$\set{Y}_u$ \Ra{contain independent random variables} with the exception of the $u$-th input, i.e., $\vec{y}^{(s)}_{u} = (y^{(s)}_1, y^{(s)}_2, \cdots, \Ra{x}^{(s)}_u, \cdots, y^{(s)}_I)^T$ for $s=1,\ldots,N$. Using these sets\Ra{, $\set{Z} = \{ \set{X}, \set{Y}_u \}$,} the estimator for $V_u$ is 
\begin{align}
\scalemath{0.9}{
    \mat{Q}_{V_u}(\set{Z}) = \frac{1}{N}\sum_{s=1}^N f(\Ra{\vec{x}}^{(s)})f(\vec{y}^{(s)}_{u}) - \left( 
\frac{1}{N}\sum_{s=1}^N f(\Ra{\vec{x}}^{(s)}) \right)^2 = \frac{1}{N^2}\sum_{s=1}^N\sum_{t=1}^N \left[f(\Ra{\vec{x}}^{(s)})f(\vec{y}^{(s)}_{u}) - f(\Ra{\vec{x}}^{(s)})f(\Ra{\vec{x}}^{(t)}) \right] }.
\label{eq:Sobolest}
\end{align}
These estimators have a bias of $-\mathbb{V}ar[f\Ra{(\vec{z})}]/N$ \cite{SobolBias}, but are used for their simplicity. The variance of the Sobol estimator is
\begin{align}
        \mathbb{V}ar[\mat{Q}_{V_u}\Ra{(\set{Z})}] &= \frac{1}{N^3}\left[ (N-1)^2\mathbb{V}ar\left[f\Ra{(\vec{x})}f\Ra{(\vec{y}_u)}-2f\Ra{(\vec{x})}\mathbb{E}[f\Ra{(\vec{x})}]\right]  \right. \notag\\
        &\quad \left. + 2(N-1) \mathbb{C}ov[f\Ra{(\vec{x})}f\Ra{(\vec{y}_u)}-2f\Ra{(\vec{x})}\mathbb{E}[f\Ra{(\vec{x})}], f\Ra{(\vec{x})} f\Ra{(\vec{y}_u)}-f\Ra{(\vec{x})}^2] \right. \notag\\
    &\quad + \mathbb{V}ar[ f\Ra{(\vec{x})}f\Ra{(\vec{y}_u)}-f\Ra{(\vec{x})}^2] + 2(N-1)\mathbb{V}ar[f\Ra{(\vec{x})}]^2\left. \right].
    \label{eq:meMCvar}
\end{align}
To the best of our knowledge, Equation \eqref{eq:meMCvar} is introduced in this paper and follows from Proposition \ref{SobolEstCov2}.

\Rb{The form of the main-effect variance estimator in Equation \eqref{eq:Sobolest} was chosen because it was the most straightforward to calculate the covariance between two estimators of this type. There are, however, other forms of the main effect variance estimator, such as the Satelli, Janon, or Owen estimators \cite{MFMC} that may have reduced variance compared to Equation \eqref{eq:Sobolest}. To find the optimal ACV weights, we would need to find the covariance between estimators of these forms. The derivation would follow a similar structure to the proofs in Appendix \ref{appendix:SobolEstCov}. }

\subsection{Multi-Output Control Variates}
\label{sec:cv}
The estimator variances described above all decrease at a rate of $1/ N$, which is prohibitive for expensive function evaluations. \Rb{While the variance of multi-fidelity estimators based on linear control variates that we consider here also decay at this rate, } variance reduction methods reduce this expense \Rb{by changing the constant factor}. CV approaches reduce variance by leveraging additional estimators with known statistics \cite{CV}.

\Ra{Let $\set{Z} = \{\set{Z}_0,\set{Z}^*\}$ denote a set of inputs.} Let \Ra{$\vec{Q}\Ra{(\set{Z}_0)}$ be an $\reals^D$-valued estimator,} and let $\vec{Q}^*\Ra{(\set{Z}^*)}$ denote \Ra{an $\reals^E$-valued} random vector with known mean $\vec{q}^* \equiv \ee{\vec{Q}^*\Ra{(\set{Z}^*)}}$. The CV estimator $\cv\Ra{(\set{Z})}$ is defined by  
\begin{align}
    \cv(\Ra{\set{Z}, }\mat{\alpha}) &= \vec{Q}\Ra{(\set{Z}_0)} + \mat{\alpha} (\vec{Q}^*\Ra{(\set{Z}^*)} - \vec{q}^*) 
    = \vec{Q}\Ra{(\set{Z}_0)} + \mat{\alpha}\un{\vec{\Delta}}\Ra{(\set{Z}^*)},
\end{align}
where $\mat{\alpha} \in \reals^{D\times E}$ is a matrix of weights and $\un{\vec{\Delta}}\Ra{(\set{Z}^*)} \equiv \vec{Q}^*\Ra{(\set{Z}^*)} - \vec{q}^*$. This new estimator $\Tilde{\bold Q}\Ra{(\set{Z})}$, has the same mean as $\bold Q \Ra{(\set{Z}_0)}$. Furthermore, its variance is
 \begin{align}
 \scalemath{0.9}{
     \mathbb{V}ar[\Tilde{\bold Q}\Ra{(\set{Z})}](\mat{\alpha}) = 
     \mathbb{V}ar[\bold Q\Ra{(\set{Z}_0)}] + \mat{\alpha} \mathbb{V}ar[\un{\vec{\Delta}}\Ra{(\set{Z}^*)}]\mat{\alpha}^T + \mathbb{C}ov[\bold Q\Ra{(\set{Z}_0)}, \un{\vec{\Delta}}\Ra{(\set{Z}^*)}]\mat{\alpha}^T + \mat{\alpha}\mathbb{C}ov[\un{\vec{\Delta}}\Ra{(\set{Z}^*)},\bold Q\Ra{(\set{Z}_0)}]. }
     \label{eq:MFempvar}
 \end{align}
 The weights $\mat{\alpha}$ can be chosen to minimize some scalar-valued measure of the uncertainty represented by this variance. Rubinstein and Marcus \cite{MVCV} minimize the determinant, yielding
 \begin{align}
 \scalemath{0.9}{
     \mat{\alpha}^* = -\mathbb{C}ov[\bold Q\Ra{(\set{Z}_0)}, \un{\vec{\Delta}}\Ra{(\set{Z}^*)}]\mathbb{V}ar[\un{\vec{\Delta}}\Ra{(\set{Z}^*)}]^{-1}}
     \label{eq:OptWei}
\end{align}
with variance
\begin{align}
      \scalemath{0.9}{
          \mathbb{V}ar[\Tilde{\bold Q}\Ra{(\set{Z})}] = \mathbb{V}ar[\bold Q\Ra{(\set{Z}_0)}] -\mathbb{C}ov[\bold Q\Ra{(\set{Z}_0)}, \un{\vec{\Delta}}\Ra{(\set{Z}^*)}]\mathbb{V}ar[\un{\vec{\Delta}}\Ra{(\set{Z}^*)}]^{-1}\mathbb{C}ov[\bold Q\Ra{(\set{Z}_0)}, \un{\vec{\Delta}}\Ra{(\set{Z}^*)}]^T}.
     \label{eq:OptVar}
 \end{align}
The determinant of the variance \Ra{is} $|\mathbb{V}ar[\Tilde{\bold Q}\Ra{(\set{Z})}]| = \left|\mathbb{V}ar[\vec{Q}\Ra{(\set{Z}_0)}]\right|~\left[\prod_{d=1}^{\min(D,E)} (1 - \rho_d^2)\right] $
 where $\{\rho_d\}$ are the canonical correlations between $\vec{Q}\Ra{(\set{Z}_0)}$ and $\un{\vec{\Delta}}\Ra{(\set{Z}^*)}$ ~\cite{MVCV}. \Ra{Canonical correlations between two vectors, $\vec{X}$ and $\vec{Y}$, are found by maximizing the correlation $\rho$ between their linear combinations, $\vec{U} = a^T\vec{X}$ and $\vec{V}=b^T\vec{Y}$ where $a$ and $b$ are coefficient vectors \cite{CCA}.} Clearly, greater (anti)-correlations yield greater reductions in variance. 
 
\Ra{The minimization of the determinant of the variance is equivalent to minimizing the confidence region volume of the estimator \cite{MVCV}. While Rubinstein and Marcus provide a proof of Equation \eqref{eq:OptWei} by minimizing the variance's determinant, Equation \eqref{eq:OptWei} is also optimal when minimizing the trace of the estimator variance (proof in Appendix \ref{sec:ACVTrace}).}

If $\Tilde{\bold Q}\Ra{(\set{Z})}$ \Rc{contains random variables that are linear combinations of each other, then $\var{\vec{\un{\Delta}}\Ra{(\set{Z}^*)}}$ becomes singular in Equation \eqref{eq:OptWei}.} 
\Rc{This situation is guaranteed to occur when considering estimation of the covariance matrix} \Rc{in Equation \eqref{eq:varianceMC} because a covariance matrix (and its corresponding estimator) has duplicate entries in the upper and lower triangular portions of the matrix.}
\Rc{For example, a 2$\times$2 variance matrix
\begin{align}
    \scalemath{0.9}{\var{\vec{X}} = \begin{bmatrix}
        \cov{X_1, X_1} & \cov{X_1, X_2} \\
        \cov{X_2, X_1} & \cov{X_2, X_2}
    \end{bmatrix} \quad \Longrightarrow\quad \begin{bmatrix}
        \cov{X_1, X_1} \\
        \cov{X_1, X_2} \\
        \cov{X_2, X_1} \\
        \cov{X_2, X_2}
    \end{bmatrix}}
\label{eq:covexample}
\end{align}
contains duplicate entries where $\cov{X_1, X_2} = \cov{X_2, X_1}$, causing the variance matrix to be singular.} 
The troubles of inverting a singular $\var{\vec{\un{\Delta}}\Ra{(\set{Z}^*)}}$ can be avoided \Rc{entirely} by removing duplicate outputs\Rc{, such as $\cov{X_2, X_1}$ in Equation \eqref{eq:covexample}.}

In the context of estimating statistics of computational models, the random variables $\mat{Q}\Ra{(\set{Z}_0)}$ and $\un{\vec{\Delta}}\Ra{(\set{Z}^*)}$ are the estimators using high- and low-fidelity models, respectively. In the uncertainty quantification problem, $\un{\vec{\Delta}}\Ra{(\set{Z}^*)}$ typically arises from an ensemble of $K$ lower-fidelity estimators $\left(\vec{Q}_{k}\Ra{(\set{Z}_k^*)}\right)_{k=1}^K$, \Ra{where $\set{Z}^* = \{\set{Z}^*_1,\ldots,\set{Z}^*_K\}$,} according to\footnote{In this work, we assume for simplicity that all models share the same number of outputs. This assumption, however, can easily be disregarded by changing the shapes of the defined covariance matrices. The theory and results that follow can be easily modified to allow for varying quantities of model outputs.}
\begin{align}
    \scalemath{0.9}{\un{\vec{\Delta}}\Ra{(\set{Z}^*)} = \begin{bmatrix}
        \bold \Delta_1\Ra{(\set{Z}_1^*)} \\
        \vdots \\
        \bold \Delta_K\Ra{(\set{Z}_K^*)}
    \end{bmatrix} = \begin{bmatrix}
        \vec{Q}_1\Ra{(\set{Z}_1^*)} - \ee{\vec{Q}_1\Ra{(\set{Z}_1^*)}} \\
        \vdots \\
        \vec{Q}_K\Ra{(\set{Z}_K^*)} - \ee{\vec{Q}_K\Ra{(\set{Z}_K^*)}}
    \end{bmatrix} }
    \label{eq:delta_multiple}
\end{align}
\Ra{where $\un{\vec{\Delta}}\Ra{(\set{Z}^*)}$ is a $\reals^{DK}$-valued random vector, such that $E=DK$.}

\subsection{Multi-Output Approximate Control Variates} \label{sec:moacv}

In the UQ setting, $(\ee{\vec{Q}_k\Ra{(\set{Z}_k^*)}})_{k=1}^K$ \Ra{in Equation~\eqref{eq:delta_multiple}} are unknown. One approach to overcome this issue is to introduce new estimators for these terms and form an approximate control variate (ACV)~\cite{ACV}. The ACV estimators have only been defined in the scalar-function context, but we extend them here to vector-valued estimators by following the same ideas as in Section \ref{sec:cv}:
\begin{align}
    \Tilde{\bold Q}(\mat{\alpha}, \set{Z}) &= \bold Q(\set{Z}_0) + \mat{\alpha} \begin{bmatrix}
        \bold Q_1(\set{Z}_1^*) - \bold Q_1(\set{Z}_1) \\
        \vdots \\
        \bold Q_K(\set{Z}_K^*) - \bold Q_K(\set{Z}_K)
    \end{bmatrix}=
    \bold Q(\set{Z}_0)
    + \mat{\alpha} \begin{bmatrix}
        \bold \Delta_1\Ra{(\un{\set{Z}}_1)} \\
        \vdots \\
        \bold \Delta_K\Ra{(\un{\set{Z}}_K)}
    \end{bmatrix}\\
    &= \bold Q(\set{Z}_0) + \mat{\alpha} \un{\vec{\Delta}}(\set{Z}_1^*,\set{Z}_1,\ldots,\set{Z}_K^*, \set{Z}_K) \Ra{= \bold Q(\set{Z}_0) + \mat{\alpha} \un{\vec{\Delta}}(\un{\set{Z}})}
\end{align}
\Ra{where $\vec{\Delta}_k\Ra{(\un{\set{Z}}_k)} \equiv \bold Q_k(\set{Z}_k^*) - \bold Q_k(\set{Z}_k)$} and we now have potentially $2K+1$ sample sets $\set{Z} = \{\set{Z}_0,\set{Z}_1^*,\set{Z}_1,\ldots\}$ \Ra{and $\un{\set{Z}}_k = \{\set{Z}_k^*, \set{Z}_k\}$}.
We have redefined $\un{\vec{\Delta}}\Ra{(\un{\set{Z}})} = [\vec{\Delta}_1\Ra{(\un{\set{Z}}_1)},\ldots, \vec{\Delta}_K\Ra{(\un{\set{Z}}_K)}]^T$ \Ra{as an $\reals^{DK}$-valued random vector where $\un{\set{Z}} = \{ \un{\set{Z}}_1,\ldots,\un{\set{Z}}_K \}$}.
If $\bold Q_i(\set{Z}_i^*)$ and $ \bold Q_i(\set{Z}_i)$ have the same expectation for all $i$, the resulting estimator has the same bias as $\mat{Q}(\set{Z}_0).$
 
 The expressions for the optimal weights $\mat{\alpha}^*$ and the variance $\mathbb{V}ar[\cv\Ra{(\set{Z})}]$ in Equations \eqref{eq:OptWei} and \eqref{eq:OptVar} still apply to the ACV estimator using the new definition of $\un{\vec{\Delta}}\Ra{(\un{\set{Z}})}$.

\section{Estimator Covariance Expressions}
\label{sec:Body}

In this section, we provide a collection of results for the covariance between several estimators that are needed for many multi-fidelity sampling strategies. These estimator covariances can then be used within multifidelity UQ sampling approaches for considering multiple outputs and/or for systems needing multiple statistics. Specifically for ACVs, the covariance expressions are needed for evaluating $\mathbb{C}ov[\bold Q\Ra{(\set{Z}_0)}, \un{\vec{\Delta}}\Ra{(\un{\set{Z}})}]$ and $\mathbb{V}ar[\un{\vec{\Delta}}\Ra{(\un{\set{Z}})}].$  Section \ref{sec:CVsetup} summarizes how to find $\var{\un{\vec{\Delta}}\Ra{(\un{\set{Z}})}}$ and $\cov{\vec{Q}\Ra{(\set{Z}_0)},\un{\vec{\Delta}}\Ra{(\un{\set{Z}})}}$ for any estimator and sets up the following sections. Section \ref{sec:meanandvar} introduces estimators for the mean and covariance of multi-fidelity vector-valued functions. Section \ref{sec:sensana} introduces estimators for the simultaneous estimation of variance and main effects in the context of scalar-valued functions.

\subsection{Setup and Summary}
\label{sec:CVsetup}
In this section, we describe the structure of the results that follow.
Since $\un{\vec{\Delta}}\Ra{(\un{\set{Z}})}$ is a vector of stacked estimators, the variance, $\mathbb{V}ar[\un{\vec{\Delta}}\Ra{(\un{\set{Z}})}]$, can be separated into a set of block covariance matrices:
\begin{align}
    \scalemath{0.9}{
    \mathbb{V}ar[\un{\vec{\Delta}}\Ra{(\un{\set{Z}})}] = \begin{bmatrix}
     \mathbb{V}ar[\bold \Delta_1\Ra{(\un{\set{Z}}_1)}] & \mathbb{C}ov[\bold \Delta_1\Ra{(\un{\set{Z}}_1)}, \bold \Delta_2\Ra{(\un{\set{Z}}_2)}] & \cdots & \mathbb{C}ov[\bold \Delta_1\Ra{(\un{\set{Z}}_1)},\bold \Delta_K\Ra{(\un{\set{Z}}_K)}] \\
    \mathbb{C}ov[\bold \Delta_2\Ra{(\un{\set{Z}}_2)}, \bold \Delta_1\Ra{(\un{\set{Z}}_1)}] & \mathbb{V}ar[\bold \Delta_2\Ra{(\un{\set{Z}}_2)}] & & \vdots \\
    \vdots &&\ddots\\
    \mathbb{C}ov[\bold \Delta_K\Ra{(\un{\set{Z}}_K)}, \bold \Delta_1\Ra{(\un{\set{Z}}_1)}] & \cdots & & \mathbb{V}ar[\bold \Delta_K\Ra{(\un{\set{Z}}_K)}]
    \end{bmatrix}.}
    \label{eq:vardelta}
\end{align}
Since $\mathbb{V}ar[\bold \Delta_i\Ra{(\un{\set{Z}}_i)}] = \mathbb{C}ov[\bold \Delta_i\Ra{(\un{\set{Z}}_i)}, \bold \Delta_i\Ra{(\un{\set{Z}}_i)}]$, we further decompose each covariance block into 
\begin{align}
\scalemath{0.8}{
    \mathbb{C}ov[\bold \Delta_i\Ra{(\un{\set{Z}}_i)}, \bold \Delta_j\Ra{(\un{\set{Z}}_j)}] 
    = \mathbb{C}ov[\bold Q_i(\set{Z}_i^*) ,\bold Q_j(\set{Z}_j^*)] -\mathbb{C}ov[\bold Q_i(\set{Z}_i^*) ,\bold Q_j(\set{Z}_j)] 
    - \mathbb{C}ov[\bold Q_i(\set{Z}_i) ,\bold Q_j(\set{Z}_j^*)] + \mathbb{C}ov[\bold Q_i(\set{Z}_i) ,\bold Q_j(\set{Z}_j)]},
    \label{eq:covdeltdelt}
\end{align}
Lastly, the covariance with the high fidelity estimator $\mathbb{C}ov[\bold Q\Ra{(\set{Z}_0)}, \un{\vec{\Delta}}\Ra{(\un{\set{Z}})}]$ is separated into
\begin{align}
    \mathbb{C}ov[\bold Q\Ra{(\set{Z}_0)}, \un{\vec{\Delta}}\Ra{(\un{\set{Z}})}] = \begin{bmatrix}
        \mathbb{C}ov[\bold Q\Ra{(\set{Z}_0)}, \bold \Delta_1\Ra{(\un{\set{Z}}_1)}] & \cdots & \mathbb{C}ov[\bold Q\Ra{(\set{Z}_0)}, \bold \Delta_K\Ra{(\un{\set{Z}}_K)}]
    \end{bmatrix},
\end{align}
where
\begin{align}    
    \mathbb{C}ov[\bold Q\Ra{(\set{Z}_0)}, \bold \Delta_i\Ra{(\un{\set{Z}}_i)}] = \mathbb{C}ov[\bold Q\Ra{(\set{Z}_0)}, \bold Q_i(\set{Z}_i^*)] - \mathbb{C}ov[\bold Q\Ra{(\set{Z}_0)}, \bold Q_i(\set{Z}_i)].
    \label{eq:covqdelt}
\end{align}

The subsequent sections derive expressions for the block components of these estimators, which can then be assembled into the final form.  A summary of the estimator settings we consider and references to the results is provided in Table \ref{table:1}. \Rb{We also show that these expressions simplify to previous results for a scalar function in MFMC estimation in Appendix \ref{sec:prevworks}.} For each case, the covariance $\cov{\vec{Q}_i(\set{N}),\vec{Q}_j(\set{M})}$ between the required estimators of two fidelities $\vec{Q}_i$ and $\vec{Q}_j$ is first computed for arbitrary input sets $\set{N}$ and $\set{M}$,  where $\set{N} = \{\vec{n}^{(s)} ; s=1,\ldots,N \}$ and $\set{M} = \{\vec{m}^{(s)} ; s=1,\ldots,M \}$ \Ra{such that $\vec{n}^{(s)}$ and $\vec{m}^{(s)}$ are $\reals^{I}$-valued random vectors}. Here, let $\set{P} = \set{N}\cap\set{M}$ be the intersection between two sets such that $P = |\set{N}\cap \set{M}|$ denotes the size of $\set{P}$. 

The computation of these components require certain statistics of the underlying multi-fidelity functions. To this end, subsequent sections begin with a highlighted box that describes what exactly is needed. In practice, these statistics can be available either analytically for some problems or must be obtained from pilot samples\Ra{, a set of independent samples used to estimate these statistics.} Later, we  numerically show that pilot samples are effective in Section~\ref{PilotSampleTradeoff}.

\addtolength{\tabcolsep}{-0.3em}

\begin{table}[h]
\scriptsize
\begin{center}
\caption{Proposition references to each of the introduced estimators.}
\begin{tabular}{ | c | c | c c | c c c |}
\hline
\multicolumn{4}{|c|}{MOACV Estimators} & \multicolumn{3}{c|}{Propositions} \\
\hline
\hline
  Estimators & Abbr. & Statistic & Model Output & 
  $\mathbb{C}ov[\bold Q_{i}\Ra{(\set{N})},\bold Q_{j}\Ra{(\set{M})}]$ & $\cov{\vec{\Delta}_i\Ra{(\un{\set{Z}}_i)}, \vec{\Delta}_j\Ra{(\un{\set{Z}}_j)}}$ & $\mathbb{C}ov[\bold Q\Ra{(\set{Z}_0)}, \vec{\Delta}_i\Ra{(\un{\set{Z}}_i)}]$ \\ 
  \hline
 Mean & M & Single & Multiple & \ref{meancov} & \ref{meanvar} & \ref{meanhighlow}\\  
 Variance & V & Single & Multiple  & \ref{var_cov} & \ref{Var_delt} & \ref{Var_high}\\
 Mean \& Variance & MV & Multiple & Multiple & \ref{MeanCov} & \ref{MeanCovVar} & \ref{MeanVarHighLow} \\
 Main Effect & ME & Multiple & Single& \ref{SobolEstCov2} & \ref{MultiSobolVar} & \ref{MultiSobolHighLow} \\
 ME \& Variance & MEV & Multiple & Single& \ref{VarSobolCov} & \ref{VarMEVarLow} & \ref{VarMEVarHighLow} \\
 \hline
\end{tabular}
\label{table:1}
\end{center}
\end{table}

\subsection{Mean and Variance Estimation} 
\label{sec:meanandvar}
In this section, we estimate the mean and variance of a vector-valued function \Ra{with a random input}. In Section \ref{sec:meanest} and Section \ref{VarEstSect} we separately estimate the means and covariance, respectively. Finally, in Section \ref{sec:jointMV} we simultaneously estimate the mean and covariance for vector-valued functions.

We further define notation for this section. Let $\Ra{\un{\vec{f}}}: \reals^I \to \reals^{D(K+1)}$, and $\Ra{\un{\vec{g}}}: \reals^I \to \reals^{D^2 (K+1)}$ be vector-valued functions collecting the outputs of a high-fidelity model and $K$ low-fidelity models according to  
\begin{align}
\label{eq:vecdefs}
    \Ra{\un{\vec{f}}}= \begin{bmatrix}
    \Ra{\vec{f}}_0 \\ \Ra{\vec{f}}_1 \\  
    \vdots \\ \Ra{\vec{f}}_K
    \end{bmatrix}
   \textrm{\quad and \quad } 
    \Ra{\un{\vec{g}}}= \begin{bmatrix}
    (\Ra{\vec{f}}_0- \ee{\Ra{\vec{f}}_0\Ra{(\vec{z})}})^{\otimes 2} \\
    (\Ra{\vec{f}}_1- \ee{\Ra{\vec{f}}_1\Ra{(\vec{z})}})^{\otimes 2} \\
    \vdots \\
    (\Ra{\vec{f}}_K- \ee{\Ra{\vec{f}}_K\Ra{(\vec{z})}})^{\otimes 2} 
    \end{bmatrix}.
\end{align} 

\subsubsection{Mean Estimator}
\label{sec:meanest}
We now consider mean estimation of a vector-valued function \Ra{with a random input}.

\begin{tcolorbox}[colback=blue!5,colframe=blue!55!black,title=Required Covariances for Mean Estimation]
The estimators in this section require these covariances
\begin{align}
    \un{\mat{A}} \equiv \var{\Ra{\un{\vec{f}}}\Ra{(\vec{z})}} \textrm{\quad where \quad} \mat{A}_{ij} = \cov{\Ra{\vec{f}}_i\Ra{(\vec{z})},\Ra{\vec{f}}_j\Ra{(\vec{z})}},
\end{align}
for $i,j\in \{ 0,1,\ldots,K \}.$
\end{tcolorbox}

\begin{prop}[Covariance between Mean Estimators]
\label{meancov}
    The covariance of two MC mean estimators~\eqref{eq:meanMC}, $\vec{Q}_i\Ra{(\set{N})}$ and $\vec{Q}_j\Ra{(\set{M})}$, corresponding to fidelities $i,j$ computed via input sets $\set{N}, \set{M}$, respectively, is $\mathbb{C}ov[\bold Q_i(\set{N}), \bold Q_j(\set{M})] = \frac{P}{NM}{\mat{A}}_{ij}.$ 
\end{prop}
\begin{proof}
   Using the definition of covariance, we obtain
\begin{align}
\scalemath{0.9}{
    \mathbb{C}ov[\bold Q_i(\set{N}), \bold Q_j(\set{M})] = \mathbb{C}ov\left[\frac{1}{N}\sum^N_{t=1} \Ra{\vec{f}}_i(\vec{n}^{(t)}), \frac{1}{M}\sum^M_{s=1} \Ra{\vec{f}}_j(\vec{m}^{(s)})\right] = \frac{1}{NM}\sum_{t=1}^N\sum_{s=1}^M\mathbb{C}ov\left[\Ra{\vec{f}}_i(\vec{n}^{(t)}),\Ra{\vec{f}}_j(\vec{m}^{(s)})\right]}. \nonumber
\end{align}
The function outputs are only correlated if the \Ra{input random variables} are the same. Thus, each covariance term is only nonzero if $\vec{n}^{(t)} = \vec{m}^{(s)}$.  The only nonzero covariance terms are due to samples in $\set{P}$. Thus, there are $P$ nonzero covariance terms, and the stated result follows
\begin{align}
    \mathbb{C}ov[\bold Q_i(\set{N}), \bold Q_j(\set{M})] &= \frac{|\set{N}\cap\set{M}|}{NM}\mathbb{C}ov\left[\Ra{\vec{f}}_i\Ra{(\vec{z})},\Ra{\vec{f}}_j\Ra{(\vec{z})}\right]
    = \frac{P}{NM}\mathbb{C}ov\left[\Ra{\vec{f}}_i\Ra{(\vec{z})},\Ra{\vec{f}}_j\Ra{(\vec{z})}\right].
\end{align}
\end{proof}
Using this result, we obtain the covariance between the discrepancies as follows\footnote{In this result, and those that follow, when $i=j$, the equation simplifies greatly. Here it becomes $F_{ii} = \frac{1}{\Nsing{i}{*}} + \frac{1}{\Nsing{i}{}} - 2\frac{\Ncomb{i}{}{i}{*}}{\Nsing{i}{}\Nsing{i}{*}}$. All other matrices defined similarly have a reduced form along the diagonals.}.
\begin{prop}[Variance of discrepancies for M]
\label{meanvar}
The covariance between discrepancies is
$   \cov{\vec{\Delta}_i\Ra{(\un{\set{Z}}_i)}, \vec{\Delta}_j\Ra{(\un{\set{Z}}_j)}} = F_{ij} {\mat{A}}_{ij}$
where
\begin{align}
\label{MeanF}
    F_{ij} = \frac{\Ncomb{i}{*}{j}{*}}{\Nsing{i}{*}\Nsing{j}{*}}-\frac{\Ncomb{i}{*}{j}{}}{\Nsing{i}{*}\Nsing{j}{}}-\frac{\Ncomb{i}{}{j}{*}}{\Nsing{i}{}\Nsing{j}{*}}+\frac{\Ncomb{i}{}{j}{}}{\Nsing{i}{}\Nsing{j}{}},
\end{align} 
for $i,j = 1,\ldots,K,$. 
\end{prop}
\begin{proof}
The result follows a straightforward calculation
\begin{align}
    \mathbb{C}ov[\vec{\Delta}_i\Ra{(\un{\set{Z}}_i)}, \vec{\Delta}_j\Ra{(\un{\set{Z}}_j)}] & = \mathbb{C}ov[\bold Q_i(\set{Z}_i^*) ,\bold Q_j(\set{Z}_j^*)] -\mathbb{C}ov[\bold Q_i(\set{Z}_i^*) ,\bold Q_j(\set{Z}_j)] \nonumber\\
    &\quad -\mathbb{C}ov[\bold Q_i(\set{Z}_i) ,\bold Q_j(\set{Z}_j^*)] + \mathbb{C}ov[\bold Q_i(\set{Z}_i) ,\bold Q_j(\set{Z}_j)] \nonumber \\
    &= \left[\frac{\Ncomb{i}{*}{j}{*}}{\Nsing{i}{*}\Nsing{j}{*}}-\frac{\Ncomb{i}{*}{j}{}}{\Nsing{i}{*}\Nsing{j}{}}-\frac{\Ncomb{i}{}{j}{*}}{\Nsing{i}{}\Nsing{j}{*}}+\frac{\Ncomb{i}{}{j}{}}{\Nsing{i}{}\Nsing{j}{}}\right]\mathbb{C}ov[\Ra{\vec{f}}_i\Ra{(\vec{z})},\Ra{\vec{f}}_j\Ra{(\vec{z})}].
\end{align}
\end{proof}
Note that when $D=1$, Proposition \ref{meanvar} is equivalent to~\cite[Eq. 13]{Param}. Finally, the covariance between the high-fidelity and discrepancy estimators is provided. A similar argument yields the following result.
\begin{prop}[Variance between high-fidelity and discrepancies for M]
\label{meanhighlow}
    The covariance between the high-fidelity and discrepancy estimator is $
    \mathbb{C}ov[\bold Q\Ra{(\set{Z}_0)}, \vec{\Delta}_i\Ra{(\un{\set{Z}}_i)}] = G_i {\mat{A}}_{0i}
    $
where
\begin{align}
\label{MeanG}
    G_i = \frac{\Ncomb{0}{}{i}{*}}{\Nsing{0}{}\Nsing{i}{*}} - \frac{\Ncomb{0}{}{i}{}}{\Nsing{0}{}\Nsing{i}{}}.
\end{align}
\end{prop}

\subsubsection{Variance Estimator}
\label{VarEstSect}
We now estimate the variance of a vector-valued function \Ra{with a random input}. The proofs are provided in the Appendix for brevity.

\begin{tcolorbox}[colback=blue!5,colframe=blue!55!black,title=Required Covariances for Variance Estimation]
The estimators in this section require these covariances
\begin{align}
    \un{\mat{V}} &= \begin{bmatrix}
        \mat{V}_{00} & \mat{V}_{01} & \cdots & \mat{V}_{0K} \\
        \mat{V}_{10} & \mat{V}_{11} & & \vdots \\
        \vdots & & \ddots & \\
        \mat{V}_{K0} & \cdots & & \mat{V}_{KK}
    \end{bmatrix} \in \reals^{(K+1)D^2 \times (K+1)D^2} \\
    \un{\mat{W}} &= \var{\Ra{\un{\vec{g}}}\Ra{(\vec{z})}} \in \reals^{(K+1)D^2 \times (K+1)D^2},
    \label{eq:W}
    \end{align} 
where $\Ra{\un{\vec{g}}}$ can be seen in Equation \eqref{eq:vecdefs}. The elements of $\un{\mat{V}}$ are
\begin{align}
\scalemath{0.9}{
    {\mat{V}}_{ij} = \mathbb{C}ov[\Ra{\vec{f}}_i\Ra{(\vec{z})}, \Ra{\vec{f}}_j\Ra{(\vec{z})}]^{\otimes 2} + \left( \bold 1_D^T \otimes \mathbb{C}ov[\Ra{\vec{f}}_i\Ra{(\vec{z})},\Ra{\vec{f}}_j\Ra{(\vec{z})}] \otimes \bold 1_D \right) \circ \left(\bold 1_D \otimes \mathbb{C}ov[\Ra{\vec{f}}_i\Ra{(\vec{z})},\Ra{\vec{f}}_j\Ra{(\vec{z})}] \otimes \bold 1_D^T \right),}
\end{align}
where $\mat{V}_{ij} \in \reals^{D^2 \times D^2}$. Elements of $\un{\mat{W}}$ are $\scalemath{0.81}{{\mat{W}}_{ij} = \cov{\left(\Ra{\vec{f}}_i\Ra{(\vec{z})}-\mathbb{E}[\Ra{\vec{f}}_i\Ra{(\vec{z})}]\right)^{\otimes 2},\left(\Ra{\vec{f}}_j\Ra{(\vec{z})}-\mathbb{E}[\Ra{\vec{f}}_j\Ra{(\vec{z})}]\right)^{\otimes 2}}.}$
\end{tcolorbox}

\begin{prop}[Covariance between Variance Estimators]
\label{var_cov}
    The covariance between two MC variance estimators \eqref{eq:varianceMC}, $\vec{Q}_i\Ra{(\set{N})}$ and $\vec{Q}_j\Ra{(\set{M})}$, corresponding to fidelities $i,j$ computed via input sets $\set{N}, \set{M}$, respectively, is
    \begin{align}
    \label{variancecovariance}
    \mathbb{C}ov[\bold Q_i(\set{N}), \bold Q_j(\set{M})] = \frac{P(P-1)}{N(N-1)M(M-1)} \mat{V}_{ij} + \frac{P}{NM}\mat{W}_{ij}.
\end{align}
\end{prop}
Using this result, we obtain the covariance between the discrepancies as follows.

\begin{prop}[Variance of discrepancies for V]
\label{Var_delt}
Let $F_{ij}$ be the same as in Equation~\eqref{MeanF}. The covariance between discrepancies is $\cov{\vec{\Delta}_i\Ra{(\un{\set{Z}}_i)},\vec{\Delta}_j\Ra{(\un{\set{Z}}_j)}} = F_{ij} {\mat{W}}_{ij} + H_{ij} {\mat{V}}_{ij}$
where  
\begin{align}
    H_{ij} &= 
         \frac{\Ncomb{i}{*}{j}{*}(\Ncomb{i}{*}{j}{*}-1)}{\Nsing{i}{*}(\Nsing{i}{*}-1)\Nsing{j}{*}(\Nsing{j}{*}-1)} - \frac{\Ncomb{i}{*}{j}{}(\Ncomb{i}{*}{j}{}-1)}{\Nsing{i}{*}(\Nsing{i}{*}-1)\Nsing{j}{}(\Nsing{j}{}-1)} \nonumber\\
         &\quad- \frac{\Ncomb{i}{}{j}{*}(\Ncomb{i}{}{j}{*}-1)}{\Nsing{i}{}(\Nsing{i}{}-1)\Nsing{j}{*}(\Nsing{j}{*}-1)}  + \frac{\Ncomb{i}{}{j}{}(\Ncomb{i}{}{j}{}-1)}{\Nsing{i}{}(\Nsing{i}{}-1)\Nsing{j}{}(\Nsing{j}{}-1)}.
\end{align}
\end{prop}
Finally, the covariance between the high-fidelity and discrepancy estimators is provided.

\begin{prop}[Variance between high-fidelity and discrepancies for V]
\label{Var_high}
Let $G_{i}$ be the same as in Equation~\eqref{MeanG}.
    The covariance between the high-fidelity and discrepancy estimator is $\mathbb{C}ov[\bold Q\Ra{(\set{Z}_0)}, \vec{ \Delta}_i\Ra{(\un{\set{Z}}_i)}] = J_{i} {\mat{V}}_{0i} + G_{i} {\mat{W}} _{0i}$
    where
    \begin{align}
        J_i = \frac{\Ncomb{0}{}{i}{*}(\Ncomb{0}{}{i}{*}-1)}{\Nsing{0}{}(\Nsing{0}{}-1)\Nsing{i}{*}(\Nsing{i}{*}-1)}- \frac{\Ncomb{0}{}{i}{}(\Ncomb{0}{}{i}{}-1)}{\Nsing{0}{}(\Nsing{0}{}-1)\Nsing{i}{}(\Nsing{i}{}-1)}.
    \end{align}
\end{prop}

\subsubsection{Mean and Variance Estimators}
\label{sec:jointMV}
We now consider a combined estimator, simultaneously providing a mean and variance (MV) estimate.

\begin{tcolorbox}[colback=blue!5,colframe=blue!55!black,title=Required Covariances for Mean and Variance Estimation]
In addition to the covariances from Sections~\ref{sec:meanest} and~\ref{VarEstSect}, the covariance \\
$\un{\mat{B}} = \mathbb{C}ov[\Ra{\un{\vec{f}}(\vec{z})},\Ra{\un{\vec{g}}(\vec{z})}] \in \reals^{D(K+1)\times D^2(K+1)}$ is required such that \\
${\mat{B}}_{ij} = \mathbb{C}ov[\Ra{\vec{f}}_i\Ra{(\vec{z})}, (\Ra{\vec{f}}_j\Ra{(\vec{z})} - \mathbb{E}[\Ra{\vec{f}}_j\Ra{(\vec{z})}])^{\otimes 2}]$.
\end{tcolorbox}

The stacked MC mean and variance estimator is
\begin{align}
    \bold Q_{i}(\set{N}) = \begin{bmatrix}
    \vec{Q}_{\mu, i}(\set{N})  \\
        \vec{Q}_{V, i}(\set{N}) 
    \end{bmatrix} =\begin{bmatrix}
    \frac{1}{N}\sum^N_{t=1} \Ra{\vec{f}}_i(\vec{n}^{(t)}) \\
        \frac{1}{2N(N-1)}\sum^N_{s=1}\sum^N_{t=1} \left( \Ra{\vec{f}}_i(\vec{n}^{(s)})-\Ra{\vec{f}}_i(\vec{n}^{(t)}) \right)^{\otimes 2}
    \end{bmatrix}
    \label{eq:stacked}
\end{align}
\Ra{where $\bold Q_{i}(\set{N})$ is a $\reals^{D + D^2}$-valued random vector.}
\begin{prop}[Covariance between Mean and Variance Estimators]
\label{MeanCov}
The covariance between two stacked MC estimators \eqref{eq:stacked}, $\vec{Q}_i\Ra{(\set{N})}$ and $\vec{Q}_j\Ra{(\set{M})}$, corresponding to fidelities $i,j$ computed via input sets $\set{N}, \set{M}$, respectively, is
\begin{align}
    \mathbb{C}ov[\bold Q_{i}(\set{N}), \bold Q_{j}(\set{M})] = \begin{bmatrix}
    \mathbb{C}ov[\vec{Q}_{\mu, i}(\set{N}), \vec{Q}_{\mu, j}(\set{M})] & \mathbb{C}ov[\vec{Q}_{\mu, i}(\set{N}), \vec{Q}_{V, j}(\set{M})]\\
    \mathbb{C}ov[\vec{Q}_{V, i}(\set{N}), \vec{Q}_{\mu, j}(\set{M})]  & \mathbb{C}ov[\vec{Q}_{V, i}(\set{N}), \vec{Q}_{V, j}(\set{M})],
    \end{bmatrix}
\end{align}
where the diagonal terms were found in Propositions \ref{meancov} and \ref{var_cov}. The covariance between the mean and variance estimator is
$\mathbb{C}ov[\vec{Q}_{\mu, i}(\set{N}), \vec{Q}_{V, j}(\set{M})] = \frac{P}{NM} {\mat{B}}_{ij}.$
\end{prop}
Using this result, we obtain the covariance between the discrepancies as follows.
\begin{prop}[Variance of discrepancies for MV]
\label{MeanCovVar}
    The variance of the discrepancies is
    \begin{align}
        \mathbb{C}ov[\bold \Delta_i\Ra{(\un{\set{Z}}_i)}, \bold \Delta_j\Ra{(\un{\set{Z}}_j)}] = \begin{bmatrix}
            \mathbb{C}ov[\bold \Delta_{\mu,i}\Ra{(\un{\set{Z}}_i)}, \bold \Delta_{\mu,j}\Ra{(\un{\set{Z}}_j)}]& \mathbb{C}ov[\bold \Delta_{\mu,i}\Ra{(\un{\set{Z}}_i)}, \bold \Delta_{V,j}\Ra{(\un{\set{Z}}_j)}] \\
            \mathbb{C}ov[\bold \Delta_{V,i}\Ra{(\un{\set{Z}}_i)}, \bold \Delta_{\mu,j}\Ra{(\un{\set{Z}}_j)}] & \mathbb{C}ov[\bold \Delta_{V,i}\Ra{(\un{\set{Z}}_i)}, \bold \Delta_{V,j}\Ra{(\un{\set{Z}}_j)}]
        \end{bmatrix},
    \end{align}
    such that $\mathbb{C}ov[\bold \Delta_{\mu,i}\Ra{(\un{\set{Z}}_i)}, \bold \Delta_{V,j}\Ra{(\un{\set{Z}}_j)}] = F_{ij}  {\mat{B}}_{ij}$
    where $F_{ij}$ is from Equation \eqref{MeanF} and \\
    $\mathbb{C}ov[\bold \Delta_{\mu,i}\Ra{(\un{\set{Z}}_i)}, \bold \Delta_{\mu,j}\Ra{(\un{\set{Z}}_j)}]$ and $\mathbb{C}ov[\bold \Delta_{V,i}\Ra{(\un{\set{Z}}_i)}, \bold \Delta_{V,j}\Ra{(\un{\set{Z}}_j)}]$ can be found in Propositions \ref{meanvar} and \ref{Var_delt} respectively.
\end{prop}
Finally, the covariance between the high-fidelity and discrepancy estimators is provided.
\begin{prop}[Variance between high-fidelity and discrepancies for MV]
\label{MeanVarHighLow}
    The covariance between the high-fidelity and discrepancy estimators is
    \begin{align}
        \mathbb{C}ov[\bold Q\Ra{(\set{Z}_0)}, \bold \Delta_i\Ra{(\un{\set{Z}}_i)}] = \begin{bmatrix}
            \mathbb{C}ov[\bold Q_\mu\Ra{(\set{Z}_0)}, \bold \Delta_{\mu,i}\Ra{(\un{\set{Z}}_i)}] &  \mathbb{C}ov[\bold Q_\mu\Ra{(\set{Z}_0)}, \bold \Delta_{V,i}\Ra{(\un{\set{Z}}_i)}] \\
             \mathbb{C}ov[\bold Q_V\Ra{(\set{Z}_0)}, \bold \Delta_{\mu,i}\Ra{(\un{\set{Z}}_i)}] &  \mathbb{C}ov[\bold Q_V\Ra{(\set{Z}_0)}, \bold \Delta_{V,i}\Ra{(\un{\set{Z}}_i)}]
        \end{bmatrix},
    \end{align}
    such that $\mathbb{C}ov[\bold Q_\mu\Ra{(\set{Z}_0)}, \bold \Delta_{V,i}\Ra{(\un{\set{Z}}_i)}] = G_i {\mat{B}}_{0i}$ and $\mathbb{C}ov[\bold Q_V\Ra{(\set{Z}_0)}, \bold \Delta_{\mu,i}\Ra{(\un{\set{Z}}_i)}] = G_i \{{\un{\mat{B}}}^T\}_{0i} $
    where $G_i$ is from Equation \eqref{MeanG} and $\mathbb{C}ov[\bold Q_\mu\Ra{(\set{Z}_0)}, \bold \Delta_{\mu,i}\Ra{(\un{\set{Z}}_i)}]$ and $\mathbb{C}ov[\bold Q_V\Ra{(\set{Z}_0)}, \bold \Delta_{V,i}\Ra{(\un{\set{Z}}_i)}]$ can be found in Propositions \ref{meanhighlow} and \ref{Var_high} respectively.
\end{prop}

\subsection{Sensitivity Analysis}
\label{sec:sensana}
In this section, we estimate the covariances required for main effect (ME) Sobol indices of a scalar function. \Ra{Unlike previous sections, the following results are only available for a single function output, $f: \reals^I \to \reals.$} In Section \ref{MultipleME}, multiple ME variances are estimated simultaneously. In Section \ref{VarMEVarSection}, the variance and multiple ME variances are estimated simultaneously.

For notation in this section, let $\Ra{\un{\vec{h}}(\vec{x})} = \Ra{\un{\vec{f}}}\Ra{(\vec{x})} \otimes \bold 1_I$ \Ra{be a $\mathbb{R}^{I(K+1)}$-valued random vector} and 
\begin{align}
   \Ra{\bar{\vec{h}}(\vec{y})} = \begin{bmatrix}
        f_{0}\Ra{(\vec{y}_1)} & f_{0}\Ra{(\vec{y}_2)} & \cdots & f_{0}\Ra{(\vec{y}_I)} & f_{1}\Ra{(\vec{y}_1)} & \cdots f_{K}\Ra{(\vec{y}_I)}
    \end{bmatrix}^T,
\end{align}
\Ra{where $\bar{\vec{h}}\Ra{(\vec{y})}$ is an $\reals^{I(K+1)}$-valued random vector. Note that $\vec{x}$ and $\vec{y}$ are independent, but $\vec{y}_u$ shares the $u$-th input element of $\vec{x}$ as described in Section \ref{sec:MCests}.}

\begin{tcolorbox}[colback=blue!5,colframe=blue!55!black,title=Required Covariances for Variance and Main Effect Variances Estimation]
The estimators in this section require these covariances
\begin{align}
    \mathbb{V}ar\begin{bmatrix}
        \Ra{\un{\vec{h}}(\vec{x})}\circ\Ra{\bar{\vec{h}}(\vec{y})} - \Ra{\un{\vec{h}}(\vec{x})}^{\circ 2} \\
        \Ra{\un{\vec{h}}(\vec{x})} \circ \Ra{\bar{\vec{h}}(\vec{y})} - 2 \Ra{\un{\vec{h}}(\vec{x})} \circ \mathbb{E}[\Ra{\un{\vec{h}}(\vec{x})}] \\
        \Ra{\un{\vec{g}}}
    \end{bmatrix} &= 
    \begin{bmatrix}
        \un{\mat{O}} & \un{\mat{S}} & \un{\mat{C}} \\
        & \un{\mat{R}} & \un{\mat{E}} \\
        \textrm{Sym.} & & \un{\mat{W}}
    \end{bmatrix} \in \reals^{(2I+1)(K+1)\times (2I+1)(K+1)} \\
    \un{\mat{U}} &= \var{\Ra{\un{{\vec{h}}}(\vec{x})}}^{\circ 2} \in \reals^{I(K+1)\times I(K+1)},
\end{align}
where $\un{\mat{W}}$ is from Eq. \eqref{eq:W}, $\un{\mat{O}}, \un{\mat{S}}, \un{\mat{R}} \in \reals^{I(K+1)\times I(K+1)}$, and $\un{\mat{C}}, \un{\mat{E}} \in \reals^{I(K+1)\times (K+1)}$. \Rb{More specifically,
\begin{align}
    \un{\mat{O}} &= \var{\un{\vec{h}}(\vec{x})\circ\bar{\vec{h}}(\vec{y}) - \un{\vec{h}}(\vec{x})^{\circ 2}} \\
    \un{\mat{R}} &= \var{\un{\vec{h}}(\vec{x}) \circ \bar{\vec{h}}(\vec{y}) - 2 \un{\vec{h}}(\vec{x}) \circ \mathbb{E}[\un{\vec{h}}(\vec{x})]} \\
    \un{\mat{S}} &= \cov{\un{\vec{h}}(\vec{x})\circ\bar{\vec{h}}(\vec{y}) - \un{\vec{h}}(\vec{x})^{\circ 2},\un{\vec{h}}(\vec{x}) \circ \bar{\vec{h}}(\vec{y}) - 2 \un{\vec{h}}(\vec{x}) \circ \mathbb{E}[\un{\vec{h}}(\vec{x})]}\\
    \un{\mat{C}} &= \cov{\un{\vec{h}}(\vec{x})\circ\bar{\vec{h}}(\vec{y}) - \un{\vec{h}}(\vec{x})^{\circ 2},\un{\vec{g}}}\\
    \un{\mat{E}} &= \cov{\un{\vec{h}}(\vec{x}) \circ \bar{\vec{h}}(\vec{y}) - 2 \un{\vec{h}}(\vec{x}) \circ \mathbb{E}[\un{\vec{h}}(\vec{x})],\un{\vec{g}}}.
\end{align}
}
\end{tcolorbox}

\subsubsection{Main Effect Variance Estimators}
\label{MultipleME}
We now estimate multiple ME variances of a scalar function. The combined ME variance estimator is
\begin{align}
    \bold Q_{i}(\set{N}) &= 
    \begin{bmatrix}
        \bold Q_{ i,1}(\set{N}) \\
        \bold Q_{ i,2}(\set{N}) \\
        \vdots \\
        \bold Q_{ i, I}(\set{N})
    \end{bmatrix} = \begin{bmatrix}
        \frac{1}{N^2}\sum_{j=1}^N\sum_{k=1}^N \left[f_i(\Ra{\vec{x}}^{(j)})f_i(\vec{y}^{(j)}_{1}) - f_i(\Ra{\vec{x}}^{(j)})f_i(\Ra{\vec{x}}^{(k)}) \right] \\
        \frac{1}{N^2}\sum_{j=1}^N\sum_{k=1}^N \left[f_i(\Ra{\vec{x}}^{(j)})f_i(\vec{y}^{(j)}_{2}) - f_i(\Ra{\vec{x}}^{(j)})f_i(\Ra{\vec{x}}^{(k)}) \right] \\
        \vdots \\
        \frac{1}{N^2}\sum_{j=1}^N\sum_{k=1}^N \left[f_i(\Ra{\vec{x}}^{(j)})f_i(\vec{y}^{(j)}_{I}) - f_i(\Ra{\vec{x}}^{(j)})f_i(\Ra{\vec{x}}^{(k)}) \right]
    \end{bmatrix}
    \label{eq:stackedSobol}
\end{align}
\Ra{where $\bold Q_{i}(\set{N})$ is a $\reals^{I}$-valued random vector and $\vec{Q}_{i,u}(\set{N})$ is the main effect variance estimator of the $i$-th fidelity on the $u$-th input element.}

\begin{prop}[Covariance between Main Effect Variance Estimators]
\label{SobolEstCov2}
    The covariance between two stacked MC estimators \eqref{eq:stackedSobol}, $\vec{Q}_i\Ra{(\set{N})}$ and $\vec{Q}_j\Ra{(\set{M})}$, corresponding to fidelities $i,j$ computed via input sets $\set{N}, \set{M}$, respectively, is
    \begin{align}
    \label{MultipleMEcovariance}
    \mathbb{C}ov[\bold Q_{i}(\set{N}),\bold Q_{j}(\set{M})] &= \frac{P}{M^2N^2}\left[ (N-1)(M-1) {\mat{R}}_{ij} \right. + (N-1)  \{{\un{\mat{S}}}^T\}_{ij} \notag\\
    &\quad + (M-1) {\mat{S}}_{ij}  + {\mat{O}}_{ij} + 2(P-1) {\mat{U}}_{ij} \left. \right].
    \end{align}
\end{prop}
Using this result, we obtain the covariance between the discrepancies as follows.

\begin{prop}[Var. of discrepancies for ME]
\label{MultiSobolVar}
    The covariance between discrepancies is 
    $\cov{\vec{\Delta}_{i}\Ra{(\un{\set{Z}}_i)},\vec{\Delta}_{j}\Ra{(\un{\set{Z}}_j)}} = F_{ij} {\mat{O}}_{ij} + G_{ij} {\mat{R}}_{ij} + H_{ij} \{{\un{\mat{S}}}^T\}_{ij} + H_{ji} {\mat{S}}_{ij} + J_{ij} {\mat{U}}_{ij}$
    where  
    \begin{align}
    \label{MEF}
        F_{ij} &= 
             \frac{\Ncomb{i}{*}{j}{*}}{\Nsing{i}{*}^2\Nsing{j}{*}^2} - \frac{\Ncomb{i}{*}{j}{}}{\Nsing{i}{*}^2\Nsing{j}{}^2} - \frac{\Ncomb{i}{}{j}{*}}{\Nsing{i}{}^2\Nsing{j}{*}^2} + \frac{\Ncomb{i}{}{j}{}}{\Nsing{i}{}^2\Nsing{j}{}^2} \\
        G_{ij} &= 
            \frac{\Ncomb{i}{*}{j}{*}(\Nsing{i}{*}-1)(\Nsing{j}{*}-1)}{\Nsing{i}{*}^2 \Nsing{j}{*}^2} -
            \frac{\Ncomb{i}{*}{j}{}(\Nsing{i}{*}-1)(\Nsing{j}{}-1)}{\Nsing{i}{*}^2 \Nsing{j}{}^2} \nonumber\\
            &\quad -
            \frac{\Ncomb{i}{}{j}{*}(\Nsing{i}{}-1)(\Nsing{j}{*}-1)}{\Nsing{i}{}^2 \Nsing{j}{*}^2} + 
            \frac{\Ncomb{i}{}{j}{}(\Nsing{i}{}-1)(\Nsing{j}{}-1)}{\Nsing{i}{}^2 \Nsing{j}{}^2}\\
        H_{ij} &= 
            \frac{\Ncomb{i}{*}{j}{*}(\Nsing{i}{*}-1)}{\Nsing{i}{*}^2 \Nsing{j}{*}^2} - \frac{\Ncomb{i}{*}{j}{}(\Nsing{i}{*}-1)}{\Nsing{i}{*}^2 \Nsing{j}{}^2} \nonumber\\
            &\quad - \frac{\Ncomb{i}{}{j}{*}(\Nsing{i}{}-1)}{\Nsing{i}{}^2 \Nsing{j}{*}^2} + \frac{\Ncomb{i}{}{j}{}(\Nsing{i}{}-1)}{\Nsing{i}{}^2 \Nsing{j}{}^2} \\
        \label{MEJ}
        J_{ij} &=
            2\frac{\Ncomb{i}{*}{j}{*}(\Ncomb{i}{*}{j}{*}-1)}{\Nsing{i}{*}^2 \Nsing{j}{*}^2} - 2\frac{\Ncomb{i}{*}{j}{}(\Ncomb{i}{*}{j}{}-1)}{\Nsing{i}{*}^2 \Nsing{j}{}^2} \nonumber\\
            &\quad - 2\frac{\Ncomb{i}{}{j}{*}(\Ncomb{i}{}{j}{*}-1)}{\Nsing{i}{}^2 \Nsing{j}{*}^2} + 2\frac{\Ncomb{i}{}{j}{}(\Ncomb{i}{}{j}{}-1)}{\Nsing{i}{}^2 \Nsing{j}{}^2}.
        \end{align}
\end{prop}
Finally, the covariance between the high-fidelity and discrepancy estimators is provided.

\begin{prop}[Variance between high-fidelity and discrepancies for ME]
\label{MultiSobolHighLow}
    The covariance between the high-fidelity and discrepancy estimator is\\
    $\mathbb{C}ov[\bold Q\Ra{({\set{Z}}_0)}, \vec{\Delta}_{i}\Ra{(\un{\set{Z}}_i)}] = V_i {\mat{O}}_{0i} + W_i {\mat{R}}_{0i} + X_{i0} \{{\un{\mat{S}}}^T\}_{0i} + X_{0i} {\mat{S}}_{0i} + Z_i {\mat{U}}_{0i}$ where \Ra{$\set{Z}_0^* \equiv \set{Z}_0$} and 
    \begin{align}
    \label{MEV}
        V_i &= \frac{\Ncomb{0}{}{i}{*}}{\Nsing{0}{}^2 \Nsing{i}{*}^2} - \frac{\Ncomb{0}{}{i}{}}{\Nsing{0}{}^2 \Nsing{i}{}^2}\\ 
        W_i &= \frac{\Ncomb{0}{}{i}{*}(\Nsing{i}{*}-1)(\Nsing{0}{} - 1)}{\Nsing{0}{}^2 \Nsing{i}{*}^2} - \frac{\Ncomb{0}{}{i}{}(\Nsing{i}{}-1)(\Nsing{0}{} - 1)}{\Nsing{0}{}^2 \Nsing{i}{}^2} \\
        X_{ij} &= \frac{\Ncomb{i}{*}{j}{*}(\Nsing{j}{*}-1)}{\Nsing{i}{*}^2 \Nsing{j}{*}^2} - \frac{\Ncomb{i}{}{j}{}(\Nsing{j}{}-1)}{\Nsing{i}{}^2 \Nsing{j}{}^2} \\
        \label{MEZ}
        Z_i &= 2\frac{\Ncomb{0}{}{i}{*}(\Ncomb{0}{}{i}{*}-1)}{\Nsing{0}{}^2 \Nsing{i}{*}^2} - 2\frac{\Ncomb{0}{}{i}{}(\Ncomb{0}{}{i}{}-1)}{\Nsing{0}{}^2 \Nsing{i}{}^2}.
    \end{align}
\end{prop}

\subsubsection{Variance and Main Effect Variance Estimator}
\label{VarMEVarSection}
We now estimate the variance and multiple ME variances (MEV estimator) of a scalar function. The stacked variance and ME variance estimator is
\begin{align}
    \bold Q_i(\set{N}) &=\begin{bmatrix}
        \bold Q_{ i,1}(\set{N}) \\
        \bold Q_{ i,2}(\set{N}) \\
        \vdots \\
        \bold Q_{ i,I}(\set{N}) \\
        \bold Q_{V, i}(\set{N})
    \end{bmatrix} = \begin{bmatrix}
        \frac{1}{N^2}\sum_{j=1}^N\sum_{k=1}^N \left[f_i(\Ra{\vec{x}}^{(j)})f_i(\vec{y}^{(j)}_{1}) - f_i(\Ra{\vec{x}}^{(j)})f_i(\Ra{\vec{x}}^{(k)}) \right] \\
        \frac{1}{N^2}\sum_{j=1}^N\sum_{k=1}^N \left[f_i(\Ra{\vec{x}}^{(j)})f_i(\vec{y}^{(j)}_{2}) - f_i(\Ra{\vec{x}}^{(j)})f_i(\Ra{\vec{x}}^{(k)}) \right] \\
        \vdots \\
        \frac{1}{N^2}\sum_{j=1}^N\sum_{k=1}^N \left[f_i(\Ra{\vec{x}}^{(j)})f_i(\vec{y}^{(j)}_{I}) - f_i(\Ra{\vec{x}}^{(j)})f_i(\Ra{\vec{x}}^{(k)}) \right]\\
        \frac{1}{2N(N-1)}\sum_{j=1}^N\sum_{k=1}^N \left( f_i(\Ra{\vec{x}}^{(j)})-f_i(\Ra{\vec{x}}^{(k)}) \right)^{2}
        \end{bmatrix}
        \label{eq:stackedMEV}
\end{align}
\Ra{where $\bold Q_{i}(\set{N})$ is a $\reals^{I+1}$-valued random vector.} 

\begin{prop}[Covariance between Variance and Main Effect Variance Estimators]
\label{VarSobolCov}
The covariance between two stacked MC estimators \eqref{eq:stackedMEV}, $\vec{Q}_i\Ra{(\set{N})}$ and $\vec{Q}_j\Ra{(\set{M})}$, corresponding to fidelities $i,j$ computed via input sets $\set{N}, \set{M}$, respectively, is
\begin{align}
    \mathbb{C}ov[\bold Q_i(\set{N}), \bold Q_j(\set{M})] = \begin{bmatrix}
        \mathbb{C}ov[\bold Q_{i,\vec{x}}(\set{N}), \bold Q_{j,\vec{x}}(\set{M})] & \mathbb{C}ov[\bold Q_{ i,\vec{x}}(\set{N}), \vec{Q}_{V, j}(\set{M})] \\
        \mathbb{C}ov[\vec{Q}_{V, i}(\set{N}), \vec{Q}_{j,\vec{x}}(\set{M})] & \mathbb{C}ov[\vec{Q}_{V, i}(\set{N}), \vec{Q}_{V, j}(\set{M}) \\
    \end{bmatrix},
\end{align}
where $\bold Q_{i,\vec{x}}(\set{N}) = [\mat{Q}_{i,1}(\set{N}) \ldots \mat{Q}_{i,I}(\set{N})]^T.$ The diagonal terms of this block-matrix can be found in Propositions \ref{var_cov} and \ref{SobolEstCov2}. Now, the covariance between the ME variance estimator and the variance estimator is
    \begin{align}
        \mathbb{C}ov[\bold Q_{i,\vec{x}}(\set{N}), \bold Q_{V,j}(\set{M})] &=
        \frac{P(N-1)}{MN^2} {\vec{E}}_{ij} + \frac{P}{MN^2} {\vec{C}}_{ij} + \frac{2P(P-1)}{M(M-1)N^2} {\vec{U}}_{ij,0},
    \end{align}
    where ${\vec{U}}_{ij,0} \in \reals^I$ is the first column of $\vec{U}_{ij}$.
\end{prop}
Using this result, we obtain the covariance between the discrepancies as follows.

\begin{prop}[Variance of discrepancies for MEV]
\label{VarMEVarLow}
    The variance of the discrepancies is
    \begin{align}
        \mathbb{C}ov[\bold \Delta_i\Ra{(\un{\set{Z}}_i)}, \bold \Delta_j\Ra{(\un{\set{Z}}_j)}] = \begin{bmatrix}
            \mathbb{C}ov[\bold \Delta_{i,\vec{x}}\Ra{(\un{\set{Z}}_i)}, \bold \Delta_{j,\vec{x}}\Ra{(\un{\set{Z}}_j)}]& \mathbb{C}ov[\bold \Delta_{i,\vec{x}}\Ra{(\un{\set{Z}}_i)}, \bold \Delta_{V,j}\Ra{(\un{\set{Z}}_j)}] \\
            \mathbb{C}ov[\bold \Delta_{ V,i}\Ra{(\un{\set{Z}}_i)}, \bold \Delta_{j,\vec{x}}\Ra{(\un{\set{Z}}_j)}] & \mathbb{C}ov[\bold \Delta_{ V,i}\Ra{(\un{\set{Z}}_i)}, \bold \Delta_{ V,j}\Ra{(\un{\set{Z}}_j)}]
        \end{bmatrix},
    \end{align}
    where the diagonal terms can be seen in Propositions \ref{MultiSobolVar} and \ref{Var_delt}. Now, \\
    $\mathbb{C}ov[\bold \Delta_{i,\vec{x}}\Ra{(\un{\set{Z}}_i)}, \bold \Delta_{V,j}\Ra{(\un{\set{Z}}_j)}] = F_{ij} {\vec{E}}_{ij} + G_{ij} {\vec{C}}_{ij} + H_{ij} {\vec{U}}_{ij,0}$ such that
    \begin{align}
        F_{ij} &= \scalemath{0.9}{
            \frac{\Ncomb{i}{*}{j}{*}(\Nsing{i}{*}-1)}{\Nsing{i}{*}^2 \Nsing{j}{*}} - \frac{\Ncomb{i}{*}{j}{}(\Nsing{i}{*}-1)}{\Nsing{i}{*}^2 \Nsing{j}{}}  - \frac{\Ncomb{i}{}{j}{*}(\Nsing{i}{}-1)}{\Nsing{i}{}^2 \Nsing{j}{*}} + \frac{\Ncomb{i}{}{j}{}(\Nsing{i}{}-1) }{\Nsing{i}{}^2 \Nsing{j}{}}} \\
        G_{ij} &=
            \frac{\Ncomb{i}{*}{j}{*}}{\Nsing{i}{*}^2\Nsing{j}{*}} - \frac{\Ncomb{i}{*}{j}{}}{\Nsing{i}{*}^2\Nsing{j}{}} - \frac{\Ncomb{i}{}{j}{*}}{\Nsing{i}{}^2\Nsing{j}{*}} + \frac{\Ncomb{i}{}{j}{}}{\Nsing{i}{}^2\Nsing{j}{}}  \\
        H_{ij} &= 
            \frac{2\Ncomb{i}{*}{j}{*}(\Ncomb{i}{*}{j}{*} -1)}{\Nsing{i}{*}^2\Nsing{j}{*}(\Nsing{j}{*}-1)} - \frac{2\Ncomb{i}{*}{j}{}(\Ncomb{i}{*}{j}{} -1)}{\Nsing{i}{*}^2\Nsing{j}{}(\Nsing{j}{}-1)} \nonumber\\
            &\quad - \frac{2\Ncomb{i}{}{j}{*}(\Ncomb{i}{}{j}{*} -1)}{\Nsing{i}{}^2\Nsing{j}{*}(\Nsing{j}{*}-1)} + \frac{2\Ncomb{i}{}{j}{}(\Ncomb{i}{}{j}{} -1)}{\Nsing{i}{}^2\Nsing{j}{}(\Nsing{j}{}-1)}.
    \end{align}
\end{prop}
Finally, the covariance between the high-fidelity and discrepancy estimators is provided.

\begin{prop}[Variance between high-fidelity and discrepancies for MEV]
\label{VarMEVarHighLow}
    The covariance between the high-fidelity and discrepancy estimators is
     \begin{align}
        \mathbb{C}ov[\bold Q\Ra{({\set{Z}}_0)}, \bold \Delta_i\Ra{(\un{\set{Z}}_i)}] = \begin{bmatrix}
            \mathbb{C}ov[\bold Q_{0, \bold x}\Ra{({\set{Z}}_0)}, \bold \Delta_{i,\vec{x}}\Ra{(\un{\set{Z}}_i)}] &  \mathbb{C}ov[\bold Q_{0, \bold x}\Ra{({\set{Z}}_0)}, \bold \Delta_{V,i}\Ra{(\un{\set{Z}}_i)}] \\
             \mathbb{C}ov[\bold Q_{ V, 0}\Ra{({\set{Z}}_0)}, \bold \Delta_{i,\vec{x}}\Ra{(\un{\set{Z}}_i)}] &  \mathbb{C}ov[\bold Q_{ V,0}\Ra{({\set{Z}}_0)}, \bold \Delta_{V,i}\Ra{(\un{\set{Z}}_i)}]
        \end{bmatrix},
    \end{align}
    where the diagonal terms can be seen in Propositions \ref{MultiSobolHighLow} and \ref{Var_high}. The covariance between the high and low fidelities of the variance and ME variance estimators is
    \begin{align}
        \mathbb{C}ov[\bold Q_{0,\vec{x}}\Ra{({\set{Z}}_0)}, \bold \Delta_{V,i}\Ra{(\un{\set{Z}}_i)}] &= L_{0i} {\vec{E}}_{0i} + I_{0i} {\vec{C}}_{0i} + J_{0i} {\vec{U}}_{0i,0}\\
        \mathbb{C}ov[\bold Q_{V,0}\Ra{({\set{Z}}_0)}, \bold \Delta_{i,\vec{x}}\Ra{(\un{\set{Z}}_i)}] &= L_{i0} \{{\un{\mat{E}}}^T\}_{0i} + I_{i0} \{{\un{\mat{C}}}^T\}_{0i} + J_{i0} \{{\vec{U}}_{i0,0}\}^T,
    \end{align}
    such that
    \begin{align}
        L_{ij} &=  \frac{\Ncomb{i}{*}{j}{*}(\Nsing{i}{*} -1)}{\Nsing{i}{*}^2 \Nsing{j}{*}} - \frac{\Ncomb{i}{}{j}{}(\Nsing{i}{} -1)}{\Nsing{i}{}^2 \Nsing{j}{}} \\
        I_{ij} &= \frac{\Ncomb{i}{*}{j}{*}}{\Nsing{i}{*}^2 \Nsing{j}{*}} - \frac{\Ncomb{i}{}{j}{}}{\Nsing{i}{}^2 \Nsing{j}{}} \\
        J_{ij} &= 2\frac{\Ncomb{i}{*}{j}{*}(\Ncomb{i}{*}{j}{*}-1)}{\Nsing{i}{*}^2 \Nsing{j}{*}(\Nsing{j}{*}-1)} - 2\frac{\Ncomb{i}{}{j}{}(\Ncomb{i}{}{j}{}-1)}{\Nsing{i}{}^2 \Nsing{j}{}(\Nsing{j}{}-1)}.
    \end{align}
\end{prop}
\begin{remark}
     The Sobol estimator can also apply to other effects, not just the ME. We can let multiple indices of interest in $\set{Y}$ be dependent on \Ra{$\set{X}$}, and estimate a combined effect variance. For example, we can estimate \Ra{$V_{uw} = \mathbb{V}ar_{\vec{z}_{u,w}}[\mathbb{E}_{\Ra{\vec{z}}_{\sim u,w}}[f(\Ra{\vec{z}})|\vec{z}_{u,w}]]$} in the ANOVA decomposition using the ME variance estimator, and thus, the same MOACV estimator as introduced above can be used. Therefore, the MOACV Sobol estimator is not restricted to only ME variance, and other variance terms in the ANOVA decomposition can be calculated.
\end{remark}

\section{Synthetic Numerical Examples}
\label{sec:numerical}
In this section, the performance of the multi-output estimator is investigated on synthetic vector-valued functions. In Section \ref{ComparisonResults}, the variance of the introduced multi-output estimator is compared to individual ACV estimation, and superior performance is shown. Section \ref{PilotSampleTradeoff} explores the estimator performance when the required pilot covariances (given in the blue boxes of Section~\ref{sec:Body}) are estimated. Generally, we find that as the number of estimated outputs are increased,  more pilot samples are required if the pilot covariances are unknown. Moreover, higher order statistics, such as the ME variance, also can require significantly more samples.

The optimal sample allocations for the examples below are found by minimizing the determinant of the estimator variance \Ra{\eqref{eq:OptVar}} subject to cost constraints
\begin{align}
    \min_{\Ra{\set{Z}}} \left|\var{\cv\Ra{(\set{Z})}}\right|
     \textrm{\quad where \quad } C_0 \Ra{|\set{Z}_0|} + \sum_{i=1}^K C_i \left|\set{Z}_i \cup \set{Z}_i^* \right| \leq \textrm{Budgeted Cost}.
     \label{eq:varOpt}
\end{align}
This formulation is consistent with both determinant minimization used for optimal weight determination in~\eqref{eq:OptWei}, and variance minimization in the single ACV case. \Ra{The trace of the variance matrix may also be used as the objective function since Equation~\eqref{eq:OptWei} is also optimal in this case. Minimizing the trace may result in a different sample allocation. In this work, we choose to minimize the determinant, as this minimizes the confidence region and is consistent with previous works \cite{MVCV}. We leave the investigation of alternative sample allocation strategies for future work.} Since the covariances $\var{\cv\Ra{(\set{Z})}}$ and $\cov{\vec{Q}\Ra{(\set{Z}_0)}, \un{\vec{\Delta}}\Ra{(\un{\set{Z}})}}$ are functions of the number of estimator samples ($\Nsing{0}{}$, $\Nsing{i}{}$, $\Nsing{i}{*}$, etc.), the required pilot covariances are used to find the optimal sample allocation. 

For demonstration purposes, all estimators in this section follow the ACV-IS sampling scheme~\cite{ACV}, we find that changing this scheme does not change the qualitative conclusions. \Ra{ACV-IS greatly simplifies the input set, $\set{Z}$, which follows $\set{Z}_i^*\subset \set{Z}_i$ and $\set{Z}_i^*= \set{Z}_{i+1}^*$ for all $i=1,\ldots,K$. Other common sampling strategies include MLMC where $\set{Z}_i^*\cap \set{Z}_i = \emptyset$ and $\set{Z}_i = \set{Z}_{i+1}^*$ and MFMC where $\set{Z}_i^*\subset \set{Z}_i$ and $\set{Z}_i= \set{Z}_{i+1}^*$. Refer to the visualization in \cite[Fig. 2]{ACV} for other simple allocation strategies.} 

Optimal allocations are found using the MXMCPy library\footnote{\url{https://github.com/nasa/MXMCPy}} \cite{MXMCPy} for mean ACV and single-statistic MOACV estimators. Since MXMCPy does not offer optimization for variance estimation, the Scipy optimization library\footnote{\url{https://docs.scipy.org/doc/scipy/reference/optimize.html}} is used for variance estimators with optimizers SLSQP and SHGO.

The procedure in Section~\ref{sec:numerical} and Section~\ref{sec:application} estimates variance matrices. Traditional CV estimation of covariance matrices, however, may suffer from losing positive-definiteness \cite{PDCov}, which may lead to negative variance estimates. \Rb{Since our results directly compare variances, the diagonals of the covariance matrix, we set any estimated negative variances to zero.}
Next, the results in the following sections consider variance reduction, defined as the variance of a MC estimator divided by the variance of a multi-fidelity estimator\Ra{, $\var{\vec{Q}_{MC}(\set{Z})}/\var{\vec{Q}_{MF}(\set{Z})}$, so that it indicates how many times smaller the variance of the control variate estimator is compared to MC. Higher variance reduction corresponds to better estimator performance.} Finally, we use the following acronyms in the following sections. An MOACV estimator that estimates a single statistic, such as the mean or variance, is denoted as single MOACV (S-MOACV). An MOACV estimator that estimates multiple statistics is denoted as combined MOACV (C-MOACV).

\subsection{Comparison to ACV Estimators}
\label{ComparisonResults}
The variance reduction of MOACV estimators and individual ACV estimators is compared in this section. Consider a system with three models of decreasing fidelity, each with one input $x \sim \mathcal{U}(0,1)$ and three outputs
\begin{align}
    \Ra{\vec{f}}_0(x) &= \begin{bmatrix}
        \sqrt{11}x^5 & x^4 & \sin(2\pi x)
    \end{bmatrix}^T \\
    \Ra{\vec{f}}_1(x) &= \begin{bmatrix}
        \sqrt{7}x^3 & \sqrt{7}x^2 & \cos(2\pi x + \frac{\pi}{2})
    \end{bmatrix}^T \\
    \Ra{\vec{f}}_2(x) &= \begin{bmatrix}
        \frac{\sqrt{3}}{2}x^2 & \frac{\sqrt{3}}{2}x & \cos(2\pi x + \frac{\pi}{4})
    \end{bmatrix}^T.
\end{align}
The endowed costs of each model are shown in Table \ref{table:2}. We assume perfect knowledge of the covariance between the models and their outputs. These correlations are shown in Figure \ref{Corr}, and are computed using 100,000 pilot samples.

\noindent
\begin{minipage}{\textwidth}
  \begin{minipage}[b]{0.5\textwidth}
    \centering
    \includegraphics[trim={50pt 20pt 50pt 0},width=1\textwidth]{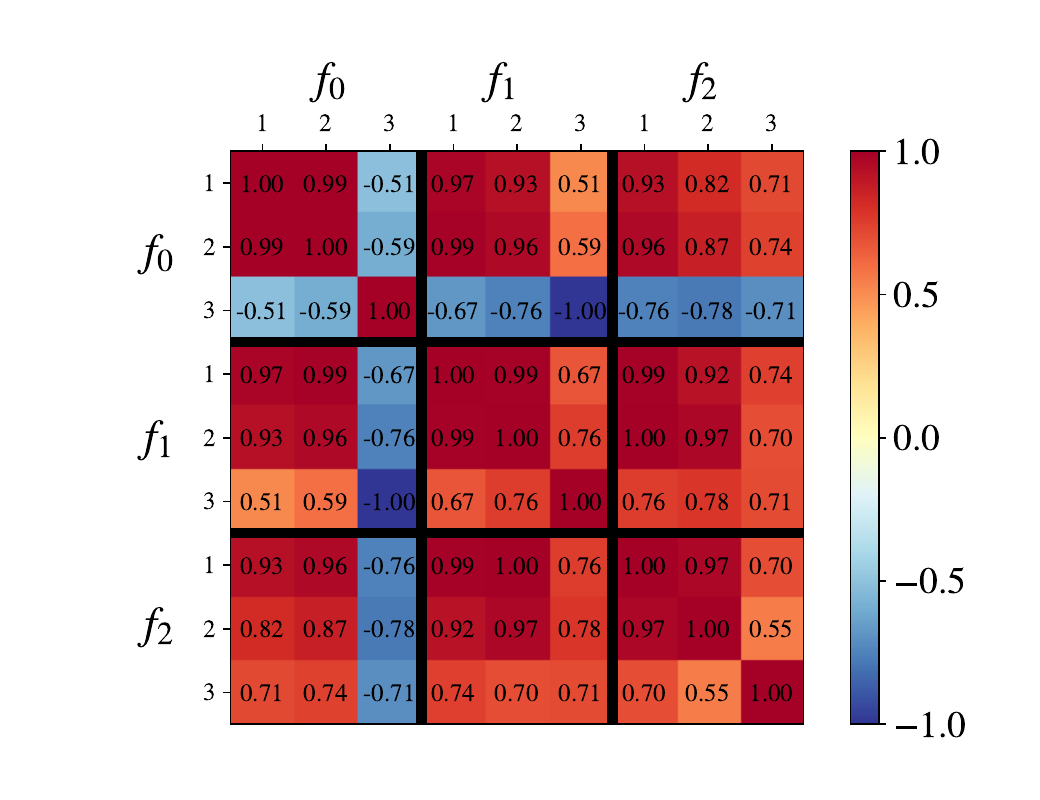}
    \vspace{-20pt}
    \captionof{figure}{Correlations between model outputs across model fidelities in Section \ref{ComparisonResults}.}
    \label{Corr}
  \end{minipage}
  \hfill
  \begin{minipage}[b]{0.45\textwidth}
    \small

    \centering
    \begin{tabular}{ | c | c | c c c |}
    \hline
    \multicolumn{5}{|c|}{Sample Allocations across Optimizations} \\
    \hline
    \hline
      Model & Cost & ACV & S-MOACV & C-MOACV   \\ 
      \hline
      $\Ra{\vec{f}}_0(x)$ & $1$ & $4$ & $2$ & $2$   \\
      $\Ra{\vec{f}}_1(x)$ & $0.01$ & $508$ & $499$ & $75$ \\
      $\Ra{\vec{f}}_2(x)$ & $0.001$ & $631$ & $2955$ & $7187$ \\
      \hline
      \multicolumn{2}{|c|}{Total Cost} & $9.711$ & $9.945$ & $9.937$ \\
      \hline
    \end{tabular}
    \captionof{table}{Sample allocations for each of the compared optimizations.}
    \label{table:2}
    
    \end{minipage}
\end{minipage}

\subsubsection{Mean and Variance Estimation}
\label{sec:MVcomb}
This section compares the variance reduction of individual ACV estimators with MOACV estimators for mean and variance estimation. \Rb{While we may estimate the full $3\times 3$ covariance matrix, we choose to compare only the diagonal elements of the matrix to simplify the comparison.} To estimate the mean and variance of each model output, six individual ACV estimators are constructed to estimate the three means and three variances \Rb{of the model outputs}. Two S-MOACV estimators are constructed, one for the three means and one for the $3\times 3$ covariance matrix. 
Finally, a C-MOACV is created to estimate the three means and the $3\times 3$ covariance matrix simultaneously. The same sample allocation is used for all estimators in this section, which is found by minimizing the variance of the ACV mean-estimator for the first output of $\Ra{\vec{f}}_0$ given a budget of 10 and is shown in the ACV column in Table \ref{table:2}. Note, the sample allocations shown in the other columns will be used in the subsequent sections.

 Figures \ref{MeanComp} and \ref{VarComp} display the variance reduction of the mean and variance estimators respectively. As seen in Figure \ref{MeanComp}, both the S-MOACV and C-MOACV estimators achieve greater variance reduction than the ACV estimator of each output by over an order of magnitude in some cases. C-MOACV also achieves improved variance reduction over the mean-specific S-MOACV. Figure \ref{VarComp} \Rb{only shows the diagonal elements of the 3 $\times$ 3 covariance matrix and }demonstrates similar results for the estimator variance, but indicates less benefit of C-MOACV over S-MOACV. The significantly improved variance reduction from S-MOACV estimation demonstrates that multi-output estimation excels in systems with the high correlations between the model outputs, as seen in Figure \ref{Corr}.  

\begin{figure}[h]
\centering
\begin{subfigure}{0.45\textwidth}
    \includegraphics[width=1\linewidth]{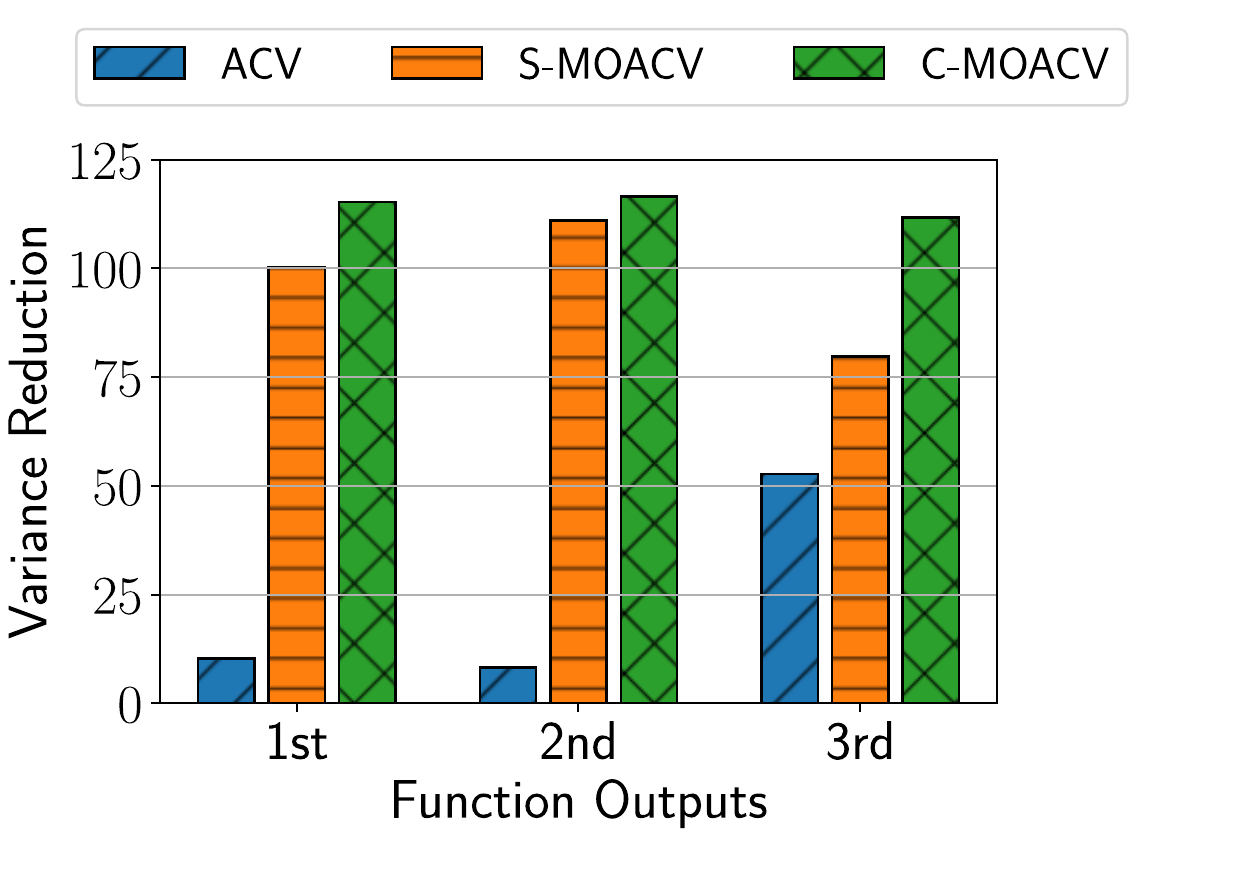}
    \caption{Mean Estimation}
    \label{MeanComp}
\end{subfigure}
\begin{subfigure}{0.45\textwidth}
    \includegraphics[width=1\linewidth]{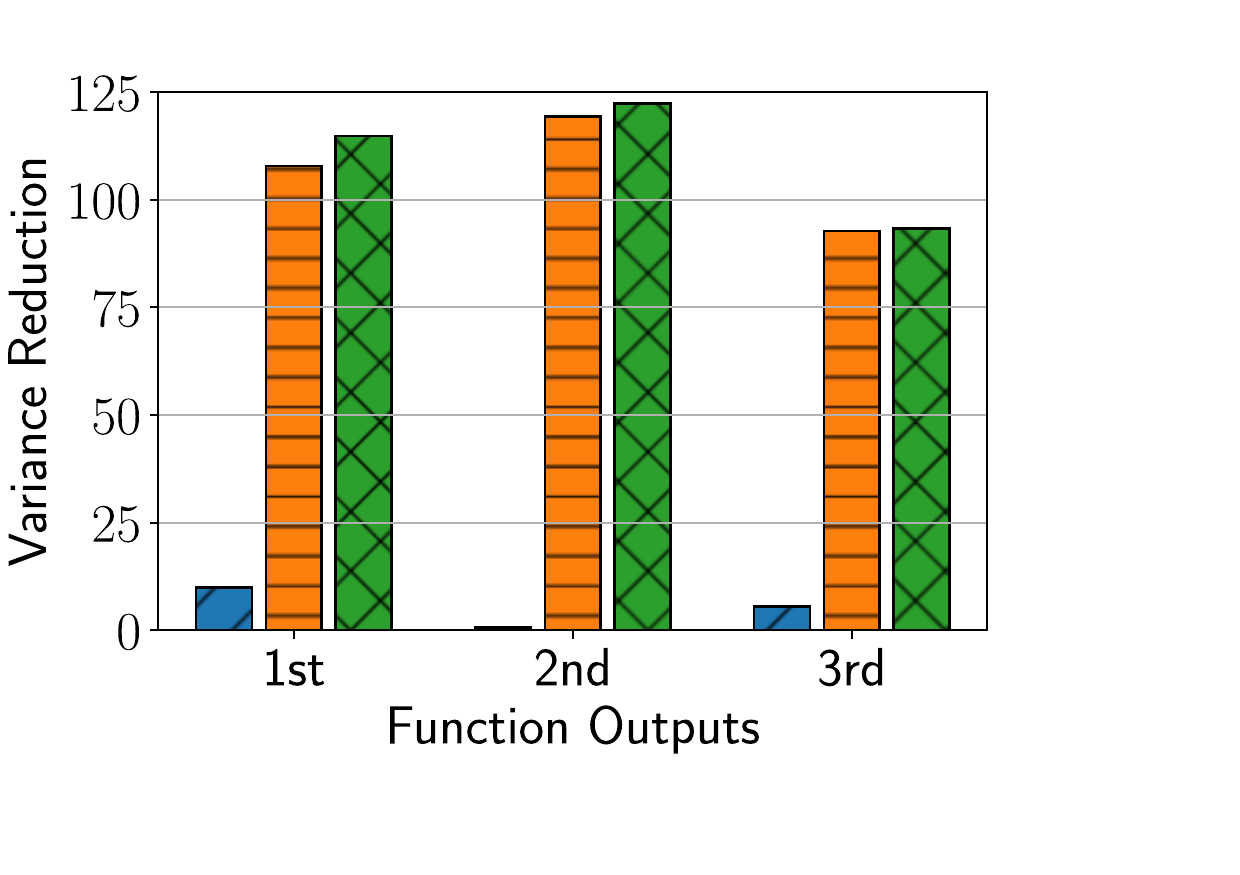}
    \vspace{-25pt}
    \caption{Variance Estimation}
    \label{VarComp}
\end{subfigure}
\caption{Variance reduction compared to MC of each function output for each estimator type (higher is better) in Section \ref{sec:MVcomb}. The C-MOACV and S-MOACV estimator provides significantly increased variance reduction compared to ACV methods.}
\label{MeanVarComp}
\end{figure}

\subsubsection{Sample Allocation Optimizations}
\label{sec:SAopts}

In the previous section, all the estimators used the same sample allocation obtained from  minimizing the variance of the mean ACV estimator for the first model output. In this section, performance of sample allocations that target variance reduction of the full multi-output estimator are demonstrated. The results in this section show variance reduction achieved only for the mean of the first model output, because that is all that the ACV can provide. We reinforce that the MOACV estimators also provide significant variance reduction for all other outputs as well.

Table \ref{table:2} shows the optimal sample allocations for various objective functions. The first column, ACV, is the ACV-specific allocation for the first model output found in the previous section. The second (S-MOACV) and third (C-MOACV) columns arise from minimizing the determinant of the estimator variance obtained by the two MOACV estimators, respectively. While the optimization method resulted in sample allocations of slightly different costs\footnote{The different allocation costs is a consequence of simplifying the discrete optimization problem into a continuous domain. The rounding of the results of the continuous optimization into the discrete solution causes the solution to not lie on the computational budget boundary. However, the S-MOACV and ACV costs only have a 2\% difference.}, the variance reduction metric is cost-independent since it divides the MC estimator variance by an equivalent-cost multi-fidelity estimator variance.

\begin{figure}
    \centering
    \includegraphics[trim={0 50pt 0 0},width=0.55\linewidth]{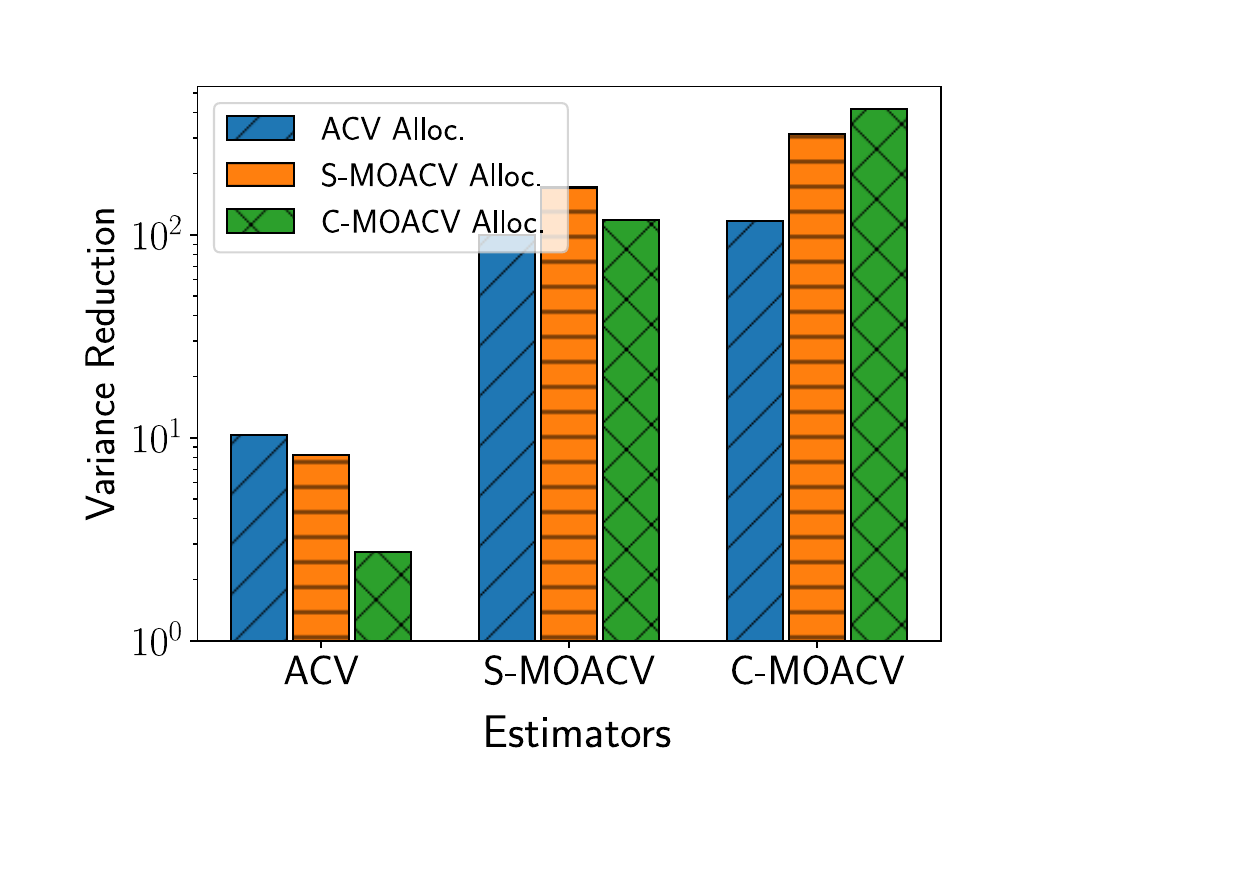}
    \captionof{figure}{Variance reduction of mean estimation of the first model output compared to MC across optimizations in Section \ref{sec:SAopts}. The C-MOACV and S-MOACV estimators outperform the ACV estimator across every sample allocation optimization.}
    \label{OptComp}
\end{figure}

The variance reduction results are shown in Figure \ref{OptComp}.
Each of the optimizations give the best variance reduction for their respective estimators. For example, the C-MOACV estimator that uses a C-MOACV optimal sample allocation achieves greater variance reduction than a C-MOACV estimator that uses the ACV optimal sample allocation. The C-MOACV estimator outperforms the ACV estimator by at least an order of magnitude under all allocation strategies. We reinforce that while variance reduction significantly improves for the mean of the first output, the combined MOACV estimator also returns the means, variances, and covariances of all other model outputs. 

\subsection{Pilot Sample Trade-off}
\label{PilotSampleTradeoff}

\Ra{Multi-fidelity estimation requires additional statistics to perform optimal estimation. Typically, an independent exploratory pilot study is performed before estimation using an independent set of samples to estimate these statistics. In this section we explore the effects of such an approach on the MOACV, while in Appendix \ref{sec:PSstrat}, we demonstrate other strategies for pilot studies that breaks the independence assumption.}
The multi-output estimators introduced in this work require exploiting more information than simple single-output estimators. Specifically, the boxes of Section \ref{sec:Body} show a large number of statistics that must be known to compute the optimal CV weights. A natural question arises as to whether there are too many unknowns to allow a small set of pilot samples to yield an effective estimate. In this section, it is shown that the required number of pilot samples depends on the number of model outputs and statistics that are estimated. The type of statistic estimated is the largest contributor to the number of pilot samples that is required. Adding more model outputs also increases the number of required samples.

\Rb{Section \ref{sec:PSlow} introduces a system with low correlations between model outputs and displays the variance reduction of mean and ME variance estimation across the number of pilot samples and the number of model outputs available. Section \ref{sec:percent} introduces a system with high correlations between model outputs that investigates the total cost of constructing estimators, including the cost of the pilot study. The MOACV estimator is able to outperform MC and ACV estimation in both cases, even when considering the cost of the pilot study in Section \ref{sec:percent}. Figure \ref{fig:PScorrs} shows the correlations between the model outputs and fidelities for each system.}

\begin{figure}[h]
    \centering
    \begin{subfigure}{0.45\textwidth}
        \includegraphics[trim={0 40pt 0 0},width=1\linewidth]{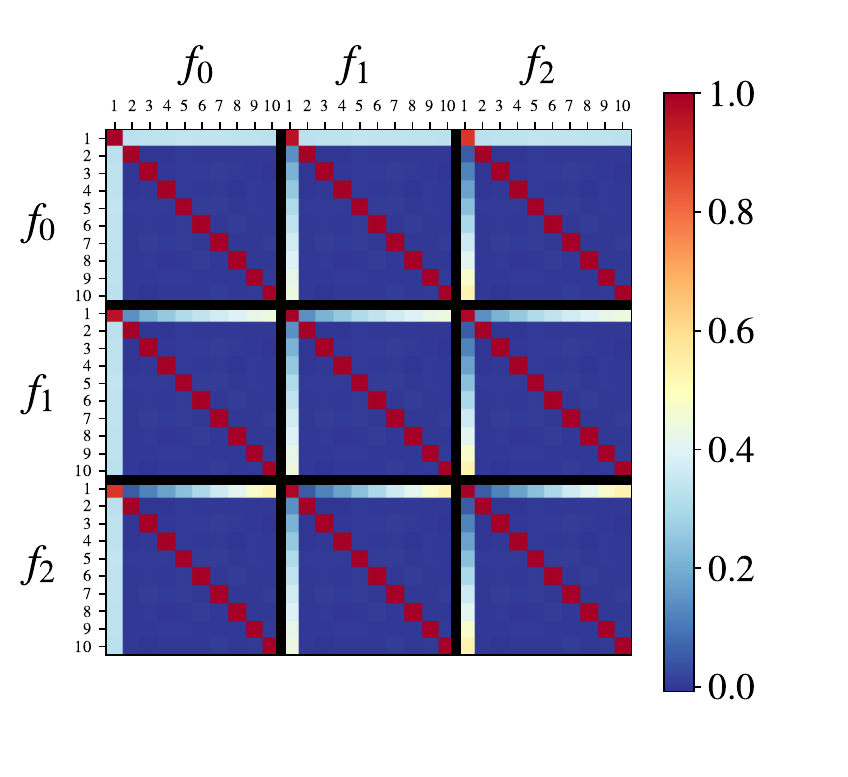}
        \caption{System in Section \ref{sec:PSlow}}
        \label{PScorrs1}
    \end{subfigure}
    \begin{subfigure}{0.45\textwidth}
        \includegraphics[trim={0 10pt 0 0},width=1\linewidth]{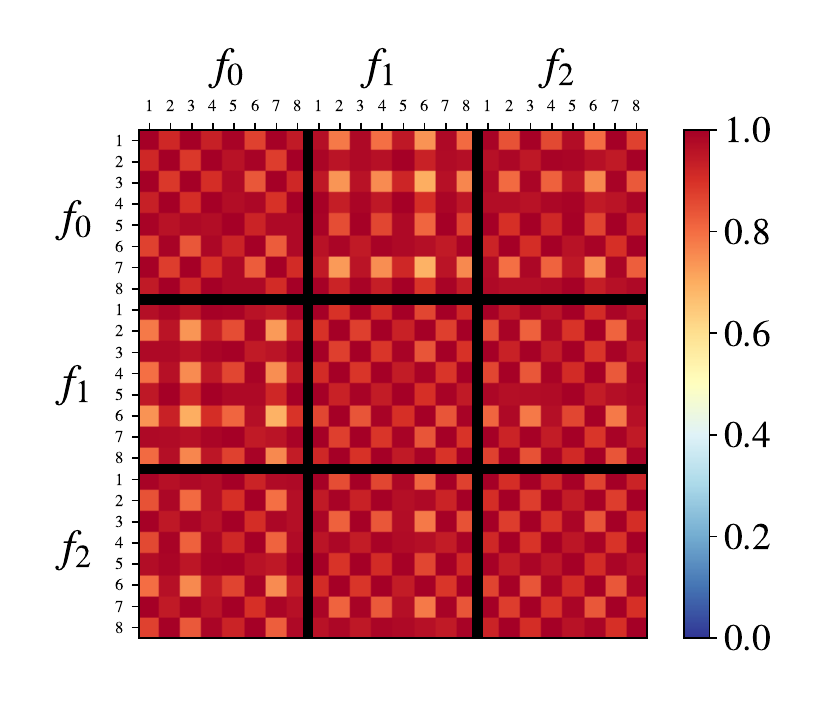}
        \caption{System in Section \ref{sec:percent}}
        \label{PScorrs2}
    \end{subfigure}
    \caption{\Rb{Correlations between model outputs and fidelities for each system in Section \ref{PilotSampleTradeoff}.}}
    \label{fig:PScorrs}
\end{figure}

\subsubsection{Number of Pilot Samples}
\label{sec:PSlow}
A new system is defined to consider a tunable number of function outputs to study convergence of increasingly complex estimators as a function of the number of pilot samples. Let the high- and low-fidelity functions be $\Ra{\vec{f}}_0, \Ra{\vec{f}}_1, \Ra{\vec{f}}_2: \vec{x} \rightarrow \mathbb{R}^{10}$, where $\vec{x} \in [0,1)^{9}$ is uniformly distributed such that
\begin{align}
    \Ra{\vec{f}}_0(\vec{x}) = \begin{bmatrix}
        \sum_{i=1}^9 x_i^3 \\
        x_1^3 \\ 
        \vdots, \\ 
        x_9^3
    \end{bmatrix} \textrm{\quad \quad}
    \Ra{\vec{f}}_1(\vec{x}) = \begin{bmatrix}
        \sum_{i=1}^9 \sqrt{i}x_i^3 \\
        \sqrt{1}x_1^3 \\ 
        \vdots, \\ 
        \sqrt{9}x_9^3
    \end{bmatrix} \textrm{\quad and \quad}
    \Ra{\vec{f}}_2(\vec{x}) = \begin{bmatrix}
        \sum_{i=1}^9 ix_i^3 \\
        1x_1^3 \\ 
        \vdots, \\ 
        9x_9^3
    \end{bmatrix}.
\end{align}
\Rb{The correlations in this system are shown in Figure \ref{PScorrs1}. The correlations between the high- and low-fidelity models for the first output are 0.96 and 0.89, inline with correlations reported in other applications in the literature \cite{CVQE,PDCov,EACV1}. The correlations between the outputs, however, are 0 or $0.33$, which are quite low.}
We study the variance reduction in the mean and main-effect variances of the first output of $\Ra{\vec{f}}_0$. For mean estimation, we consider an increasing number of model outputs formed by using more components of each of the model fidelities. For ME variance estimation, we consider an increasing number of ME variances to estimate across the 9 inputs. The ACV-IS sampling scheme was chosen with the un-optimized allocations of each fidelity being $50$, $500$, and $5000$. Figure~\ref{PSlargeplot} shows the variance reduction achieved for different numbers of pilot samples and statistics. Note that the bottom edge of the plots corresponding to one output for mean estimation in Figure \ref{PSplot} and one ME in Figure \ref{PSSAplot} corresponds to the performance of the standard ACV estimator.

\begin{figure}[h]
    \centering
    \begin{subfigure}{0.47\textwidth}
        \includegraphics[trim={0 0 0 0},width=1\linewidth]{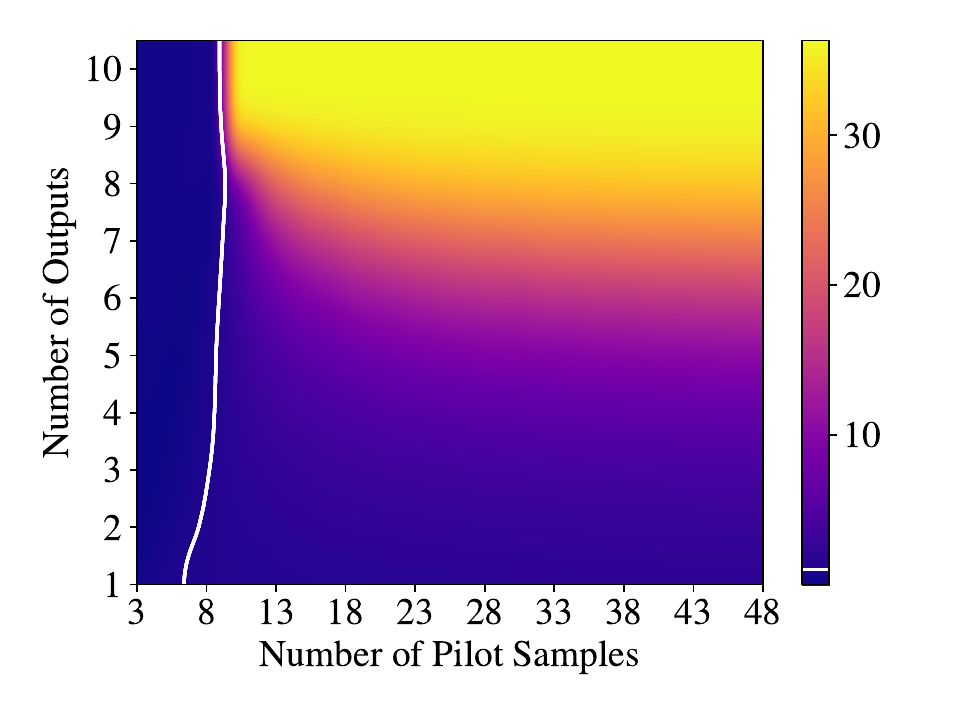}
        \caption{Mean Estimation}
        \label{PSplot}
    \end{subfigure}
    \begin{subfigure}{0.45\textwidth}
        \includegraphics[trim={0 0 20pt 0},width=1\linewidth]{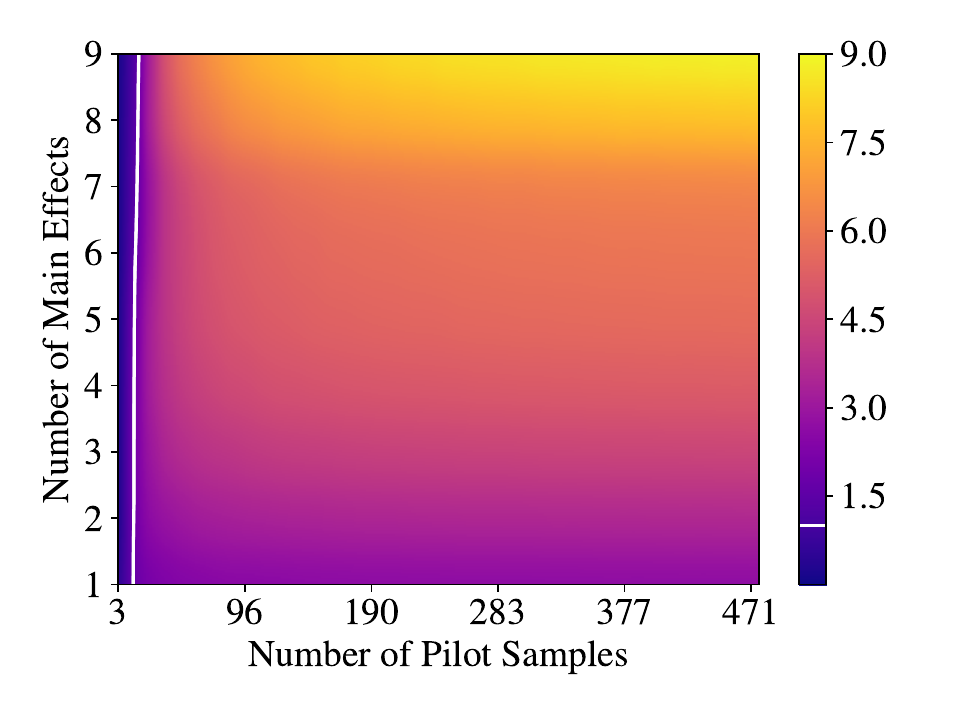}
        \caption{ME Variance Estimation}
        \label{PSSAplot}
    \end{subfigure}
    \vspace{-15pt}
    \caption{5th percentile of variance reduction (worst case) across 1000 trials as a function of pilot samples and number of outputs in Section \ref{sec:PSlow}. The white contour line represents the same performance as MC estimation. The best variance reduction achieved includes all outputs for MOACV mean estimation, with marginally more pilot samples required than ACV estimation. }
    \label{PSlargeplot}
\end{figure}

To determine the performance of the variance reduction with respect to the number of pilot samples, we sweep across combinations of {\it additional} model outputs (from 0 to 9) and numbers of pilot samples. At each pilot sample quantity, we run 1000 realizations of the pilot sample sets and compute the estimator variances. Figure~\ref{PSplot} shows the 5th percentile of the variance reduction, a statistic that demonstrates close to worst-case behavior.

The white line corresponds to the variance reduction ratio of 1, where MC has equal performance to the CV approach. Notably, we see a sharp transition at 10 pilot samples where the performance improves over MC.  With too few pilot samples, the estimators perform worse than MC. Since the white contour line is near vertical, Figure \ref{PSplot} displays that adding more correlated outputs in mean estimation only requires slightly more pilot samples for significantly improved variance reduction. 

Next we repeat the same experiment for the ME variances of the first output of $f_0$ with respect to each of the 9 inputs. As described in Section \ref{sec:sensana}, there are many more required covariances to be estimated for multiple ME estimators than for mean estimation. The statistic of interest is also of a higher order than mean estimation. Figure \ref{PSSAplot} again shows the 5th percentile of the variance reduction for the first ME variance output over the 1000 trials at different combinations of outputs. In this case, the number of outputs in the estimator reflects the number of MEs that are estimated, a maximum of 9 for the 9 total inputs. Note the different axis scales between the two plots. Similarly to mean estimation, the additional statistics can improve the variance reduction.  However, the number of required pilot samples to outperform ACV (green area, bottom edge of Figure \ref{PSSAplot}), is about 25 samples, which is double the number of samples required for mean estimation. Further, the maximum variance reduction requires around 250 pilot samples before the C-MOACV estimator variance converges. Overall, we see a similar pattern to the mean estimation with more required samples. In both mean and ME estimation, the MOACV estimators achieve larger variance reduction than ACV estimation when the ACV estimator variance has converged. Future work can focus on adaptive schemes to determine the optimal number of pilot samples.

\subsubsection{Percent of Pilot Samples}
\label{sec:percent}

\Rb{Next we investigate a similar comparison, but now consider the pilot sample estimation as a portion of the cost. We also consider a new case with more varying correlations amongst models to provide a situation where many non-zero correlations, as seen in Figure \ref{fig:PScorrs}, must be leveraged. Consider the following model where $x\in[-1,1)$ is uniformly distributed such that
\begin{align}
    \vec{f}_0(x) = \begin{bmatrix}
        x \\ x^3 \\ \sin(x) \\ \sin(x^3) \\ e^x -1 \\ e^{x^3}-1 \\ \log(x+1) \\ \log(x^3 +1)
    \end{bmatrix}  \textrm{\quad \quad}
    \vec{f}_1(x) = \begin{bmatrix}
        |x|x \\ |x^3|x^3 \\ |x|\sin(x) \\ |x^3|\sin(x^3) \\ |x|(e^x -1) \\ |x^3|(e^{x^3}-1) \\ |x|\log(x+1) \\ |x^3|\log(x^3 +1)
    \end{bmatrix}  \textrm{\quad \quad}
    \vec{f}_2(x) = \begin{bmatrix}
        \sqrt{|x|}x \\ \sqrt{|x^3|}x^3 \\ \sqrt{|x|}\sin(x) \\ \sqrt{|x^3|}\sin(x^3) \\ \sqrt{|x|}(e^x -1) \\ \sqrt{|x^3|}(e^{x^3}-1) \\ \sqrt{|x|}\log(x+1) \\ \sqrt{|x^3|}\log(x^3 +1)
    \end{bmatrix}.
\end{align}
Assume that we have a total computational budget of 500 where $\vec{f}_0$, $\vec{f}_1$, and $\vec{f}_2$ have a cost of $1$, $0.1$, and $0.01$ respectively. We study how the variance of the mean estimator changes as a function of the budget that we allocate to the pilot sample study.

We estimate the mean of the first output for the high-fidelity function, $\vec{f}_0(x)$, as we change the percent of the budget that is allocated for pilot samples and as we change the number of function outputs that are available to the MOACV estimator. We calculate the variance of the estimator empirically using 50 independent estimators for MC, ACV, and MOACV estimation. We also use 20 independent sets of pilot samples at each budget percentage. The ACV and MOACV estimators optimize their sample allocations based on the set of pilot samples given at each iteration.

Figure \ref{PercBud} displays the estimator variance of the mean estimator for the first output over the percent of the budget that we allocate for the pilot samples. The x-axis is on a log-scale that ranges from 1\% to 90\% of the budget. Each subplot represents the total number of outputs that the MOACV has access to. The shaded regions covers 90\% of the calculated estimator variances from the 20 independent sets of pilot samples. The MC and ACV estimators do not change between subplots since they both only estimate the first output. The MC estimator also does not vary over the percent of the budget since the MC estimator does not depend on any pilot samples. The variance of the MC estimator uses the entire budget of 500. The variance of the ACV estimator has a wider spread when there are few pilot samples and the estimator variance increases when most of the budget is allocated towards pilot sampling. When $90\%$ of the budget is allocated to pilot samples, the variance of the ACV estimator increases because the remaining budget has too few samples to estimate the $\vec{f}_0(x)$ mean. The variance of the MOACV estimator also increases when there are too few pilot samples and decreases when there are enough samples to properly estimate the pilot covariances. In this specific example, the variance of the MOACV estimator reduces $70\%$ from including 1 output to including 4 outputs. The estimator variance, however, does not reduce between 4 outputs and 8 outputs. This suggests that the additional correlations added by more model outputs does not reduce the variance of the estimator further. This saturation of the estimator is problem-dependent and likely depends on the relationships between the model outputs.
This figure demonstrates that the performance of the MOACV estimator depends on the number of outputs and the number of pilot samples available. 

While Figure \ref{PercBud} is for a budget of 500, when the budget is increased, less of the budget needs to be allocated to pilot samples for MOACV estimation. Similarly, when the budget is decreased, more of the budget needs to be allocated to pilot sampling. Even with 90\% of the budget allocated to pilot sampling, the MOACV estimator outperforms MC estimation. Finally, while the main contribution of this work is to provide a new estimator that can exploit multi output information, future works may use this in an adaptive scheme to find the best number of pilot samples to use. 
}

\begin{figure}
    \centering
    \includegraphics[trim={70pt 0 70pt 0},width=0.9\linewidth]{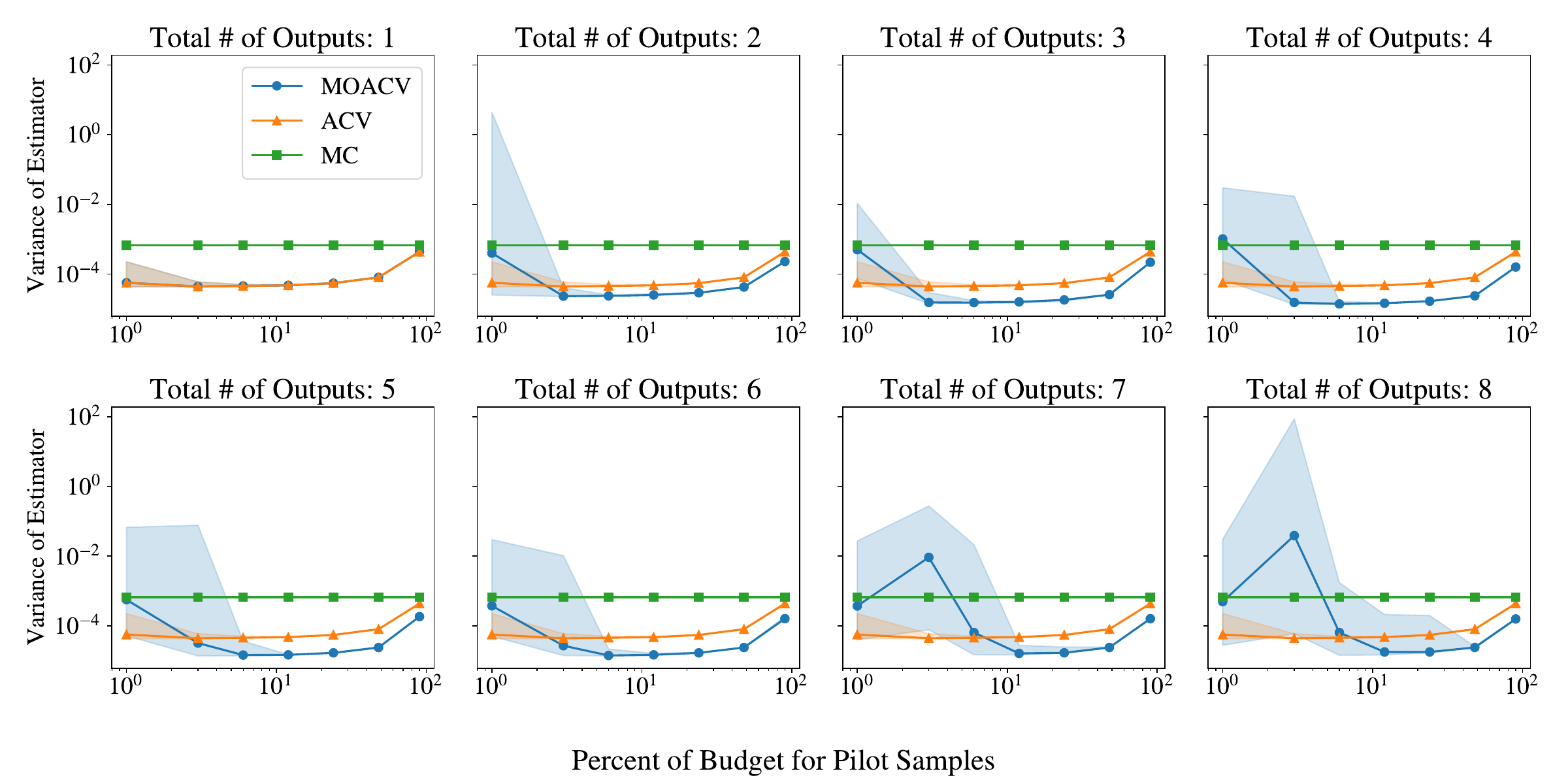}
    \captionof{figure}{\Rb{Estimator variance over the percent of the budget for pilot samples in Section \ref{sec:percent}. Each plot represents the number of outputs available to the MOACV estimator. The shaded region represents 90\% of the uncertainty in the estimator variance using 20 sets of independent pilot samples. The MOACV estimator generally outperforms MC and ACV estimation. With too few pilot samples, the MOACV estimator has large variance if too many outputs are involved.}}
    \label{PercBud}
\end{figure}

\section{Application: Entry, Descent, and Landing Trajectories}
\label{sec:application}

Entry, descent, and landing (EDL) is the final phase of a space vehicle's mission upon entering the atmosphere of a celestial body. An important aspect of successful EDL includes prediction of trajectory and touchdown properties including locations, velocities, and states of a vehicle at given times. However, these predictions are difficult because of uncertainties due to the atmosphere, initial vehicle states, and actuator precision. Analyzing predicted outcomes due to these uncertainties is also computationally challenging because high-fidelity simulations may \Ra{be expensive} to run. In this section, we consider the simulation of a sounding rocket with the aim of reducing the computational cost of estimation through multi-fidelity methods.

NASA launched the Sounding Rocket One (SR-1) in September 2018 containing the Adaptable, Deployable, Entry, and Placement Technology (ADEPT), aimed to demonstrate a deployable aeroshell used for re-entry \cite{SR1ADEPT}. Before launch, this flight was simulated using the Program to Optimize Simulated Trajectories II (POST2) software \cite{POST2} with a standard MC approach to consider system uncertainties \cite{SR1ADEPT}. The POST2 software contains around 75 uncertain inputs including initial conditions (e.g. location, velocity, angle of attack), vehicle parameters (e.g. moment of inertia, deployment impulse), and environmental parameters (e.g. atmospheric uncertainty). In Warner, et al. \cite{Traj}, ACV techniques were used to construct mean estimators for 15 trajectory QoIs, such as the touchdown latitude, longitude, velocity, and other QoIs listed in \cite[Table 1]{Traj}. Using multi-fidelity techniques, \cite{Traj} was able to reduce the variance of estimation for many of the 15 QoIs. The goal of this section is demonstrate further variance reduction using multi-output estimation. The following models of varying fidelity were introduced in \cite{Traj} to aid the multi-fidelity estimation. The POST2 simulation is used as the high-fidelity model and it takes 219 seconds on average for a single evaluation at a fixed condition. A ``reduced-physics" version of POST2 is introduced to reduce the cost of simulation by using a simplified atmospheric model, taking around 47.4 seconds per evaluation. A cheaper trajectory simulation is also created using the high-fidelity POST2 at a much larger integration time step at 2.8 seconds per evaluation, deemed the ``coarse time-step" model. Finally, a support vector machine (SVM) surrogate model (``machine learning model") is trained offline using 250 high-fidelity trajectory simulations and used as a low-fidelity model taking around 0.0007 seconds per evaluation. 

In this section, we compare the performance of multi-output methods to ACV estimation for 9 of the 15 QoIs. The 9 QoIs were chosen for their correlations between other QoIs, as seen in Figure \ref{ADcorr}, where the QoIs in red are removed from this study. Section \ref{sec:MVadept} compares ACV and MOACV methods by estimating the mean and variance of 9 QoIs. Finally, Section \ref{sec:adeptSA} uses MOACV to perform a sensitivity analysis on one QoI across three input variables.

\subsection{Mean and Variance Estimation}
\label{sec:MVadept}

In this section, we build 18 ACV estimators for the mean and variance of each of the 9 QoIs; two S-MOACV estimators, one for the 9 mean estimator and one for the \Rb{9 $\times$ 9 covariance matrix} estimator; and a single C-MOACV estimator for the mean and \Rb{covariance matrix} of the 9 QoIs simultaneously. In particular, the C-MOACV estimator simultaneously estimates 54 statistics (9 means and 45 unique covariances).

To find the pilot covariances, 60,000 pilot samples were used. With these samples, Figure \ref{ADmcorr} shows the correlations between the models across the model outputs. Notably, a few QoIs have low correlations between the low-fidelity and high-fidelity models. Traditionally, poor variance reduction is expected at these QoIs for multi-fidelity estimation. \Ra{Control variate estimation is only useful if the additional variables (low-fidelity models) have high correlations with the primary estimator (high-fidelity model) \cite{CV1, CV}.} The correlations between the outputs of the high fidelity model can be seen in Figure \ref{ADcorr}. The non-zero correlations are exploited in the MOACV techniques and used to provide more accurate estimation. \Ra{The squared biases of the low-fidelity models with respect to the high-fidelity models are shown in Appendix \ref{sec:EDLbias}.}

\begin{figure}[h]
    \centering
    \begin{subfigure}[t]{0.45\textwidth}
        \includegraphics[trim={20pt 0 0 0},width=1\linewidth]{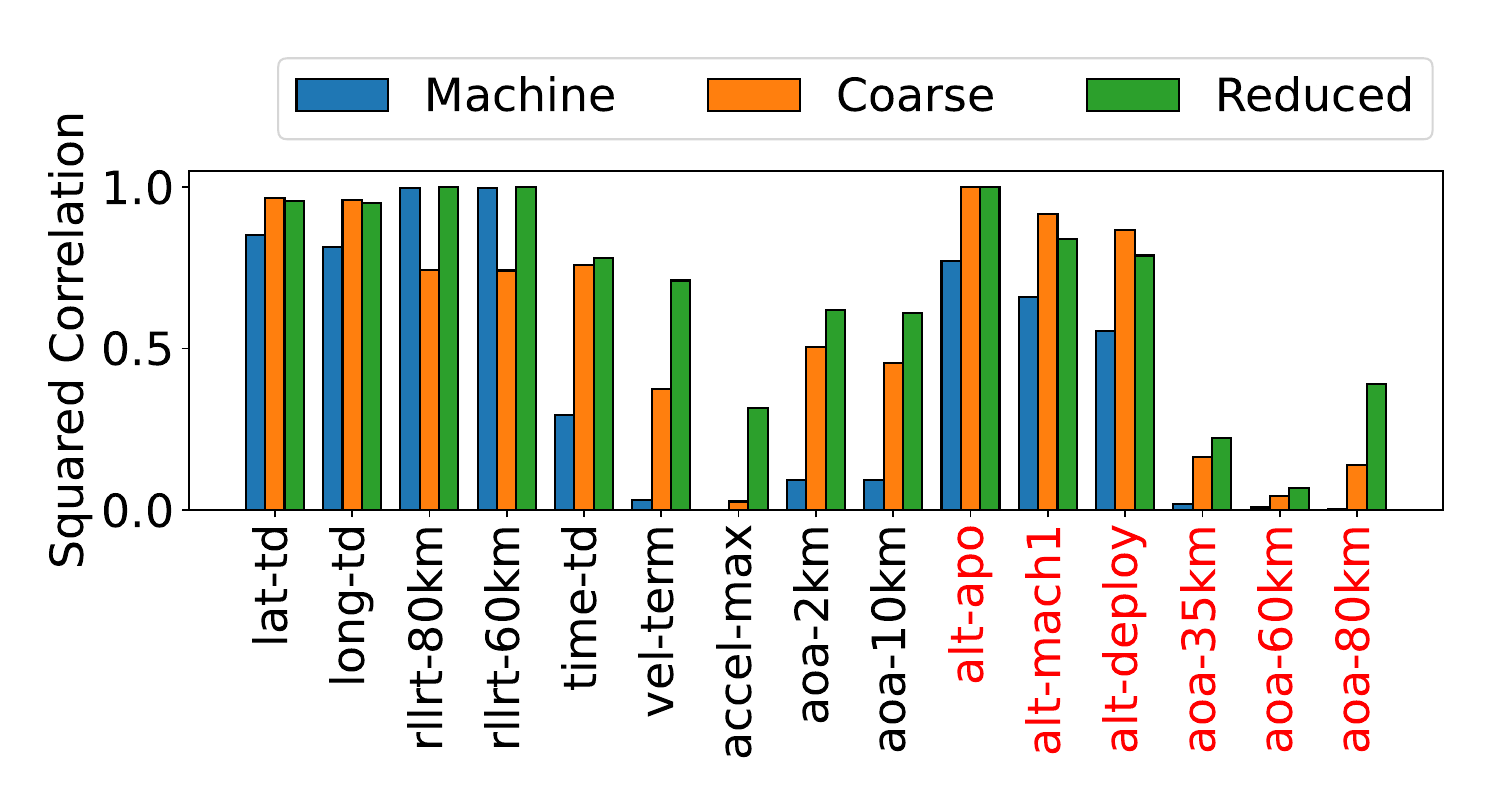}
        \caption{Squared correlations between the high-fidelity model and low-fidelity models for each output in Section \ref{sec:application}. Recreated from \cite{Traj}.}
        \label{ADmcorr}
    \end{subfigure}
    \begin{subfigure}[t]{0.45\textwidth}
        \includegraphics[trim={0 0 150pt 0},width=1\linewidth]{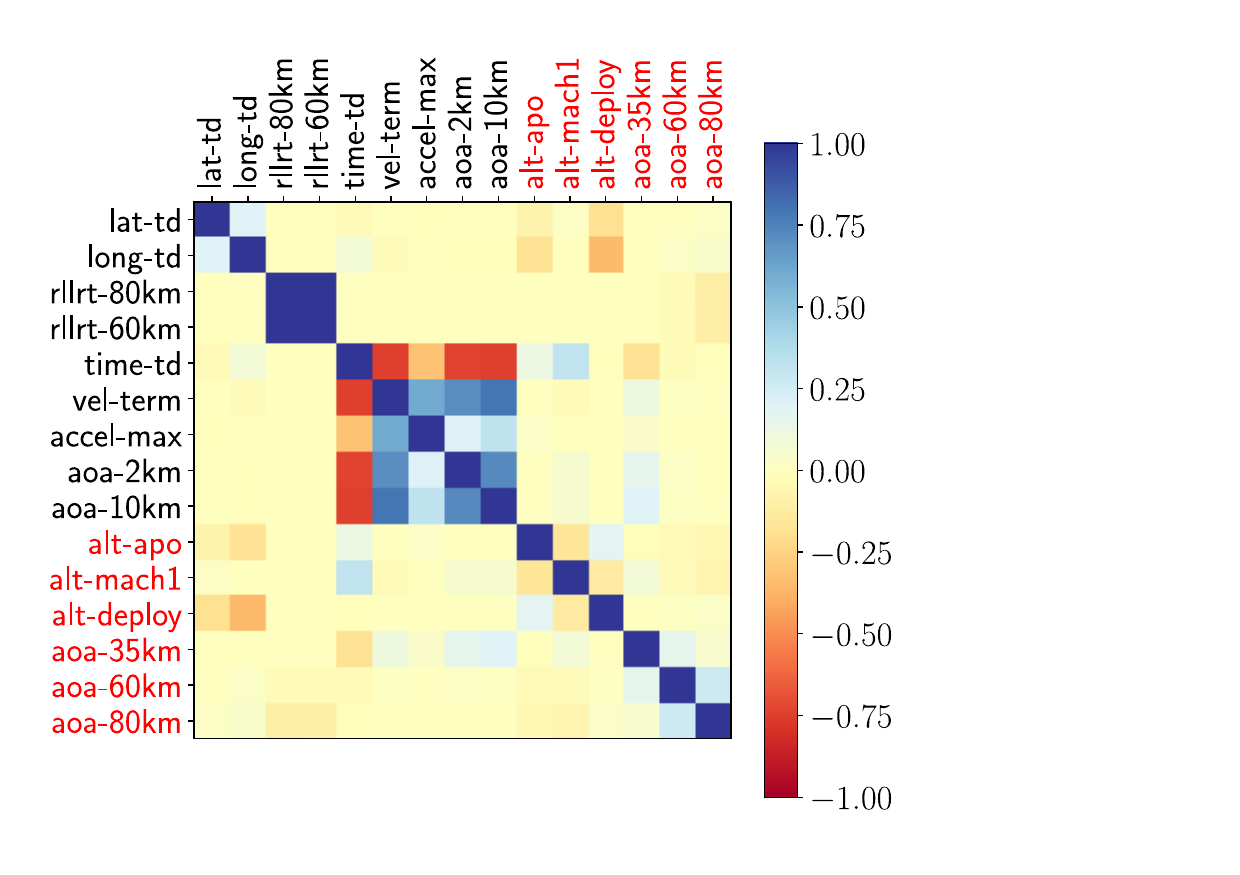}
        \caption{High-fidelity model output correlations.}
        \label{ADcorr}
    \end{subfigure}
    \caption{Correlations between model fidelities and model outputs for the EDL problem~\ref{sec:application}. Red QoIs are not estimated in this study.}
    \label{ADcorrlarge}
\end{figure}

A single ACV-IS allocation scheme was applied to all outputs to enable a fair comparison at equivalent computational costs. This sample allocation for all estimators was computed to minimize the variance of the ACV mean-estimator for the touchdown latitude (\textit{lat-td}). Similarly to Warner et al. \cite{Traj}, the optimization minimizes the variance with a computational budget of $10^4$ seconds. \Rb{Refer to Appendix \ref{sec:EDLPilot} for a similar study with a larger budget that is comparable to the pilot study cost.} The allocated samples are 31, 0, 1124, and 22075 samples for the POST2, reduced physics model, coarse time step model, and the machine learning model respectively. \Ra{The sample allocation optimizer decided to not use the reduced physics model even though it has a high correlation with the high-fidelity model for the touchdown latitude. This is caused by the high cost of the reduced physics model compared to the coarse time-step model, which is an order of magnitude cheaper and is more correlated to the high-fidelity model for this QoI. This finding demonstrates that not all highly correlated models are useful for multi-fidelity estimation, especially if their costs are high.}

We obtain the empirical variance of the estimators using 10,000 realizations of data according to the above sample allocation. The variance reduction achieved for mean and variance estimation can be seen in Figure \ref{AD_orig}. The red dotted line represents no reduction compared to the equivalent-cost MC estimator. The individual ACV estimator reduction can be seen in the blue bars. In Figure \ref{ADm} for mean estimation, the S-MOACV and C-MOACV achieves greater variance reduction than individual ACV estimation at every QoI. In Figure \ref{ADv}, both the S-MOACV and the C-MOACV estimators achieve greater variance reduction than the ACV estimator at every QoI. 

Figure \ref{AD_orig} demonstrates that the MOACV estimators can turn situations where an ACV estimator performed worse than MC, into one where performance becomes better than MC. For example, the ACV estimator for the terminal velocity ``vel-term" initially performs worse than MC estimation. However, the C-MOACV estimator is able to achieve reduction better than MC by leveraging the additional correlations. \Ra{Figure \ref{AD_ACV} displays the variance reduction compared to ACV estimators for each of the outputs.} For mean estimation, the S-MOACV estimator achieves a median 15\% greater variance reduction than ACV estimators. The C-MOACV estimator achieves a median 39\% larger variance reduction than ACV estimators, with a maximum of 113\% larger reduction for ``rllrt-60km". For variance estimation, the C-MOACV estimator provides a median 22\% greater variance reduction than ACV estimators. The C-MOACV estimates for landing latitude and longitude performed marginally better (about 1\% larger reduction) than ACV estimation. This performance is explained by the lack of correlation amongst latitude and longitude with other QoIs, as seen in Figure \ref{ADcorr}. The S-MOACV and C-MOACV estimators are able to outperform traditional ACV methods by extracting the correlations between QoI and statistics to reduce the variance of ACV estimation even further.

\begin{figure}[h]
\centering
\begin{subfigure}{0.45\textwidth}
    \includegraphics[trim={0 30pt 0 0},width=1\linewidth]{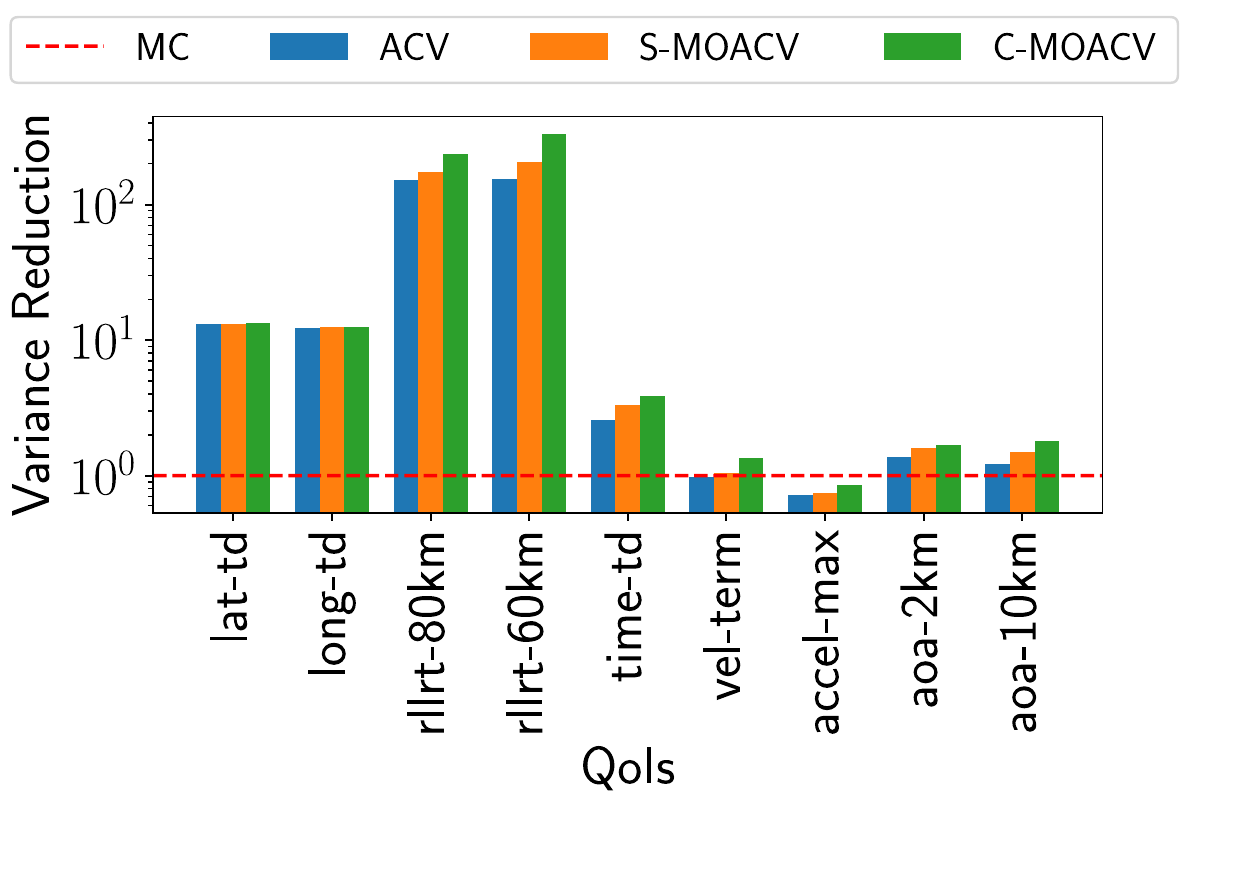}
    \caption{Mean Estimation}
    \label{ADm}
\end{subfigure}
\begin{subfigure}{0.45\textwidth}
    \includegraphics[trim={0 50pt 0 0},width=1\linewidth]{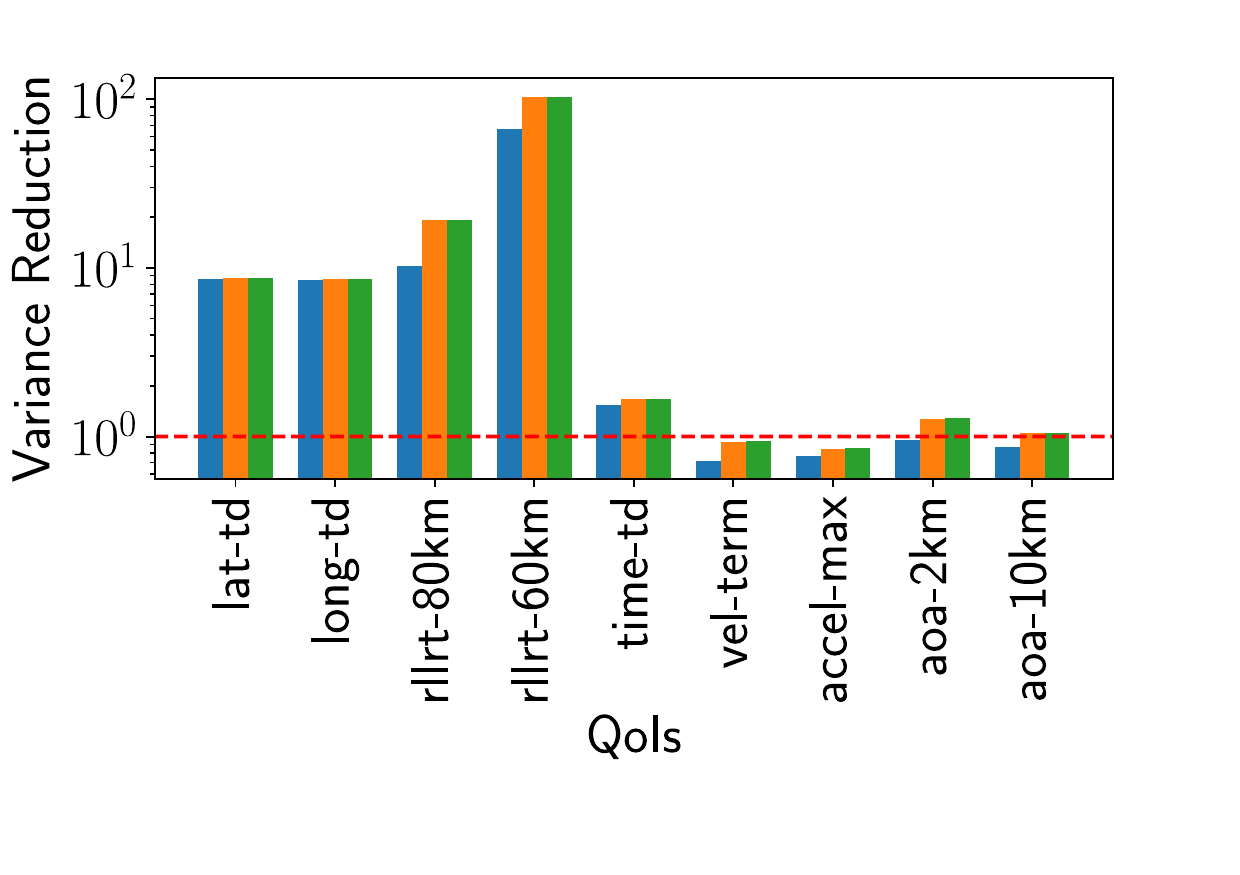}
    \caption{Variance Estimation}
    \label{ADv}
\end{subfigure}
\caption{Variance reduction compared to MC estimation for each trajectory QoI in Section \ref{sec:MVadept}. The S-MOACV and C-MOACV estimators outperform the individual ACV estimators at every QoI. }
\label{AD_orig}
\end{figure}

\begin{figure}[h]
\centering
\begin{subfigure}{0.45\textwidth}
    \includegraphics[trim={10pt 10pt 10pt 0},width=1\linewidth]{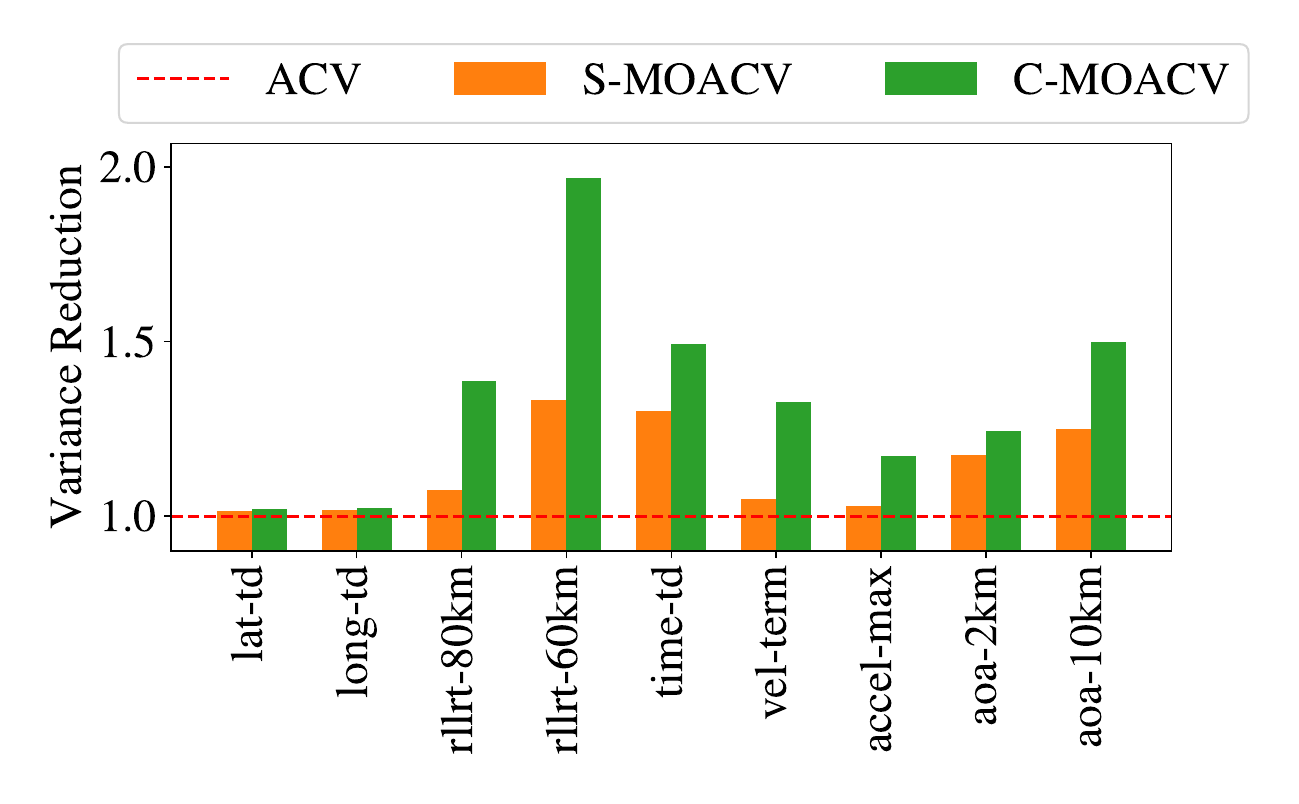}
    \caption{Mean Estimation}
    \label{ADmACV}
\end{subfigure}
\begin{subfigure}{0.45\textwidth}
    \includegraphics[trim={10pt 15pt 0 0},width=1\linewidth]{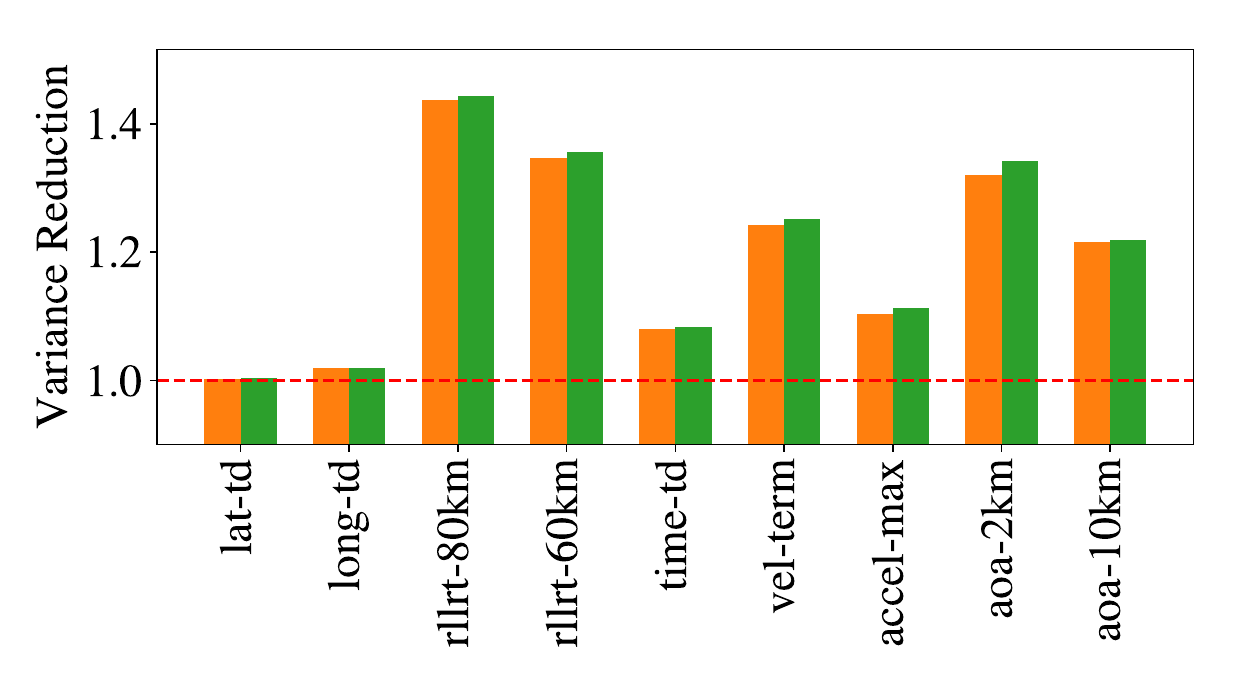}
    \caption{Variance Estimation}
    \label{ADvACV}
\end{subfigure}
\caption{\Ra{Variance reduction compared to ACV estimation for each trajectory QoI in Section \ref{sec:MVadept}. The S-MOACV and C-MOACV estimators outperform the individual ACV estimators at every QoI. The C-MOACV mean estimator for the \textit{rllrt-60km} QoI has nearly 2 times smaller estimator variance than the ACV estimator.} }
\label{AD_ACV}
\end{figure}

\subsection{Sensitivity Analysis}
\label{sec:adeptSA}
In this section, the C-MOACV estimator performs a sensitivity analysis on the roll rate at 80 km by simultaneously estimating the Sobol indices of three input variables, the initial roll rate (\textit{IRR}) and two uncertainties in the vehicle's moment of inertia, \textit{MOI1} and \textit{MOI2}. The input variables and QoI were chosen to demonstrate the ME variance estimation for variables with both high and low Sobol indices. The C-MOACV estimator contains 4 outputs, the 3 ME variance estimators and 1 total variance ($\var{\textrm{``\textit{rllrt-80km}"}}$) estimator. The Sobol indices are then constructed by dividing the ME variance estimate by the total variance estimate from the C-MOACV estimator. The pilot covariances are estimated using 5,000 pilot samples. The C-MOACV and ACV estimators use the ACV-IS sampling scheme, and the sample allocation was found by minimizing \eqref{eq:varOpt} with the C-MOACV variance subject to a budget of 10,000 seconds. The sample allocation is 21, 99, 200, and 7359 samples for the full-physics, reduced-physics, coarse time-step, and machine learning models respectively. 

Figure \ref{SA_Var} shows the variance reduction compared to MC estimation for the individual ACV estimators and the C-MOACV estimator.
The C-MOACV variance reduction is a median 300\% larger than ACV estimation. Since the ME variance estimator is only defined for scalar functions, the C-MOACV estimator can only outperform the ACV estimator if there are large correlations between the MC ME variance and total variance estimators. In Figure \ref{SA_corr}, the correlations between 10,000 MC estimators for the variance and ME variance can be seen. Large correlations are shown between the ME variance estimates and the total variance estimates. The MOACV estimator is able to extract these high correlations to reduce the variance of each of the estimators.

\begin{figure}[h]
\centering
\begin{subfigure}[t]{0.48\textwidth}
    \centering
\includegraphics[width=1\textwidth]{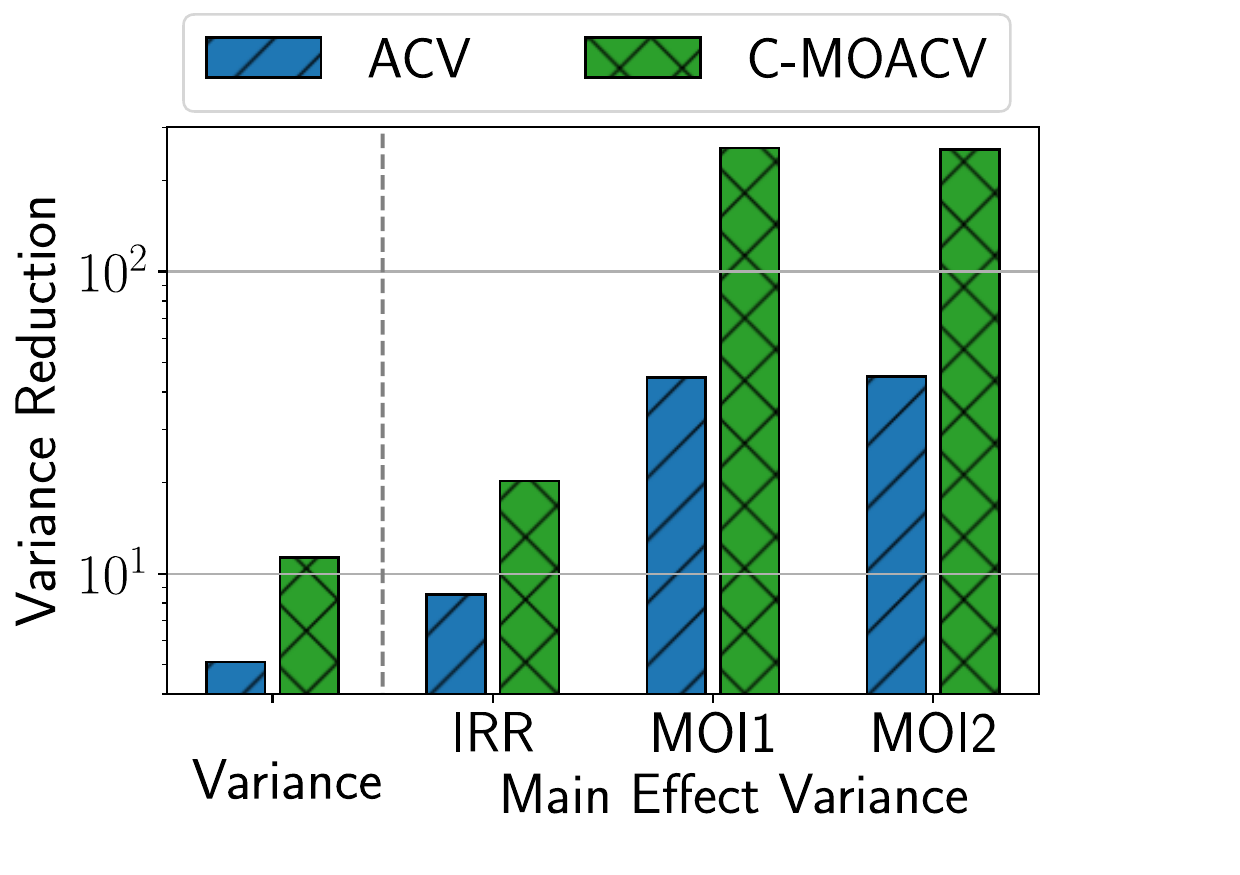}
    \caption{Variance Reduction compared to MC for each output.}
    \label{SA_Var}
\end{subfigure}
\hspace{10pt}
\begin{subfigure}[t]{0.48\textwidth}
    \centering
\includegraphics[trim={0 20pt 75pt 0},width=1\linewidth]{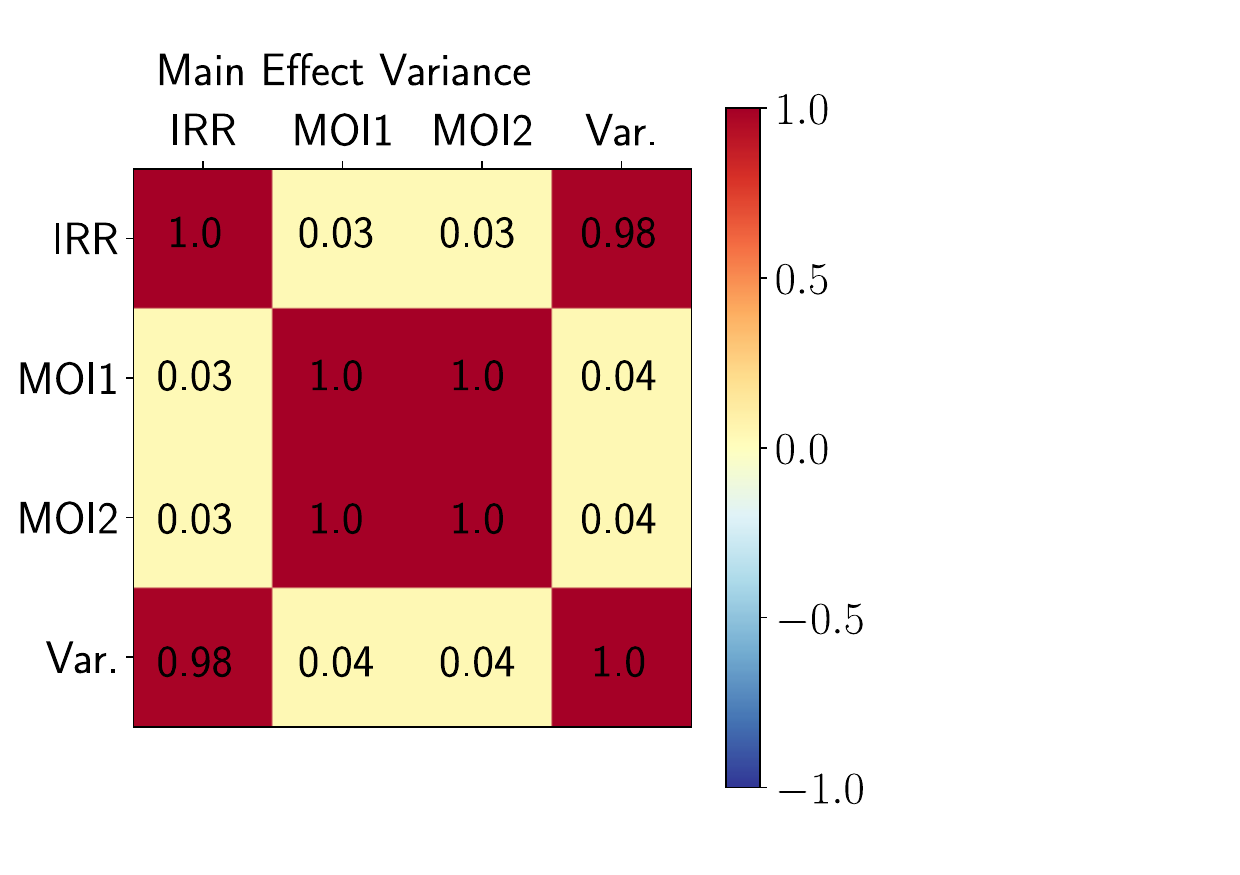}
\caption{Correlations between MC estimators for each ME variance and total variance.}
    \label{SA_corr}
\end{subfigure}
\caption{Sensitivity analysis for three EDL model inputs in Section \ref{sec:adeptSA}. C-MOACV extracts correlations between MC estimators to achieve larger variance reduction compared to individual ACV estimation. }
\label{SA_bars}
\end{figure}

We now form the Sobol index estimates by dividing the ME variance estimate by the total variance estimate for the MC, ACV, and MOACV estimators. To measure the distribution of Sobol index estimates, 10,000 estimators are constructed with random realizations of input samples. The distribution of the 10,000 Sobol index estimates is seen in Figure \ref{SA_box}. The Sobol index of \textit{IRR} is around $0.99$, and the \textit{MOI} Sobol estimates are close to $0$. Since a Sobol index is the percentage of the model's variance for an input, almost 100\% of the QoI's (roll rate at 80 km) variance is attributed to the initial roll rate. This high percentage explains the high correlation between the \textit{IRR} ME and the variance estimates seen in Figure \ref{SA_corr}. Conversely, the approximately 0\% Sobol index of the \textit{MOI}s explains the low correlation between their ME and total variance estimates.

\begin{figure}[h]
\centering
\begin{subfigure}[t]{0.48\textwidth}
    \centering
    \includegraphics[trim={0 50pt 0 0},width=1\linewidth]{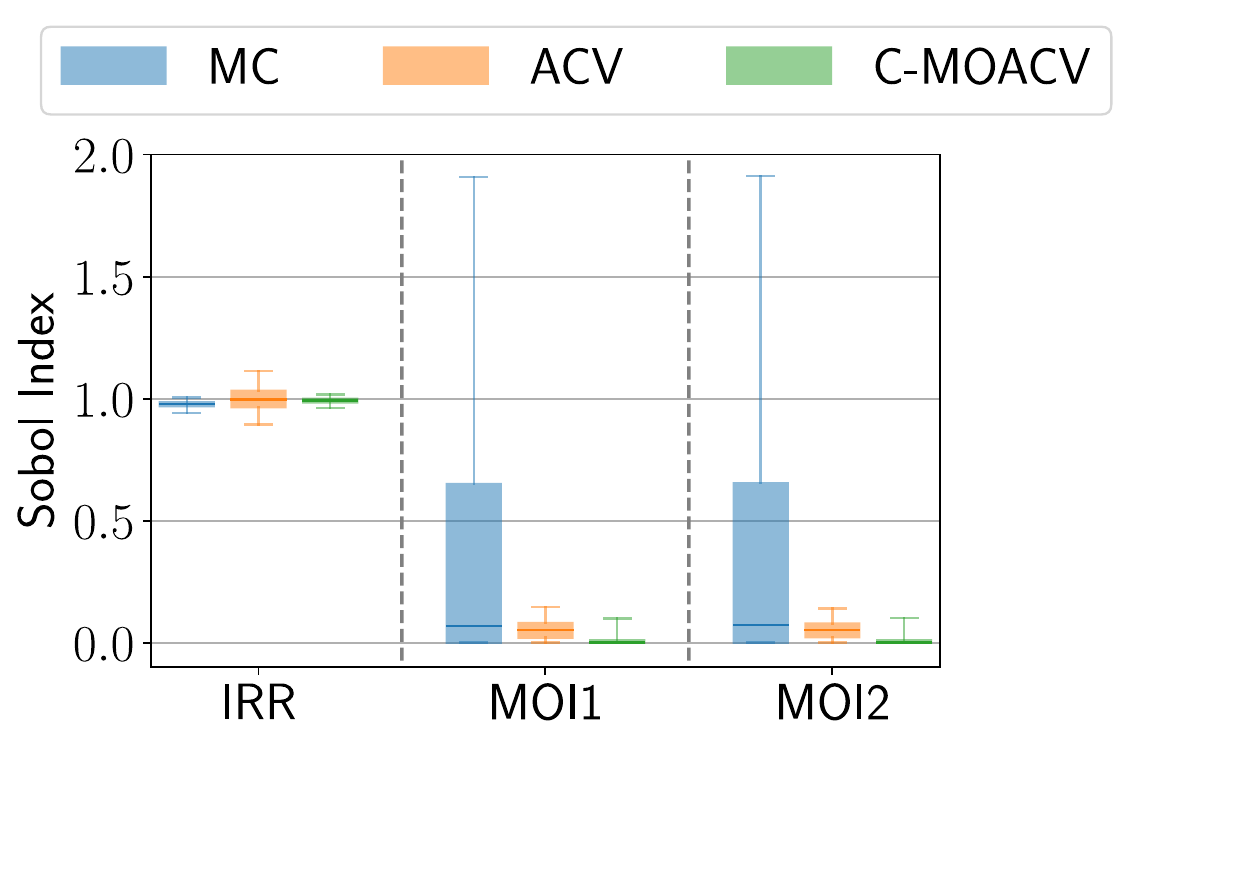}
    \caption{Sobol index estimates for each estimator.}
    \label{SA_box}
\end{subfigure}
\hspace{10pt}
\begin{subfigure}[t]{0.48\textwidth}
    \centering
    \includegraphics[trim={0 50pt 0 0},width=1\textwidth]{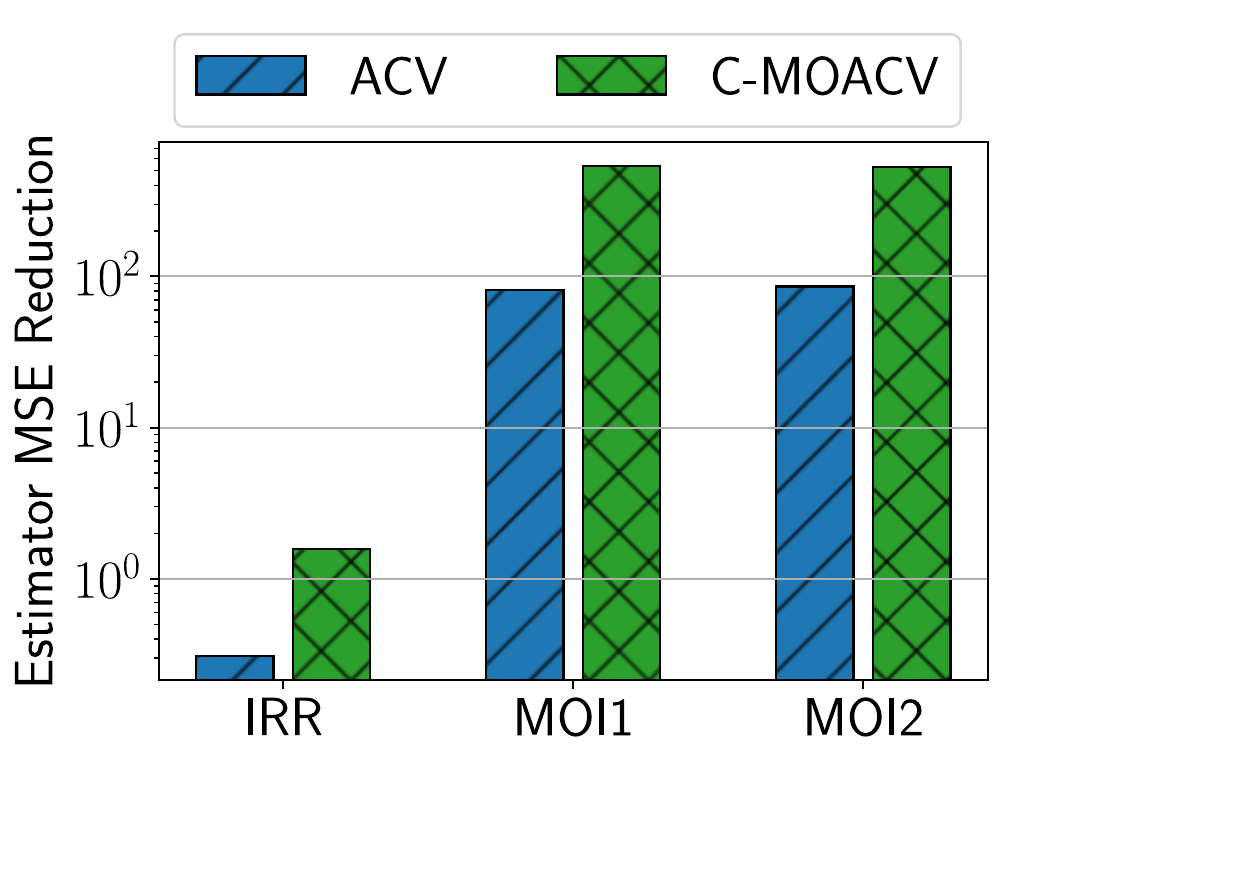}
    \caption{Mean Squared Error reduction compared to MC for each QoI.}
    \label{SA_SI}
\end{subfigure}
\caption{Sobol index estimation for the EDL model in Section \ref{sec:adeptSA}. The ME variance estimates are divided by the total variance estimates to form the Sobol indices. The C-MOACV estimator significantly reduces the MSE of Sobol index estimation compared to ACV estimators. }
\label{SA_comb}
\end{figure}

While Figure \ref{SA_box} displays the qualitative difference between ACV and C-MOACV Sobol index estimation, we now directly compare the error in each of the Sobol index estimates. Since the Sobol index estimates are found by dividing two estimators, the resulting Sobol index estimates are biased. Instead of measuring the variance of estimation, the mean squared error (MSE) is calculated to now consider the bias from the truth, which was calculated using 10,000 high-fidelity samples. The MSE for the MC estimates are divided by the MSE for the multi-fidelity estimates to calculate the MSE reduction. In Figure \ref{SA_SI}, the MSE reduction compared to MC estimation of the Sobol index estimates is seen. The C-MOACV estimator reduces the MSE in all Sobol index estimates compared to ACV estimation. The MSE reduction for the C-MOACV estimates is a median 515\% greater than the MSE reduction for ACV estimates. This section validates that the MOACV estimator can be used for more accurate sensitivity analysis and provides an example that the MOACV estimator can achieve further variance reduction by using the correlation between estimators.

\section{Conclusion}
In this work, we have introduced closed-form expressions for the covariance between MC estimators of multi-output functions for a variety of statistics. We have also used these results in the ACV context to construct the multi-output ACV estimator. The introduced multi-fidelity estimators include the vector-valued mean and variance estimators that utilize the correlations between models, outputs, and estimators to improve variance reduction. For sensitivity analysis, the MOACV estimator is demonstrated to simultaneously estimate the variance and multiple ME variances for more accurate Sobol indices. 

Numerous results demonstrate that the correlations between model fidelities, model outputs, and estimators can be extracted to provide further variance reduction. In the synthetic numerical results, the C-MOACV estimator is able to achieve up to 183 times larger variance reduction compared to a traditional ACV estimator. The MOACV estimator is also applied to an entry, descent, and landing application to more accurately estimate 9 QoIs given a fixed computational budget. Further, a variance-based sensitivity analysis is performed to illustrate the expected improved accuracy of the C-MOACV estimator. The C-MOACV estimator is able to increase the MSE reduction of Sobol index estimates by up to 557\% compared to traditional ACV estimation. In summary, multi-output estimation techniques are able to significantly outperform traditional ACV methods when high correlations exist between model outputs and estimators.

In future work, the ME variance estimator can be extended to vector-valued functions. Since the variance estimator has already been defined for multiple outputs, the extension to the ME variance estimator will be able to take advantage of correlations between other model outputs. Extending the estimator to vector-valued functions would enable the sensitivity analysis to be performed on multiple model outputs and inputs simultaneously. Additionally, the introduced estimator covariances can be applied to other multi-fidelity sampling strategies, such as the MLBLUE estimator for multi-statistic estimation. New strategies can be introduced to find the optimal number of pilot samples that minimize the total model evaluation cost, such as multi-arm bandit learning approaches \cite{BLMF}. Finally, covariance estimation techniques can be used to mitigate the loss of positive-definiteness by estimating on a covariance manifold \cite{PDCov}.

\clearpage
\bibliographystyle{siamplain}
\bibliography{references}

\setcounter{equation}{0}
\renewcommand{\theequation}{\thesection.\arabic{equation}}

\clearpage

\appendix

\Ra{
\section{Optimal ACV Weights - Trace Minimization}
\label{sec:ACVTrace}
Rubinstein and Marcus \cite{MVCV} provide a proof that the CV weights \eqref{eq:OptWei} are optimal when minimizing the determinant of the estimator variance \eqref{eq:MFempvar}. We provide a parallel proof here to demonstrate that \eqref{eq:OptWei} is also optimal when minimizing the trace of the estimator variance.

We show that the optimal weights $\mat{\alpha}^*$ that minimizes the trace of the estimator variance, $\textrm{Tr}(\mathbb{V}ar[{\Tilde{\vec{Q}}}])$, is equivalent to $\mat{\alpha}^*=-\mathbb{C}ov[\bold Q(\set{Z}_0), \un{\vec{\Delta}}(\set{Z}^*)]\mathbb{V}ar[\un{\vec{\Delta}}(\set{Z}^*)]^{-1}$. 

To begin, let $\mat{\alpha}^* = \mat{\alpha} + \mat{D}$ denote the optimal weights as a sum of two matrices where $\mat{\alpha} = -\mathbb{C}ov[\bold Q(\set{Z}_0), \un{\vec{\Delta}}(\set{Z}^*)]\mathbb{V}ar[\un{\vec{\Delta}}(\set{Z}^*)]^{-1}$. Thus, if $\mat{D}=\mat{0}$, then $\mat{\alpha}$ is optimal, $\mat{\alpha}^*=\mat{\alpha}$. We now show that $\mat{D}=\mat{0}$. 

To minimize the trace, we show that $\mat{D}$ must be equal to $\mat{0}$. From Equation \eqref{eq:MFempvar},
\begin{align}
\scalemath{0.9}{
    \mathbb{V}ar[\Tilde{\bold Q}\Ra{(\set{Z})}]} &= \scalemath{0.9}{
     \mathbb{V}ar[\bold Q\Ra{(\set{Z}_0)}] + \mat{\alpha}^* \mathbb{V}ar[\un{\vec{\Delta}}\Ra{(\set{Z}^*)}]\mat{\alpha}^{*T} + \mathbb{C}ov[\bold Q\Ra{(\set{Z}_0)}, \un{\vec{\Delta}}\Ra{(\set{Z}^*)}]\mat{\alpha}^{*T} + \mat{\alpha}^*\mathbb{C}ov[\un{\vec{\Delta}}\Ra{(\set{Z}^*)},\bold Q\Ra{(\set{Z}_0)}].} \nonumber\\
     &= \mathbb{V}ar[\bold Q\Ra{(\set{Z}_0)}] + \mathbb{C}ov[\bold Q(\set{Z}_0), \un{\vec{\Delta}}(\set{Z}^*)]\mathbb{V}ar[\un{\vec{\Delta}}(\set{Z}^*)]^{-1}\mathbb{C}ov[\bold Q(\set{Z}_0), \un{\vec{\Delta}}(\set{Z}^*)]^T \nonumber\\
     &\quad -\mathbb{C}ov[\bold Q(\set{Z}_0), \un{\vec{\Delta}}(\set{Z}^*)]\mat{D}^T - \mat{D}\mathbb{C}ov[\bold Q(\set{Z}_0), \un{\vec{\Delta}}(\set{Z}^*)]^T + \mat{D}\mathbb{V}ar[\un{\vec{\Delta}}\Ra{(\set{Z}^*)}]\mat{D}^T \nonumber\\
     &\quad - 2\mathbb{C}ov[\bold Q\Ra{(\set{Z}_0)}, \un{\vec{\Delta}}\Ra{(\set{Z}^*)}]\mathbb{V}ar[\un{\vec{\Delta}}(\set{Z}^*)]^{-1}\mathbb{C}ov[\bold Q(\set{Z}_0), \un{\vec{\Delta}}(\set{Z}^*)]^T  \nonumber\\
     &\quad + \mathbb{C}ov[\bold Q\Ra{(\set{Z}_0)}, \un{\vec{\Delta}}\Ra{(\set{Z}^*)}]\mat{D}^T  + \mat{D}\mathbb{C}ov[\un{\vec{\Delta}}\Ra{(\set{Z}^*)},\bold Q\Ra{(\set{Z}_0)}]. \nonumber \\
     &= \left[\mathbb{V}ar[\bold Q\Ra{(\set{Z}_0)}] - \mathbb{C}ov[\bold Q(\set{Z}_0), \un{\vec{\Delta}}(\set{Z}^*)]\mathbb{V}ar[\un{\vec{\Delta}}(\set{Z}^*)]^{-1}\mathbb{C}ov[\bold Q(\set{Z}_0), \un{\vec{\Delta}}(\set{Z}^*)]^T \right] \nonumber\\
     &\quad + \mat{D}\mathbb{V}ar[\un{\vec{\Delta}}\Ra{(\set{Z}^*)}]\mat{D}^T \nonumber\\
     &= \mat{C} + \mat{D}\mathbb{V}ar[\un{\vec{\Delta}}\Ra{(\set{Z}^*)}]\mat{D}^T 
\end{align}
where $\mat{C}=\mathbb{V}ar[\bold Q\Ra{(\set{Z}_0)}] - \mathbb{C}ov[\bold Q(\set{Z}_0), \un{\vec{\Delta}}(\set{Z}^*)]\mathbb{V}ar[\un{\vec{\Delta}}(\set{Z}^*)]^{-1}\mathbb{C}ov[\bold Q(\set{Z}_0), \un{\vec{\Delta}}(\set{Z}^*)]^T$. Thus, the estimator variance $\mathbb{V}ar[\Tilde{\bold Q}\Ra{(\set{Z})}]$ can be expressed as the sum of two positive semi-definite matrices. Using the properties of trace of two matrices we have
\begin{align}
    \textrm{Tr}(\mathbb{V}ar[\Tilde{\bold Q}\Ra{(\set{Z})}]) &= \textrm{Tr}(\mat{C} + \mat{D}\mathbb{V}ar[\un{\vec{\Delta}}\Ra{(\set{Z}^*)}]\mat{D}^T) \nonumber\\
    &=\textrm{Tr}(\mat{C}) + \textrm{Tr}(\mat{D}\mathbb{V}ar[\un{\vec{\Delta}}\Ra{(\set{Z}^*)}]\mat{D}^T)
\end{align}
where $\textrm{Tr}(\mathbb{V}ar[\Tilde{\bold Q}\Ra{(\set{Z})}])$ is minimized if $\mat{D}=\mat{0}$ because the trace of a positive semi-definite matrix is always non-negative. Therefore, since  $\mat{D}=0$, the optimal weights must be $\mat{\alpha}^*=-\mathbb{C}ov[\bold Q(\set{Z}_0), \un{\vec{\Delta}}(\set{Z}^*)]\mathbb{V}ar[\un{\vec{\Delta}}(\set{Z}^*)]^{-1}$, as was to be shown.

}

\Rb{
\section{Derivations of Previous Results}
\label{sec:prevworks}
We show that our expressions can be simplified to previously derived results. First, we wish to show that the covariance between two mean estimators for MFMC scalar mean estimation is
\begin{align}
\cov{\vec{Q}_i(\set{Z}_i^*),\vec{Q}_j(\set{Z}_j^*)} &= \frac{1}{\max\{ |\set{Z}_i^*|,|\set{Z}_j^*| \}}\cov{f_i,f_j}.
\end{align}
as was declared in \cite[Eq. 3.4]{MFMC1}.

Since MFMC declares the sample sets to be $|\set{Z}_i^*\cap \set{Z}_j^*| = \min\{|\set{Z}_i^*|, |\set{Z}_j^*|\}$ by construction, Proposition \ref{meancov} gives
\begin{align}
    \cov{\vec{Q}_i(\set{Z}_i^*),\vec{Q}_j(\set{Z}_j^*)} &= \frac{|\set{Z}_i^*\cap \set{Z}_j^*|}{|\set{Z}_i^*||\set{Z}_j^*| }\cov{f_i,f_j} = \frac{\min\{|\set{Z}_i^*|, |\set{Z}_j^*|\}}{|\set{Z}_i^*||\set{Z}_j^*| }\cov{f_i,f_j}\\
    &= \frac{1}{\max\{ |\set{Z}_i^*|,|\set{Z}_j^*| \}}\cov{f_i,f_j}.
\end{align}
as was to be shown.

Now, we wish to show that the covariance between two variance estimators for MFMC scalar variance estimation
\begin{align}
    \mathbb{C}ov[\bold Q_i(\set{Z}_i^*), \bold Q_j(\set{Z}_j^*)] &= \scalemath{0.8}{\frac{1}{|\set{Z}_i^*|}\left[\cov{\left(f_i(\vec{z})-\mathbb{E}[f_i(\vec{z})]\right)^{ 2},\left(f_j(\vec{z})-\mathbb{E}[f_j(\vec{z})]\right)^{ 2}} + \frac{2}{|\set{Z}_i^*|-1} \mathbb{C}ov[f_i(\vec{z}), f_j(\vec{z})]^{2}\right]}
\end{align}
as was declared in \cite[Eq. 14]{MFMC}.

Using the MFMC sampling scheme, let $i\geq j$ such that $|\set{Z}_i^*\cap \set{Z}_j^*| = |\set{Z}_j^*|$. Proposition \ref{var_cov} gives
\begin{align}
    \mathbb{C}ov[\bold Q_i(\set{Z}_i^*), \bold Q_j(\set{Z}_j^*)] &= \frac{|\set{Z}_j^*|(|\set{Z}_j^*|-1)}{|\set{Z}_i^*|(|\set{Z}_i^*|-1)|\set{Z}_j^*|(|\set{Z}_j^*|-1)} \mat{V}_{ij} + \frac{|\set{Z}_j^*|}{|\set{Z}_i^*||\set{Z}_j^*|}\mat{W}_{ij} \\
    &= \frac{1}{|\set{Z}_i^*|(|\set{Z}_i^*|-1)} \mat{V}_{ij} + \frac{1}{|\set{Z}_i^*|}\mat{W}_{ij}.
\end{align}
For the scalar variance case,
\begin{align}
    {\mat{V}}_{ij} &= \scalemath{0.9}{\mathbb{C}ov[f_i(\vec{z}), f_j(\vec{z})]^{\otimes 2} + \left( \bold 1_1^T \otimes \mathbb{C}ov[f_i(\vec{z}), f_j(\vec{z})] \otimes \bold 1_1 \right) \circ \left(\bold 1_1 \otimes \mathbb{C}ov[f_i(\vec{z}), f_j(\vec{z})] \otimes \bold 1_1^T \right)} \\
    &= \mathbb{C}ov[f_i(\vec{z}), f_j(\vec{z})]^{2} +  \mathbb{C}ov[f_i(\vec{z}), f_j(\vec{z})]^{2} =  2\mathbb{C}ov[f_i(\vec{z}), f_j(\vec{z})]^{2} \\
    {\mat{W}}_{ij} &= \cov{\left(f_i(\vec{z})-\mathbb{E}[f_i(\vec{z})]\right)^{\otimes 2},\left(f_j(\vec{z})-\mathbb{E}[f_j(\vec{z})]\right)^{\otimes 2}} \\
    &= \cov{\left(f_i(\vec{z})-\mathbb{E}[f_i(\vec{z})]\right)^{ 2},\left(f_j(\vec{z})-\mathbb{E}[f_j(\vec{z})]\right)^{ 2}}
\end{align}
Thus,
\begin{align}
    \mathbb{C}ov[\bold Q_i(\set{Z}_i^*), \bold Q_j(\set{Z}_j^*)] &= \scalemath{0.8}{\frac{2}{|\set{Z}_i^*|(|\set{Z}_i^*|-1)} \mathbb{C}ov[f_i(\vec{z}), f_j(\vec{z})]^{2} + \frac{1}{|\set{Z}_i^*|}\cov{\left(f_i(\vec{z})-\mathbb{E}[f_i(\vec{z})]\right)^{ 2},\left(f_j(\vec{z})-\mathbb{E}[f_j(\vec{z})]\right)^{ 2}}} \\
    &= \scalemath{0.8}{\frac{1}{|\set{Z}_i^*|}\left[\cov{\left(f_i(\vec{z})-\mathbb{E}[f_i(\vec{z})]\right)^{ 2},\left(f_j(\vec{z})-\mathbb{E}[f_j(\vec{z})]\right)^{ 2}} + \frac{2}{|\set{Z}_i^*|-1} \mathbb{C}ov[f_i(\vec{z}), f_j(\vec{z})]^{2}\right]}
\end{align}
as was to be shown.
}

\Ra{
\section{Pilot sampling strategies}
\label{sec:PSstrat}
In this section, we compare pilot study techniques that introduce bias into the MOACV estimator by reusing the pilot samples within the estimator. The system is the same as in Section \ref{sec:percent} where we estimate the mean of the first output of $\vec{f}_0(x)$ using MC, ACV, MOACV, and two alternate estimators using all 8 outputs. 

First, the MOACV-Reuse estimator begins with a set of pilot samples to estimate the required covariances. These covariances are used to find the optimal sample allocation. New sets of samples are evaluated based on the sample allocation and the remaining budget. The pilot samples are then appended to the newly sampled sets and used for estimation. Using the pilot samples in the estimator creates bias, which is why we need to measure the mean squared error (MSE) of the estimator instead of just the variance.

Next we consider an adaptive scheme MOACV-Adapt that continues to re-estimate until a naive convergence criteria is met. Specifically, the MOACV-Adapt takes an independent set of pilot samples to estimate the pilot covariances. These covariances are used to find the optimal sample allocation. If the optimal sample allocation requests more high-fidelity samples than the number of pilot samples, we add more pilot samples. We repeat this until the optimal sample allocation stops requesting more high fidelity samples. We then evaluate low-fidelity samples with the remaining budget according to the optimal sample allocation. The pilot samples are then used in the estimation process. 

One important difference between the MOACV-Reuse and the MOACV-Adapt is that  MOACV-Adapt estimates with the optimal sample allocation that it finds. The MOACV-Reuse uses a modified version of the optimal sample allocation since the pilot samples are simply appended to the estimation sets.

Figure \ref{bias_study} displays the empirical estimator MSEs of 50 estimators across 25 sets of pilot samples at each budget percentage. The MOACV-Reuse strategy only outperforms the MOACV estimator at a large number of pilot samples. This is because more pilot samples are appended to the estimation sets to reduce the estimator variance. While this strategy does introduce bias to the estimator, it does not outweigh the variance in this example. The MOACV-Adapt outperforms the MOACV at every budget percent. When there are too few pilot samples, the MOACV-Adapt suggests adding more. When there are too many pilot samples, the MOACV-Adapt adds the pilot samples to the estimation sets. Again, this reduces the variance of the MOACV estimator, but also adds bias. In this example, the bias does not seem to play a significant role. The reduction in MSE seems small, but it suggests a potentially fruitful path for future work.
}

\begin{figure}
    \centering
    \includegraphics[trim={0 0 0 0},width=0.6\linewidth]{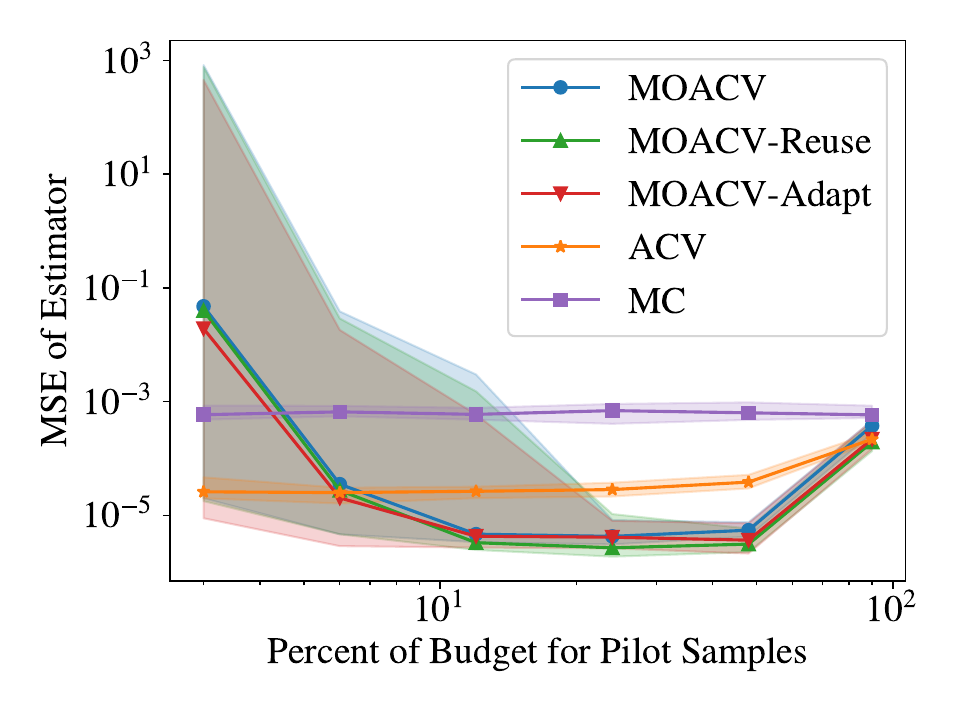}
    \captionof{figure}{\Ra{MSE of the estimators across various budget percentages allocated to the pilot samples in Appendix \ref{sec:PSstrat}. The shaded region represents 90\% of the uncertainty in the estimator variance over 25 sets of pilot samples. The MOACV-Adapt is able to outperform the MOACV estimator at every percentage by reducing the estimator variance. The introduced bias seems negligible in this example system.}}
    \label{bias_study}
\end{figure}
\FloatBarrier

\Ra{
\section{EDL Example Bias}
\label{sec:EDLbias}
As a reference we include the biases of the low-fidelity models in the EDL model. These are shown in Figure \ref{fig:EDLbias}.

\begin{figure}[h]
    \centering
    \includegraphics[trim={0 0 0 0},width=0.6\linewidth]{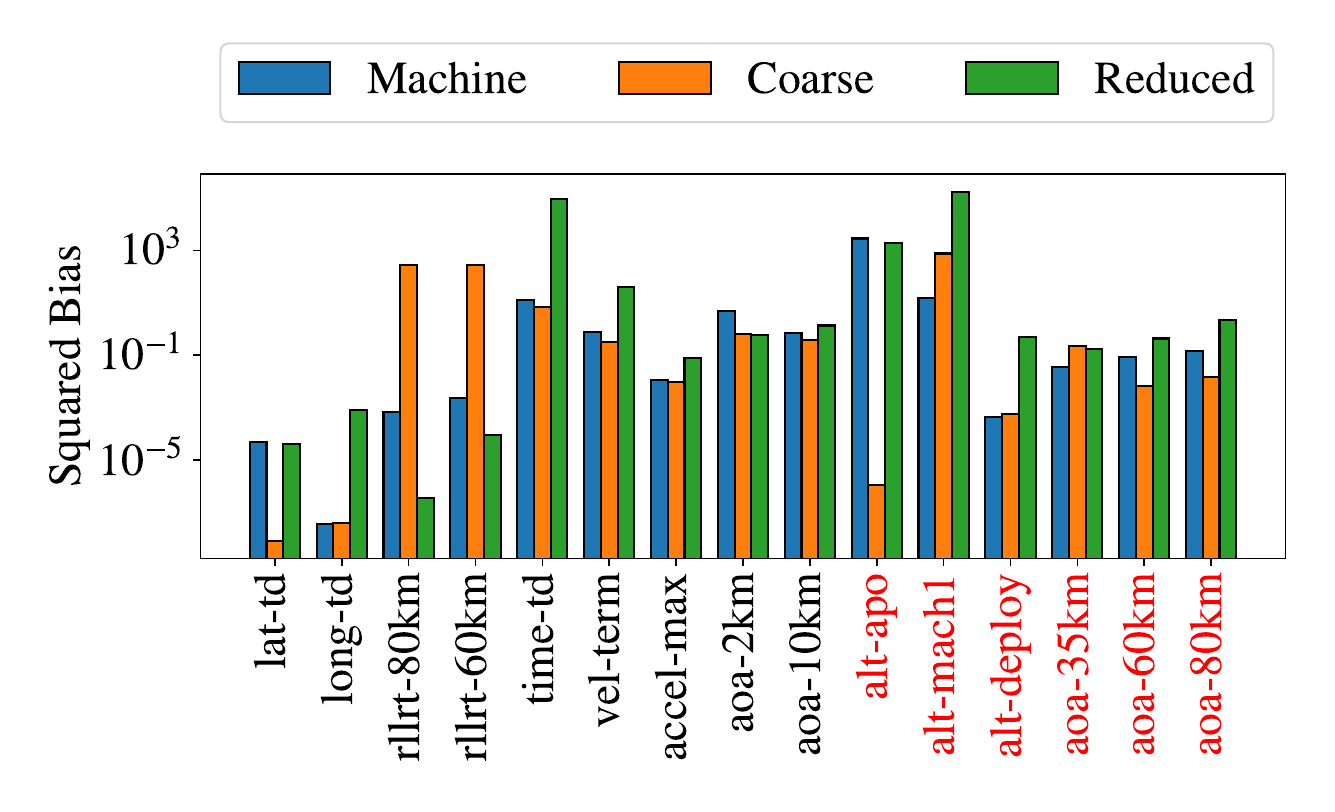}
    \captionof{figure}{\Ra{The squared biases of the low-fidelity models with respect to the high-fidelity model for each output in Section \ref{sec:application} are shown.}}
    \label{fig:EDLbias}
\end{figure}
}
\FloatBarrier

\Rb{
\section{EDL Pilot Study}
\label{sec:EDLPilot}
In this example, we consider a total budget of an equivalent 100,000 high-fidelity samples. Similarly to Section \ref{sec:MVadept}, we allocate an equivalent 60,000 high-fidelity samples to the pilot study budget. The remaining budget of 40,000 equivalent high-fidelity samples is used to estimate the means and variances of the QoIs with the same sample allocation ratios in Section \ref{sec:MVadept}. The variance reduction compared to MC estimation (using 100,000 high-fidelity samples) is shown in Figure \ref{fig:AD_PS}. This figure demonstrates that even considering the high pilot study cost, MOACV is able to outperform MC estimation for the QoIs that have high correlations with the low-fidelity models. Specifically for mean estimation, the log determinant of the mean-estimator variance is -93.86 for C-MOACV, -93.08 for S-MOACV , -90.65 for ACV, and -90.25 for MC estimation, where a more negative number corresponds to a smaller determinant. Clearly, the MOACV estimator is able to outperform ACV and MC estimation, even with a sample allocation that was optimized for ACV.

\begin{figure}[h]
\centering
\begin{subfigure}{0.45\textwidth}
    \includegraphics[trim={10pt 10pt 10pt 0},width=1\linewidth]{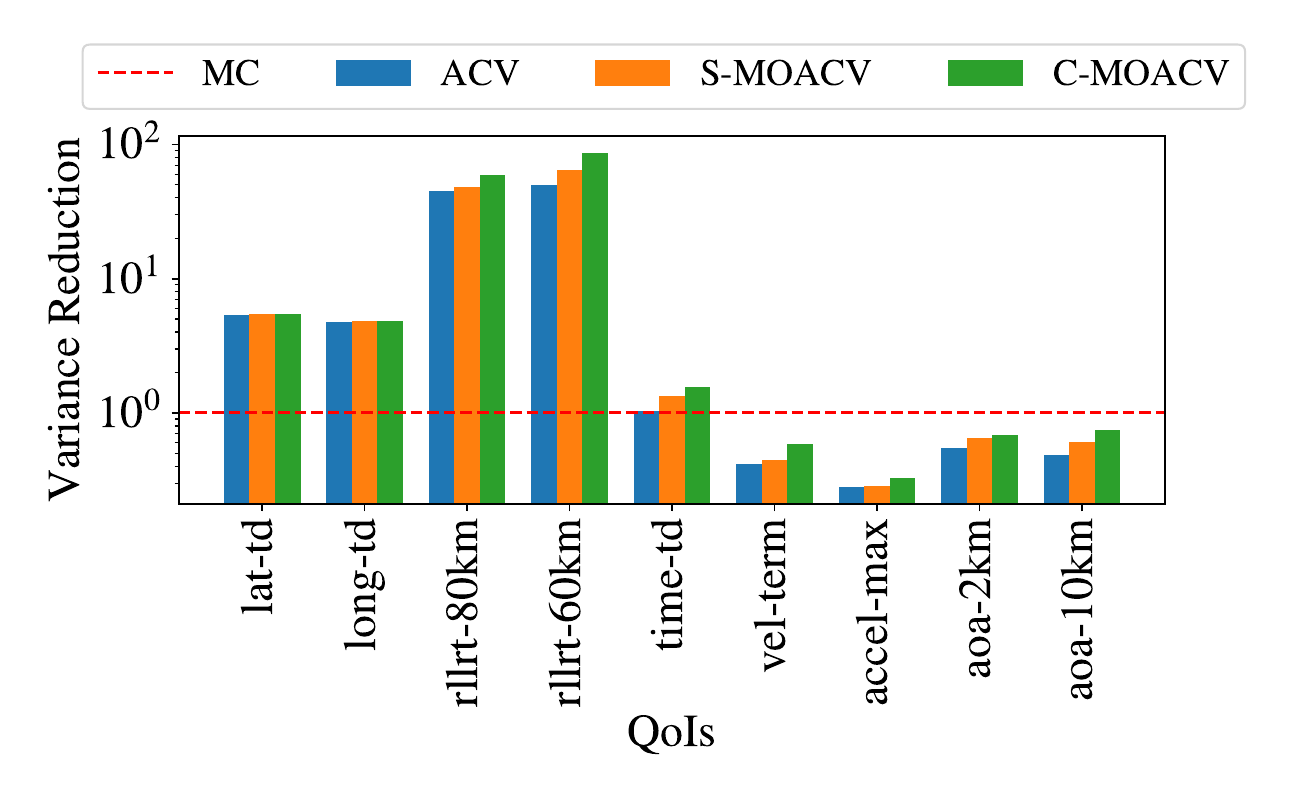}
    \caption{Mean Estimation}
    \label{fig:ADm_PS}
\end{subfigure}
\begin{subfigure}{0.45\textwidth}
    \includegraphics[trim={10pt 0 0 0},width=1\linewidth]{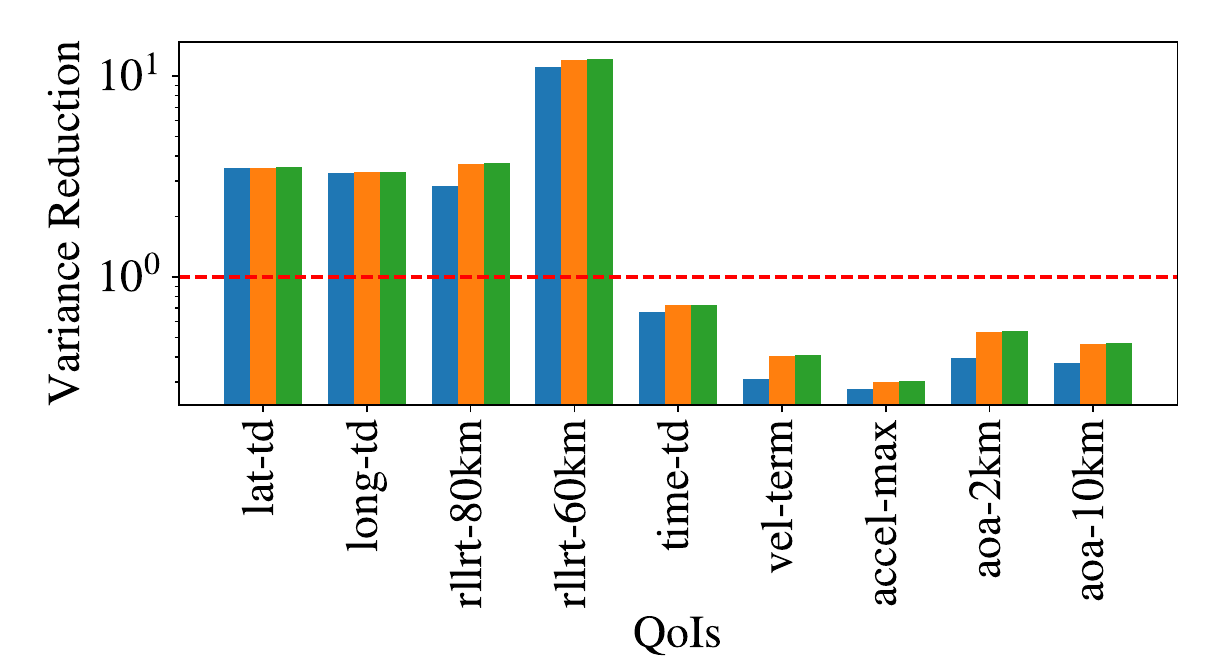}
    \caption{Variance Estimation}
    \label{fig:ADv_PS}
\end{subfigure}
\caption{\Rb{Variance reduction compared to MC estimation for each trajectory QoI in Appendix \ref{sec:EDLPilot}. MC estimation includes pilot study costs. The S-MOACV and C-MOACV estimators outperform the individual ACV estimators.} }
\label{fig:AD_PS}
\end{figure}
}

\FloatBarrier
\section{Three Useful Covariance Results}
\label{SubApp}
First, let $\vec{w},\vec{x}\in \reals^{D\times 1}$ be dependent random variables, and let $\vec{y}, \vec{z} \in \reals^{D\times 1}$ be independent random variables with respect to all other random variables. By linearity of expectation,
$\mathbb{C}ov[\vec{w}\otimes \vec{y}, \vec{x}\otimes \vec{z}] = \mathbb{C}ov[\vec{w}\otimes\mathbb{E}[\vec{y}], \vec{x}\otimes\mathbb{E}[\vec{z}]].$

Next, we wish to derive $\cov{\vec{w},\vec{x}}^{\otimes 2}$
\begin{align}
    \cov{\vec{w},\vec{x}}^{\otimes 2} &= \cov{\vec{w},\vec{x}}\otimes \cov{\vec{w},\vec{x}} \nonumber\\
    &= \left(\ee{\vec{w} \vec{x}^T} - \ee{\vec{w}}\ee{\vec{x}}\right) \otimes \left(\ee{\vec{w} \vec{x}^T} - \ee{\vec{w}}\ee{\vec{x}}\right) \nonumber \\
    &= \ee{\vec{w} \vec{x}^T}^{\otimes 2}- \ee{\vec{w} \vec{x}^T} \otimes \ee{\vec{w}}\ee{\vec{x}} -  \ee{\vec{w}}\ee{\vec{x}} \otimes \ee{\vec{w} \vec{x}^T} + \left(\ee{\vec{w}}\ee{\vec{x}}\right)^{\otimes 2}.
\end{align}
Finally, another useful covariance is
\reqnomode
\begin{align}
    &\left( \bold 1_D^T \otimes \mathbb{C}ov[\vec{w}, \vec{x}] \otimes \bold 1_D \right) \circ \left(\bold 1_D \otimes \mathbb{C}ov[\vec{w},\vec{x}] \otimes \bold 1_D^T \right) = \nonumber\\
    &\quad= \left( \bold 1_D^T \otimes (\ee{\vec{w} \vec{x}^T}-\ee{\vec{w}}\ee{\vec{x}^T}) \otimes \bold 1_D \right) \circ \left(\bold 1_D \otimes (\ee{\vec{w} \vec{x}^T}-\ee{\vec{w}}\ee{\vec{x}^T}) \otimes \bold 1_D^T \right) \nonumber\\
    &\quad = (\vec{1}_D^T \otimes \ee{\vec{w} \vec{x}^T} \otimes \vec{1}_D) \circ (\vec{1}_D \otimes \ee{\vec{w} \vec{x}^T} \otimes \vec{1}_D^T) \nonumber \tag{W}\\
    &\quad\quad - (\vec{1}_D^T \otimes \ee{\vec{w} \vec{x}^T} \otimes \vec{1}_D) \circ (\vec{1}_D \otimes \ee{\vec{w}} \ee{\vec{x}^T} \otimes \vec{1}_D^T) \nonumber \tag{X}\\
    &\quad\quad\quad - (\vec{1}_D^T \otimes \ee{\vec{w}}\ee{\vec{x}^T} \otimes \vec{1}_D) \circ (\vec{1}_D \otimes \ee{\vec{w} \vec{x}^T} \otimes \vec{1}_D^T) \nonumber \tag{Y}\\
    &\quad\quad\quad\quad + (\vec{1}_D^T \otimes \ee{\vec{w}}\ee{\vec{x}^T} \otimes \vec{1}_D) \circ (\vec{1}_D \otimes \ee{\vec{w}} \ee{\vec{x}^T} \otimes \vec{1}_D^T). \tag{Z} \nonumber 
\end{align}
\leqnomode
First, (X) is considered. Note that $\vec{A}\otimes \vec{B}^T = \vec{B}^T \otimes \vec{A}$ when $\vec{A}, \vec{B} \in \reals^{D\times 1}$. Now,
\begin{align}
    (\textrm{X}) &= - (\vec{1}_D^T \otimes \ee{\vec{w} \vec{x}^T} \otimes \vec{1}_D) \circ (\vec{1}_D \otimes \ee{\vec{x}^T} \otimes \ee{\vec{w}} \otimes \vec{1}_D^T) \nonumber \\
    &= - (\vec{1}_D^T \otimes \ee{\vec{w} \vec{x}^T} \otimes \vec{1}_D) \circ (\ee{\vec{x}^T}  \otimes \vec{1}_D \otimes \vec{1}_D^T \otimes \ee{\vec{w}} ) \nonumber \\
    &= -\ee{\vec{x}^T} \otimes \ee{\vec{w} \vec{x}^T} \otimes \ee{\vec{w}}
\end{align}
using Kronecker and Hadamard mixed product properties. Similarly, 
\begin{align}
    (\textrm{Y}) &= - (\ee{\vec{w}} \otimes \vec{1}_D^T \otimes \vec{1}_D \otimes \ee{\vec{x}^T}) \circ (\vec{1}_D \otimes \ee{\vec{w} \vec{x}^T} \otimes \vec{1}_D^T) \nonumber \\
    &= -\ee{\vec{w}} \otimes \ee{\vec{w} \vec{x}^T} \otimes \ee{\vec{x}^T}.
\end{align}
Finally, consider (Z)
\begin{align}
    (\textrm{Z}) &= \left[(\vec{1}_D \otimes \ee{\vec{x}^T}) \otimes (\ee{\vec{w}} \otimes \vec{1}_D^T)\right] \circ \left[(\ee{\vec{w}} \otimes \vec{1}_D^T ) \otimes ( \vec{1}_D  \otimes \ee{\vec{x}^T})\right] \nonumber \\
    &= \left[(\vec{1}_D \otimes \ee{\vec{x}^T}) \circ (\ee{\vec{w}} \otimes \vec{1}_D^T)\right] \otimes \left[(\ee{\vec{w}} \otimes \vec{1}_D^T ) \circ ( \vec{1}_D  \otimes \ee{\vec{x}^T})\right] \nonumber \\
    &= \left[(\vec{1}_D \circ \ee{\vec{w}}) \otimes (\ee{\vec{x}^T} \circ \vec{1}_D^T)\right] \otimes \left[(\ee{\vec{w}} \circ  \vec{1}_D) \otimes ( \vec{1}_D^T  \circ \ee{\vec{x}^T})\right] \nonumber \\
    &= \ee{\vec{w}}\otimes \ee{\vec{x}^T} \otimes \ee{\vec{w}} \otimes \ee{\vec{x}^T} \nonumber \\
    &= \left( \ee{\vec{w}}\ee{\vec{x}^T} \right)^{\otimes 2}.
\end{align}
Therefore, by combining (W), (X), (Y), and (Z)
\begin{align}
&\left( \bold 1_D^T \otimes \mathbb{C}ov[\vec{w}, \vec{x}] \otimes \bold 1_D \right) \circ \left(\bold 1_D \otimes \mathbb{C}ov[\vec{w},\vec{x}] \otimes \bold 1_D^T \right) \nonumber =\\
    &\quad= (\mathbb{E}[\vec{w}]\mathbb{E}[j^T])^{\otimes 2} + \left( \bold 1_D^T \otimes \mathbb{E}[\vec{w} \vec{x}^T] \otimes \bold 1_D \right) \circ \left(\bold 1_D \otimes \mathbb{E}[\vec{w} \vec{x}^T] \otimes \bold 1_D^T \right)  \nonumber\\
    &\quad\quad -\mathbb{E}[\vec{w}] \otimes \mathbb{E}[\vec{w} \vec{x}^T] \otimes \mathbb{E}[\vec{x}^T] - \mathbb{E}[\vec{x}^T]\otimes \mathbb{E}[\vec{w} \vec{x}^T] \otimes \mathbb{E}[\vec{w}].
\end{align}

\section{Proof of Proposition \ref{var_cov}}
\label{appendix:var_cov}
Here, the covariance between two variance estimators \eqref{eq:varianceMC} is derived. For notation, let $f_i^{(s)} \equiv \vec{f}_i(\vec{n}^{(s)})$ and $f_j^{(u)} \equiv \vec{f}_j(\vec{m}^{(u)})$, and
\begin{align}
\scalemath{0.9}{
    \mathbb{C}ov[\vec{Q}_{i}(\set{N}), \vec{Q}_{j}(\set{M})]} &= \scalemath{0.9}{\mathbb{C}ov\left[\frac{1}{2N(N-1)}\sum^N_s\sum^N_t \left( f_i^{(s)}-f_i^{(t)} \right)^{\otimes 2}, \frac{1}{2M(M-1)}\sum^M_u\sum^M_v \left( f_j^{(u)}-f_j^{(v)} \right)^{\otimes 2}\right]} \nonumber\\
     &= \scalemath{0.9}{\frac{1}{4N(N-1)M(M-1)}\sum_s^N\sum_t^N\sum_u^M\sum_v^M \mathbb{C}ov\left[ \left( f_i^{(s)}-f_i^{(t)} \right)^{\otimes 2}, \left( f_j^{(u)}-f_j^{(v)} \right)^{\otimes 2}\right]}.
    \label{4sum}
\end{align}
The covariance in the quadruple sum is zero when the input samples are not shared between fidelities ($s\neq t \neq u \neq v$). Also, when $s=t$, $\left( f_i^{(s)}-f_i^{(s)} \right)^{\otimes 2}=0$ which means the covariance is similarly zero. The same is true when $u=v$. We can break the rest of the nonzero terms of the quadruple sum into 6 cases that each correspond to different combinations of shared inputs. Let $P=|\set{P}| = |\set{N}\cap \set{M}|$, $N=|\set{N}|$, and $M=|\set{M}|$. Thus
{\footnotesize
\reqnomode
\begin{align}
    \mathbb{C}ov&[\vec{Q}_{i}(\set{N}), \vec{Q}_{j}(\set{M})] = \frac{1}{4N(N-1)M(M-1)} \Bigg[ \nonumber\\
    &
    \sum_{s\in\set{P}}\sum_{t\in \set{N}\setminus \{s\}}\sum_{v\in \set{M}\setminus \{s,t\}} \mathbb{C}ov\left[ \left( f_i^{(s)}-f_i^{(t)} \right)^{\otimes 2}, \left( f_j^{(s)}-f_j^{(v)} \right)^{\otimes 2}\right] \quad\quad\textrm{when $s=u$ and $t \neq v$}\nonumber\tag{A}\\
    &\quad
    + \sum_{t\in\set{P}}\sum_{s\in \set{N}\setminus \{t\}}\sum_{u\in \set{M}\setminus \{s,t\}} \mathbb{C}ov\left[ \left( f_i^{(s)}-f_i^{(t)} \right)^{\otimes 2}, \left( f_j^{(u)}-f_j^{(t)} \right)^{\otimes 2}\right]  \quad\quad\textrm{when $s\neq u$ and $t=v$}\nonumber\tag{B}\\
    &\quad+ \sum_{s\in\set{P}}\sum_{t\in \set{N}\setminus \{s\}}\sum_{u\in \set{M}\setminus \{s,t\}} \mathbb{C}ov\left[ \left( f_i^{(s)}-f_i^{(t)} \right)^{\otimes 2}, \left( f_j^{(u)}-f_j^{(s)} \right)^{\otimes 2}\right] \quad\quad\textrm{when $s=v$ and $t\neq u$}\nonumber\tag{C}\\
    &\quad + \sum_{t\in\set{P}}\sum_{s\in \set{N}\setminus \{t\}}\sum_{v\in \set{M}\setminus \{s,t\}} \mathbb{C}ov\left[ \left( f_i^{(s)}-f_i^{(t)} \right)^{\otimes 2}, \left( f_j^{(t)}-f_j^{(v)} \right)^{\otimes 2}\right] \quad\quad\textrm{when $s\neq v$ and $t=u$}\nonumber\tag{D}\\ 
    &\quad + \sum_{s\in\set{P}}\sum_{t\in \set{P}\setminus \{s\}} \mathbb{C}ov\left[ \left( f_i^{(s)}-f_i^{(t)} \right)^{\otimes 2}, \left( f_j^{(s)}-f_j^{(t)} \right)^{\otimes 2}\right]\quad\quad\textrm{when $s=u$ and $t=v$}\nonumber\tag{E}\\
    &\quad
    + \sum_{s\in\set{P}}\sum_{t\in \set{P}\setminus \{s\}} \mathbb{C}ov\left[ \left( f_i^{(s)}-f_i^{(t)} \right)^{\otimes 2}, \left( f_j^{(t)}-f_j^{(s)} \right)^{\otimes 2}\right] \Bigg]\quad\quad\textrm{when $s=v$ and $t=u$}. \tag{F}
    \label{VarBreakdown}
\end{align}
}
\leqnomode
Note that we abuse notation in the sum notation, $s\in \set{P}$, such that $s$ is an integer representing the input sample in $\set{P}$. First, we simplify (A) by rewriting the covariance in terms of the underlying functions' ($f_i$ and $f_j$) statistics. We then note that the terms for (A) are the same as (B), (C), and (D) such that (A)=(B)=(C)=(D). Finally, we will simplify (E) and note that (E)=(F). First, consider the terms in (A) by expanding the Kronecker squares
\begin{align}
    \mathbb{C}ov\left[ \left( f_i^{(s)}-f_i^{(t)} \right)^{\otimes 2}, \left( f_j^{(s)}-f_j^{(v)} \right)^{\otimes 2}\right] = \mathbb{C}ov&\left[ f_i^{(s)\otimes 2}+f_i^{(t)\otimes 2} - (f_i^{(s)}\otimes f_i^{(t)}) - (f_i^{(t)}\otimes f_i^{(s)}), \right. \nonumber\\
    &\left. f_j^{(s)\otimes 2}+f_j^{(v)\otimes 2} - (f_j^{(s)}\otimes f_j^{(v)}) - (f_j^{(v)}\otimes f_j^{(s)}) \right].
    \label{eq:b2}
\end{align}
Since $f_i^{(t)}$ are $f_j^{(v)}$ are independent of each other while $f_i^{(s)}$ are $f_j^{(s)}$ are dependent, we use Appendix \ref{SubApp} to introduce expectations. Further, the input samples are i.i.d. such that the expectations of $f_i^{(s)}$ can be rewritten as expectations of the underlying function $f_i$. Finally, we factor the equation into a simplified form
\begin{align}
    \eqref{eq:b2} &= 
    \mathbb{C}ov\left[ f_i^{(s) \otimes 2} + \mathbb{E}[f_i^{(t)}]^{\otimes 2}- (f_i^{(s)}\otimes \mathbb{E}[f_i^{(t)}]) - ( \mathbb{E}[f_i^{(t)}]\otimes f_i^{(s)}), \right. \nonumber\\
    &\quad\quad\quad\quad \left. f_j^{(s)\otimes 2} + \mathbb{E}[f_j^{(v)}]^{\otimes 2} - (f_j^{(s)}\otimes  \mathbb{E}[f_j^{(v)}]) - ( \mathbb{E}[f_j^{(v)}]\otimes j^{(s)}) \right] \nonumber\\
    &= \scalemath{0.9}{\mathbb{C}ov\left[ f_i^{\otimes 2} + \mathbb{E}[f_i]^{\otimes 2}- (f_i\otimes \mathbb{E}[f_i]) - ( \mathbb{E}[f_i]\otimes f_i), f_j^{\otimes 2} + \mathbb{E}[f_j]^{\otimes 2} - (f_j\otimes  \mathbb{E}[f_j]) - ( \mathbb{E}[f_j]\otimes j) \right]} \nonumber\\
    &= \mathbb{C}ov\left[ (f_i - \mathbb{E}[f_i])^{\otimes 2}, (f_j - \mathbb{E}[f_j])^{\otimes 2} \right]. \nonumber \\
    &
    \label{eq:xy}
\end{align}
The same argument holds for (B), (C), and (D) simply by recognizing that $(x-y)^2 = (y-x)^2$ $\forall x,y \in \reals^D $ and noting that the triple sums are equivalent. Thus, $(\textrm{A})=(\textrm{B})=(\textrm{C})=(\textrm{D})$, and 
\begin{align}
    &(\textrm{A}) + (\textrm{B}) + (\textrm{C}) + (\textrm{D}) = 4\sum_{s\in\set{P}}\sum_{t\in \set{N}\setminus \{s\}}\sum_{v\in \set{M}\setminus \{s,t\}} \mathbb{C}ov\left[ (f_i - \mathbb{E}[f_i])^{\otimes 2}, (f_j - \mathbb{E}[f_j])^{\otimes 2} \right]. 
    \label{ABCD}
\end{align}
Because the covariance terms are not dependent on the specific samples, we only need to tabulate how many times the covariance arises in the sum. To this end, note that the set $\{s,t,v ; s\in \set{P}, t\in \set{N}\setminus \{s\}, v\in \set{M}\setminus \{ s,t \} \}$ can be partitioned as a disjoint union of four sets, $\{ s,t,v ; s\in \set{P}, t\in \set{P}\setminus \{s\}, v\in \set{P}\setminus \{ s,t \} \}$, $\{ s,t,v ; s\in \set{P}, t\in \set{P}\setminus \{s\}, v\in \set{M}\setminus \set{P} \}$, $\{ s,t,v ; s\in \set{P}, t\in \set{N}\setminus \set{P}, v\in \set{P}\setminus \{ s,t \} \}$, and $\{ s,t,v ; s\in \set{P}, t\in \set{N}\setminus \set{P}, v\in \set{M}\setminus \set{P} \}$. The cardinality of these sets are $P(P-1)(P-2)$, $P(P-1)(M-P)$, $P(N-P)(P-1)$, and $P(N-P)(M-P)$ respectively. By adding the cardinalities together,
\begin{align}
    &(\textrm{A}) + (\textrm{B}) + (\textrm{C}) + (\textrm{D}) = 4P\left[ (N-1)(M-1) - (P-1) \right] \mathbb{C}ov\left[ (f_i - \mathbb{E}[f_i])^{\otimes 2}, (f_j - \mathbb{E}[f_j])^{\otimes 2} \right]. 
\end{align}
Now, we consider the covariance in term (E) by following the same procedure
\reqnomode
\begin{align}
    \mathbb{C}ov&\left[ \left( f_i^{(s)}-f_i^{(t)} \right)^{\otimes 2}, \left( f_j^{(s)}-f_j^{(t)} \right)^{\otimes 2}\right] \nonumber\\
    &= \scalemath{0.9}{\mathbb{C}ov\left[ f_i^{(s)\otimes 2}+f_i^{(t)\otimes 2} - f_i^{(s)}\otimes f_i^{(t)} - f_i^{(t)}\otimes f_i^{(s)}, f_j^{(s)\otimes 2}+f_j^{(t)\otimes 2} - f_j^{(s)}\otimes f_j^{(t)} -f_j^{(t)} \otimes f_j^{(s)} \right]} \nonumber\\
\end{align}
\begin{align}
    &= 
    \mathbb{C}ov\left[ f_i^{(s)\otimes 2}, f_j^{(s)\otimes 2}+f_j^{(t)\otimes 2} - f_j^{(s)}\otimes f_j^{(t)} -f_j^{(t)} \otimes f_j^{(s)} \right] \nonumber \tag{G}\\
    &\quad + 
    \mathbb{C}ov\left[ f_i^{(t)\otimes 2}, f_j^{(s)\otimes 2}+f_j^{(t)\otimes 2} - f_j^{(s)}\otimes f_j^{(t)} -f_j^{(t)} \otimes f_j^{(s)} \right] \nonumber \tag{H}\\
    &\quad\quad -
    \mathbb{C}ov\left[f_i^{(s)}\otimes f_i^{(t)}, f_j^{(s)\otimes 2}+f_j^{(t)\otimes 2} - f_j^{(s)}\otimes f_j^{(t)} -f_j^{(t)} \otimes f_j^{(s)} \right] \nonumber \tag{I}\\ 
    &\quad\quad\quad - 
    \mathbb{C}ov\left[f_i^{(t)}\otimes f_i^{(s)}, f_j^{(s)\otimes 2}+f_j^{(t)\otimes 2} - f_j^{(s)}\otimes f_j^{(t)} -f_j^{(t)} \otimes f_j^{(s)} \right]. \nonumber \tag{J}
\end{align}
\leqnomode
Again, note (G) and (H) are equivalent due to symmetry. Using Appendix \ref{SubApp}, we combine (G) and (H) using a similar process as \eqref{eq:xy}
\begin{align}
    (\textrm{G}) + (\textrm{H}) &= 2 \mathbb{C}ov\left[ f_i^{\otimes 2}, f_j^{\otimes 2} - f_j\otimes \mathbb{E}[f_j] -\mathbb{E}[f_j] \otimes f_j \right] = 2 \mathbb{C}ov\left[ f_i^{\otimes 2}, \left(f_j -\mathbb{E}[f_j] \right)^{\otimes 2} \right].
\end{align}
Now, consider (I) and (J) and note that the second covariance inputs are identical. We now combine (I) and (J)
\reqnomode
\begin{align}
    (\textrm{I}) &+ (\textrm{J}) = -
    \mathbb{C}ov\left[f_i^{(s)}\otimes f_i^{(t)} + f_i^{(t)}\otimes f_i^{(s)}, f_j^{(s)\otimes 2}+f_j^{(t)\otimes 2} - f_j^{(s)}\otimes f_j^{(t)} -f_j^{(t)} \otimes f_j^{(s)} \right] \nonumber\\ 
    &= -   \mathbb{C}ov\left[f_i^{(s)}\otimes f_i^{(t)} + f_i^{(t)}\otimes f_i^{(s)}, f_j^{(s)\otimes 2}\right] -
    \mathbb{C}ov\left[f_i^{(s)}\otimes f_i^{(t)} + f_i^{(t)}\otimes f_i^{(s)}, f_j^{(t)\otimes 2} \right] \nonumber \tag{K}\\
    &\quad +
    \mathbb{C}ov\left[f_i^{(s)}\otimes f_i^{(t)} + f_i^{(t)}\otimes f_i^{(s)}, f_j^{(s)}\otimes f_j^{(t)} \right] \nonumber\tag{L}\\
    &\quad\quad +
    \mathbb{C}ov\left[f_i^{(s)}\otimes f_i^{(t)} + f_i^{(t)}\otimes f_i^{(s)}, f_j^{(t)}\otimes f_j^{(s)} \right] \nonumber\tag{M}
\end{align}
\leqnomode
by separating the second covariance inputs. Since $s\neq t$, we combine the terms in (K) by using Appendix \ref{SubApp}, $(\textrm{K}) = - 2 \mathbb{C}ov\left[f_i\otimes \mathbb{E}[f_i] + \mathbb{E}[f_i]\otimes f_i, f_j^{\otimes 2}\right].$ Now, we consider (L) by breaking the covariance into expectations
\begin{align}
    (\textrm{L}) &= \scalemath{0.9}{\ee{\left( f_i^{(s)}\otimes f_i^{(t)} + f_i^{(t)}\otimes f_i^{(s)} \right)\left( f_j^{(s)}\otimes f_j^{(t)} \right)^T} - \ee{\left( f_i^{(s)}\otimes f_i^{(t)} + f_i^{(t)}\otimes f_i^{(s)} \right)}\ee{\left( f_j^{(s)}\otimes f_j^{(t)} \right)^T}} \nonumber \\
    &= \ee{\left( f_i^{(s)}\otimes f_i^{(t)}\right)\left( f_j^{(s)}\otimes f_j^{(t)} \right)^T} + \ee{ \left(f_i^{(t)}\otimes f_i^{(s)} \right)\left( f_j^{(s)}\otimes f_j^{(t)} \right)^T} \nonumber \\
    &\quad - \ee{\left( f_i^{(s)}\otimes f_i^{(t)}\right)}\ee{\left( f_j^{(s)}\otimes f_j^{(t)} \right)^T} - \ee{\left(f_i^{(t)}\otimes f_i^{(s)} \right)}\ee{\left( f_j^{(s)}\otimes f_j^{(t)} \right)^T} \nonumber \\
    &= \ee{\left( f_i^{(s)}f_j^{(s)T}\otimes f_i^{(t)}f_j^{(t)T}\right)}+ \ee{ \left(f_i^{(t)}f_j^{(s)T}\otimes f_i^{(s)}f_j^{(t)T} \right)} \nonumber \\
    &\quad - \left(\ee{ f_i} \otimes \ee{f_i}\right) \left( \ee{f_j^T}\otimes \ee{f_j^T} \right) - \left(\ee{f_i}\otimes \ee{f_i} \right) \left( \ee{f_j^T}\otimes \ee{f_j^T} \right) \nonumber \\
    &= \ee{f_if_j^T}\otimes\ee{ f_if_j^T}+ \ee{ \left(f_i^{(t)}f_j^{(s)T}\otimes f_i^{(s)}f_j^{(t)T} \right)} - 2\ee{ f_i}^{\otimes 2} \ee{f_j^T}^{\otimes 2} \nonumber \\
    &= \ee{f_if_j^T}^{\otimes 2} - 2\left(\ee{ f_i}\ee{ f_j}\right)^{\otimes 2} + \ee{ \left(f_i^{(t)}f_j^{(s)T}\otimes f_i^{(s)}f_j^{(t)T} \right)} . \nonumber \\
    &
\end{align}
The argument for (M) is identical with the following result
\begin{align}
    (\textrm{M}) = \ee{f_if_j^T}^{\otimes 2} - 2\left(\ee{ f_i}\ee{ f_j}\right)^{\otimes 2} + \ee{ \left(f_i^{(s)}f_j^{(t)T}\otimes f_i^{(t)}f_j^{(s)T} \right)} .
\end{align}
Now, we find $\ee{ \left(f_i^{(t)}f_j^{(s)T}\otimes f_i^{(s)}f_j^{(t)T} \right)}$ in order to remove the dependence on $s$ and $t$
\begin{align}
    \mathbb{E}[f_i^{(t)} f_j^{(s)T} \otimes f_i^{(s)} f_j^{(t)T}] &= \mathbb{E}[(f_i^{(t)} \bold 1_D^T \circ \bold 1_D f_j^{(s)T}) \otimes (f_i^{(s)}\bold 1_D^T \circ \bold 1_D f_j^{(t)T})] \nonumber\\
    &= \mathbb{E}[(f_i^{(t)} \bold 1_D^T \otimes \bold 1_D f_j^{(t)T}) \circ (\bold 1_D f_j^{(s)T} \otimes f_i^{(s)}\bold 1_D^T) )]\nonumber\\
    &= \mathbb{E}[f_i \bold 1_D^T \otimes \bold 1_D f_j^T] \circ \mathbb{E}[\bold 1_D f_j^T \otimes f_i\bold 1_D^T] \nonumber\\
    &= \left( \bold 1_D^T \otimes \mathbb{E}[f_i f_j^T] \otimes \bold 1_D \right) \circ \left(\bold 1_D \otimes \mathbb{E}[f_if_j^T] \otimes \bold 1_D^T \right)
\end{align}
using the Kronecker and Hadamard mixed-product properties. The argument for \\
$\ee{ \left(f_i^{(s)}f_j^{(t)T}\otimes f_i^{(t)}f_j^{(s)T} \right)}$ is identical such that (L)=(M).

Now, we rewrite (K) by adding and subtracting an additional term to facilitate simplification in later steps
\begin{align}
    (\textrm{K}) &= 2 \mathbb{C}ov\left[-f_i\otimes \mathbb{E}[f_i] - \mathbb{E}[f_i]\otimes f_i, f_j^{\otimes 2}\right] \nonumber\\
    &= 2 \mathbb{C}ov\left[-f_i\otimes \mathbb{E}[f_i] - \mathbb{E}[f_i]\otimes f_i, f_j^{\otimes 2}-f_j\otimes \mathbb{E}[f_j] - \mathbb{E}[f_j]\otimes f_j\right] \nonumber \\
    &\quad - 2 \mathbb{C}ov\left[-f_i\otimes \mathbb{E}[f_i] - \mathbb{E}[f_i]\otimes f_i, -f_j\otimes \mathbb{E}[f_j] - \mathbb{E}[f_j]\otimes f_j\right].
\end{align}
Adding (K) to (G) and (H) yields
\reqnomode
\begin{align}
    (\textrm{G})+(\textrm{H})+(\textrm{K}) &= 
    2 \mathbb{C}ov\left[ \left(f_i -\mathbb{E}[f_i] \right)^{\otimes 2}, \left(f_j -\mathbb{E}[f_j] \right)^{\otimes 2} \right]  \nonumber \tag{N}\\
    &\quad - 2 \mathbb{C}ov\left[-f_i\otimes \mathbb{E}[f_i] - \mathbb{E}[f_i]\otimes f_i, -f_j\otimes \mathbb{E}[f_j] - \mathbb{E}[f_j]\otimes f_j\right]. \nonumber \tag{O}
\end{align}
\leqnomode
Now, we break (O) into expectations
\begin{align}
    (\textrm{O}) &= - 2 \mathbb{C}ov\left[f_i\otimes \mathbb{E}[f_i] + \mathbb{E}[f_i]\otimes f_i, f_j\otimes \mathbb{E}[f_j] + \mathbb{E}[f_j]\otimes f_j\right] \nonumber \\
    &= -2 \left[ \cov{f_i\otimes \mathbb{E}[f_i], f_j\otimes \mathbb{E}[f_j]} + \cov{\mathbb{E}[f_i]\otimes f_i, f_j\otimes \mathbb{E}[f_j]} \right. \nonumber \\
    &\quad\quad \left.+ \cov{f_i\otimes \mathbb{E}[f_i], \mathbb{E}[f_j]\otimes f_j} + \cov{\mathbb{E}[f_i]\otimes f_i, \mathbb{E}[f_j]\otimes f_j} \right] \nonumber\\
    &= -2\left[ \ee{(f_i\otimes \mathbb{E}[f_i])( f_j\otimes \mathbb{E}[f_j])^T} - \ee{f_i\otimes \mathbb{E}[f_i]}\ee{f_j\otimes \mathbb{E}[f_j]}^T \right.\nonumber\\
    &\quad\quad + \ee{(\mathbb{E}[f_i]\otimes f_i)( f_j\otimes \mathbb{E}[f_j])^T} - \ee{\mathbb{E}[f_i]\otimes f_i}\ee{f_j\otimes \mathbb{E}[f_j]}^T \nonumber \\
    &\quad\quad\quad + \ee{(f_i\otimes \mathbb{E}[f_i])(\mathbb{E}[f_j]\otimes f_j)^T} - \ee{f_i\otimes \mathbb{E}[f_i]}\ee{\mathbb{E}[f_j]\otimes f_j}^T \nonumber \\
    &\quad\quad\quad\quad \left.+ \ee{(\mathbb{E}[f_i]\otimes f_i)(\mathbb{E}[f_j]\otimes f_j)^T} - \ee{\mathbb{E}[f_i]\otimes f_i}\ee{\mathbb{E}[f_j]\otimes f_j}^T \right] \nonumber \\
    &= -2\left[ \ee{f_i f_j^T\otimes \mathbb{E}[f_i]\mathbb{E}[f_j]^T}  + \ee{\mathbb{E}[f_i]f_j^T\otimes f_i\mathbb{E}[f_j^T]}  \right.\nonumber \\
    &\quad\quad \left. + \ee{f_i\mathbb{E}[f_j^T]\otimes \mathbb{E}[f_i]f_j^T} + \ee{\mathbb{E}[f_i]\mathbb{E}[f_j^T]\otimes f_i f_j^T} \right] + 8\left(\ee{f_i}\ee{f_j^T}\right)^{\otimes 2} \nonumber\\
    &= -2\ee{f_i f_j^T}\otimes \mathbb{E}[f_i]\mathbb{E}[f_j]^T  -2 \mathbb{E}[f_i] \otimes \ee{f_i f_j^T} \otimes \mathbb{E}[f_j^T]  \nonumber \\
    &\quad\quad -2 \ee{f_j^T} \otimes \mathbb{E}[f_i f_j^T]\otimes \ee{f_i} -2 \mathbb{E}[f_i]\mathbb{E}[f_j^T]\otimes \ee{f_i f_j^T}  + 8\left(\ee{f_i}\ee{f_j^T}\right)^{\otimes 2}.
\end{align}
Now, we combine (O), (L), and (M), and use the results from Appendix \ref{SubApp}
\begin{align}
    (\textrm{O}) + (\textrm{L}) + (\textrm{M}) &= -2\ee{f_i f_j^T}\otimes \mathbb{E}[f_i]\mathbb{E}[f_j]^T  -2 \mathbb{E}[f_i] \otimes \ee{f_i f_j^T} \otimes \mathbb{E}[f_j^T]  \nonumber \\
    &\quad \scalemath{0.9}{-2 \ee{f_j^T} \otimes \mathbb{E}[f_i f_j^T]\otimes \ee{f_i} -2 \mathbb{E}[f_i]\mathbb{E}[f_j^T]\otimes \ee{f_i f_j^T}  + 8\left(\ee{f_i}\ee{f_j^T}\right)^{\otimes 2}} \nonumber \\
    &\quad  \scalemath{0.9}{+ 2\ee{f_if_j^T}^{\otimes 2} - 4\left(\ee{ f_i}\ee{ f_j}\right)^{\otimes 2} + 2\left( \bold 1_D^T \otimes \mathbb{E}[f_i f_j^T] \otimes \bold 1_D \right) \circ \left(\bold 1_D \otimes \mathbb{E}[f_if_j^T] \otimes \bold 1_D^T \right)} \nonumber
\end{align}
\begin{align}
    &= 2\cov{f_i,f_j}^{\otimes 2}   -2 \mathbb{E}[f_i] \otimes \ee{f_i f_j^T} \otimes \mathbb{E}[f_j^T]  -2 \ee{f_j^T} \otimes \mathbb{E}[f_i f_j^T]\otimes \ee{f_i} \nonumber \\
    &\quad + 2\left(\ee{f_i}\ee{f_j^T}\right)^{\otimes 2}  + 2\left( \bold 1_D^T \otimes \mathbb{E}[f_i f_j^T] \otimes \bold 1_D \right) \circ \left(\bold 1_D \otimes \mathbb{E}[f_if_j^T] \otimes \bold 1_D^T \right) \nonumber \\
    &=2\cov{f_i,f_j}^{\otimes 2} + 2 \left( \bold 1_D^T \otimes \mathbb{C}ov[f_i, f_j] \otimes \bold 1_D \right) \circ \left(\bold 1_D \otimes \mathbb{C}ov[f_i,f_j] \otimes \bold 1_D^T \right). 
\end{align}
Therefore, by combining (N), (O), (L), and (M), we obtain a final expression for (E)
\begin{align}
     (E) &= 
      \sum_{s\in\set{P}}\sum_{t\in \set{P}\setminus \{s\}} \left( 2 \mathbb{C}ov\left[ \left(f_i -\mathbb{E}[f_i] \right)^{\otimes 2}, \left(f_j -\mathbb{E}[f_j] \right)^{\otimes 2} \right] + 2\cov{f_i,f_j}^{\otimes 2} \right.\nonumber\\
     &\quad\quad\quad\quad\quad \left. + 2 \left( \bold 1_D^T \otimes \mathbb{C}ov[f_i, f_j] \otimes \bold 1_D \right) \circ \left(\bold 1_D \otimes \mathbb{C}ov[f_i,f_j] \otimes \bold 1_D^T \right) \right). \nonumber\\
     &
\end{align}
The argument for the covariance term in (F) is identical, and their sums are equivalent such that $(E) = (F)$. Now, we only need to tabulate how many times the covariance arises in the sum since the covariance is independent of specific samples. Thus, the set $\{ s, t; s\in \set{P}, t \in \set{P}\setminus\{s\}\}$ has the cardinality of $P(P-1)$. Therefore,
\begin{align}
    \mathbb{C}ov&[\vec{Q}_{i}(\set{N}), \vec{Q}_{j}(\set{M})] = \frac{1}{4N(N-1)M(M-1)} [ \nonumber\\
    &
    4P\left[ (N-1)(M-1) - (P-1) \right] \mathbb{C}ov\left[ (i - \mathbb{E}[f_i])^{\otimes 2}, (j - \mathbb{E}[f_j])^{\otimes 2} \right] \nonumber\\
    &\quad + 2P(P-1) \left[
    2 \mathbb{C}ov\left[ (i-\mathbb{E}[f_i])^{\otimes 2}, \left(j -\mathbb{E}[f_j] \right)^{\otimes 2} \right] + 2 \mathbb{C}ov[i,j]^{\otimes 2} \right. \nonumber\\
    &\quad\quad \left.\left. + 2\left( \bold 1^T \otimes \mathbb{C}ov[i, j] \otimes \bold 1 \right) \circ \left(\bold 1 \otimes \mathbb{C}ov[i,j] \otimes \bold 1^T \right)\right]\right].
\end{align}
Combining coefficients yields our stated result
\begin{align}
    &\mathbb{C}ov[\vec{Q}_{i}(\set{N}), \vec{Q}_{j}(\set{M})] = \nonumber\\
    &\scalemath{0.9}{
     \frac{1}{4N(N-1)M(M-1)} \left[4P(N-1)(M-1) - 4P(P-1)\right] \mathbb{C}ov\left[\left(f_i-\mathbb{E}[f_i]\right)^{\otimes 2},(f_j-\mathbb{E}[f_j])^{\otimes 2}\right]} \nonumber\\
     &\quad \scalemath{0.9}{+ \frac{2P(P-1)}{4N(N-1)M(M-1)} \left[ 2 \mathbb{C}ov\left[ (f_i-\mathbb{E}[f_i])^{\otimes 2}, \left(f_j -\mathbb{E}[f_j] \right)^{\otimes 2} \right] + 2 \mathbb{C}ov[f_i,f_j]^{\otimes 2}\right.} \nonumber\\
    &\quad\quad \scalemath{0.9}{\left.+ 2\left( \bold 1^T \otimes \mathbb{C}ov[f_i, f_j] \otimes \bold 1 \right) \circ \left(\bold 1 \otimes \mathbb{C}ov[f_i,f_j] \otimes \bold 1^T \right) \right]} \nonumber\\
    &= 
    \scalemath{0.9}{\frac{P(P-1)}{N(N-1)M(M-1)}\left[ \mathbb{C}ov[f_i,f_j]^{\otimes 2}  + \left( \bold 1^T \otimes \mathbb{C}ov[f_i, f_j] \otimes \bold 1 \right) \circ \left(\bold 1 \otimes \mathbb{C}ov[f_i,f_j] \otimes \bold 1^T \right)\right]} \nonumber\\
    &\quad \scalemath{0.9}{+ \frac{P}{NM}\mathbb{C}ov\left[\left(f_i-\mathbb{E}[f_i]\right)^{\otimes 2},(f_j-\mathbb{E}[f_j])^{\otimes 2}\right]} \nonumber\\
    &= \frac{P(P-1)}{N(N-1)M(M-1)} \bold V_{ij} + \frac{P}{NM} \mat{W}_{ij}.
    \label{eq:VarMCVareq}
\end{align}

\section{Proof of Proposition \ref{Var_delt} and \ref{Var_high}}
Now, Equation \eqref{eq:VarMCVareq} is used to find the discrepancy covariances with different input samples
\begin{align}
        \scalemath{0.8}{\mathbb{C}ov[\bold \Delta_i\Ra{(\un{\set{Z}}_i)}, \bold \Delta_j\Ra{(\un{\set{Z}}_j)}]} &= 
        \scalemath{0.8}{\mathbb{C}ov[\bold Q_i(\set{Z}_i^*) ,\bold Q_j(\set{Z}_j^*)] -\mathbb{C}ov[\bold Q_i(\set{Z}_i^*) ,\bold Q_j(\set{Z}_j)] -\mathbb{C}ov[\bold Q_i(\set{Z}_i) ,\bold Q_j(\set{Z}_j^*)] + \mathbb{C}ov[\bold Q_i(\set{Z}_i) ,\bold Q_j(\set{Z}_j)]}
\end{align}
\begin{align}
    &= \left[ \frac{\Ncomb{i}{*}{j}{*}(\Ncomb{i}{*}{j}{*}-1)}{\Nsing{i}{*}(\Nsing{i}{*}-1)\Nsing{j}{*}(\Nsing{j}{*}-1)} - \frac{\Ncomb{i}{*}{j}{}(\Ncomb{i}{*}{j}{}-1)}{\Nsing{i}{*}(\Nsing{i}{*}-1)\Nsing{j}{}(\Nsing{j}{}-1)} \right.  \nonumber\\
    &\quad \left. - \frac{\Ncomb{i}{}{j}{*}(\Ncomb{i}{}{j}{*}-1)}{\Nsing{i}{}(\Nsing{i}{}-1)\Nsing{j}{*}(\Nsing{j}{*}-1)}  + \frac{\Ncomb{i}{}{j}{}(\Ncomb{i}{}{j}{}-1)}{\Nsing{i}{}(\Nsing{i}{}-1)\Nsing{j}{}(\Nsing{j}{}-1)}\right] \mat{V}_{ij}  \nonumber\\
    &\quad\quad + \left[  \frac{\Ncomb{i}{*}{j}{*}}{\Nsing{i}{*}\Nsing{j}{*}} - \frac{\Ncomb{i}{*}{j}{}}{\Nsing{i}{*}\Nsing{j}{}} - \frac{\Ncomb{i}{}{j}{*}}{\Nsing{i}{}\Nsing{j}{*}} + \frac{\Ncomb{i}{}{j}{}}{\Nsing{i}{}\Nsing{j}{}} \right]\mat{W}_{ij}.
\end{align}

Similarly to above, Equation \eqref{eq:VarMCVareq} can be used with different input samples to find the covariance between the high-fidelity estimator and the discrepancy.

\section{Proof of Proposition \ref{MeanCov}, \ref{MeanCovVar}, and \ref{MeanVarHighLow}}
We now find the covariance between the mean and variance estimators
\begin{align}
    \mathbb{C}ov[\vec{Q}_{\mu,i}(\set{N}), \vec{Q}_{V,j}(\set{M})] &= \mathbb{C}ov\left[\frac{1}{N}\sum_s^N f_i^{(s)}, \frac{1}{2M(M-1)}\sum_u^M \sum_v^M ( f_j^{(u)} - f_j^{(v)} )^{\otimes 2}\right] \nonumber\\
    &= \frac{1}{2M(M-1)N} \sum_s^N \sum_u^M \sum_v^M \mathbb{C}ov[f_i^{(s)}, ( f_j^{(u)} - f_j^{(v)} )^{\otimes 2}].
    \label{eq:MVcov}
\end{align}
We consider the 2 cases of shared input samples that relate to nonzero covariance terms
\reqnomode
\begin{align}
    &\eqref{eq:MVcov} = \frac{1}{2M(M-1)N} \sum_{s\in \set{P}} \sum_{v\in \set{M}\setminus\{s \}} \mathbb{C}ov[f_i^{(s)}, ( f_j^{(s)} - f_j^{(v)} )^{\otimes 2}] \quad \textrm{when $s=u$ and $s\neq v$} \nonumber \tag{A}\\
    &\quad + \frac{1}{2M(M-1)N} \sum_{s\in \set{P}} \sum_{u\in \set{M}\setminus\{s \}} \mathbb{C}ov[f_i^{(s)}, ( f_j^{(u)} - f_j^{(s)} )^{\otimes 2}] \quad \textrm{when $s\neq u$ and $s= v$}. \tag{B} \nonumber
\end{align}
\leqnomode
Note that (A) = (B) by the same argument from below Equation \eqref{eq:xy}. Thus, consider (A)
\begin{align}
    \mathbb{C}ov[f_i^{(s)}, ( f_j^{(s)} - f_j^{(v)} )^{\otimes 2}] &= \mathbb{C}ov[f_i^{(s)}, f_j^{(s)\otimes 2} - f_j^{(s)} \otimes f_j^{(v)} - f_j^{(v)} \otimes f_j^{(s)} + f_j^{(v)\otimes 2}] \nonumber\\
    &= \mathbb{C}ov[f_i, f_j^{\otimes 2} - f_j \otimes \mathbb{E}[f_j] - \mathbb{E}[f_j] \otimes f_j + \mathbb{E}[f_j]^{\otimes 2}] \nonumber\\
    &= \mathbb{C}ov[f_i, (f_j - \mathbb{E}[f_j])^{\otimes 2}]
\end{align}
using the results of Appendix \ref{SubApp} and the fact that the input samples are i.i.d. where the statistics can be rewritten in terms of the underlying functions' statistics. The cardinality of the set $\{ s,v; s\in \set{P}, v \in \set{M}\setminus\{ s \} \}$ is $P(M-1)$. Thus,
\begin{align}
    \mathbb{C}ov[\vec{Q}_{\mu,i}(\set{N}), \vec{Q}_{V,j}(\set{M})] &=  \frac{1}{2M(M-1)N}2P(M-1)\mathbb{C}ov[i, (j - \mathbb{E}[f_j])^{\otimes 2}] \\
    &= \frac{P}{MN}\mathbb{C}ov[i, (j - \mathbb{E}[f_j])^{\otimes 2}] = \frac{P}{MN} \mat{B}_{ij}.
\end{align}
Propositions \ref{MeanCovVar} and \ref{MeanVarHighLow} follow by using the above equation with different sets of input samples. 

\section{Proof of Proposition \ref{SobolEstCov2}, \ref{MultiSobolVar}, and \ref{MultiSobolHighLow}}
\label{appendix:SobolEstCov}
We now find the covariance between two main effect variance estimators. Let $f_i^{(a)} \equiv f_i(\vec{z}^{(a)})$ and $f_{i,x}^{(a)} \equiv f_i(\vec{y}^{(a)}_{x})$ that follow the same sample partitioning described in Section \ref{sec:MCests}. Now, 
\begin{align}
    \mathbb{C}ov[\bold Q_{i,x}(\set{N}),
    \bold Q_{j,y}(\set{M})] &= \scalemath{0.9}{\mathbb{C}ov\left[\frac{1}{N^2}\sum_a^N\sum_b^N \left[f_i^{(a)} f_{i,x}^{(a)} - f_i^{(a)} f_i^{(b)} \right], \frac{1}{M^2}\sum_c^M\sum_d^M \left[f_j^{(c)} f_{j,y}^{(c)} - f_j^{(c)} f_j^{(d)} \right]\right]} \nonumber\\
    &= \frac{1}{N^2 M^2}\sum_a^N\sum_b^N\sum_c^M\sum_d^M\mathbb{C}ov\left[ f_i^{(a)} f_{i,x}^{(a)} - f_i^{(a)} f_i^{(b)} ,  f_j^{(c)} f_{j,y}^{(c)} - f_j^{(c)} f_j^{(d)}\right].
\end{align}
There are 11 cases of nonzero covariance terms which are found through combinations of sharing input samples. We start with the case when $a=c \neq b\neq d$. Using Appendix \ref{SubApp} and the fact that the samples are i.i.d. such that the statistics can be written in terms of the underlying functions' statistics
\begin{align}
    \mathbb{C}ov\left[ f_i^{(a)} f_{i,x}^{(a)} - f_i^{(a)} f_i^{(b)} ,  f_j^{(a)} f_{j,y}^{(a)} - f_j^{(a)} f_j^{(d)}\right] &= \mathbb{C}ov[f_i f_{i,x} - f_i \mathbb{E}[f_i] ,  f_j f_{j,y} - f_j \mathbb{E}[f_j]].
\end{align}
This process is repeated for all combinations of sets of shared input samples. The results can be seen in Table \ref{tab:SAcov},

\begin{table}[h]
    \begin{center}
    \caption{Covariances resulting from combinations of shared samples}
    \label{tab:SAcov}
    \begin{tabular}{|c | c | c |}
    \hline
    Shared Samples & Covariance & Occurrence Frequency \\
    \hline
         $a=c \neq b\neq d$ & $\mathbb{C}ov[f_i f_{i,x} - f_i \mathbb{E}[f_i] ,  f_j f_{j,y} - f_j \mathbb{E}[f_j]]$ & $P(N-1)(M-1) - P(P-1)$\\
         $a=d \neq b\neq c$ & $\mathbb{C}ov[f_i f_{i,x} - f_i \mathbb{E}[f_i] , - f_j \mathbb{E}[f_j]]$ & $P(N-1)(M-1) - P(P-1)$\\
         $a\neq d \neq b = c$ & $\mathbb{C}ov[- f_i \mathbb{E}[f_i] ,  f_j f_{j,y} - f_j \mathbb{E}[f_j]]$ & $P(N-1)(M-1) - P(P-1)$\\
         $a\neq c\neq b= d$ & $\mathbb{C}ov[ - f_i \mathbb{E}[f_i] ,  - f_j \mathbb{E}[f_j]]$ & $P(N-1)(M-1) - P(P-1)$\\
         $a=b=c\neq d$ & $\mathbb{C}ov[f_i f_{i,x} - f_i^2 ,  f_j f_{j,y} - f_j \mathbb{E}[f_j]]$ & $P(M-1)$\\
         $a=c=d\neq b$  & $\mathbb{C}ov[f_i f_{i,x} - f_i \mathbb{E}[f_i] ,  f_j f_{j,y} - f_j^2]$ & $P(N-1)$\\
         $a\neq b=c=d$  & $\mathbb{C}ov[ - f_i \mathbb{E}[f_i] ,  f_j f_{j,y} - f_j^2]$ & $P(N-1)$\\
         $a=b=d\neq c$  & $\mathbb{C}ov[f_i f_{i,x} - f_i^2,  - f_j \mathbb{E}[f_j]]$ & $P(M-1)$\\
         $a=c\neq b=d$ & Eq. \eqref{eq:ac} & $P(P-1)$\\
         $a=d\neq b=c$ & Eq. \eqref{eq:ad} & $P(P-1)$\\
         $a=b=c=d$  & $\mathbb{C}ov[f_i f_{i,x} - f_i^2 ,  f_j f_{j,y} - f_j^2]$ & $P$\\
         \hline
    \end{tabular}
\end{center}
\end{table}

The derivations for $a=c\neq b=d$ and $a=d\neq b=c$ are not as straightforward as the other results in the table. Thus, consider $a=c\neq b=d$
\begin{align}
    \mathbb{C}ov\left[ f_i^{(a)} f_{i,x}^{(a)} - f_i^{(a)} f_i^{(b)} \right.& \left., f_j^{(a)} f_{j,y}^{(a)} - f_j^{(a)} f_j^{(b)}\right] =  \nonumber\\
    &\quad \scalemath{0.9}{\mathbb{E}[f_i^{(a)} f_{i,x}^{(a)} f_j^{(a)} f_{j,y}^{(a)} - f_i^{(a)} f_{i,x}^{(a)} f_j^{(a)} f_j^{(b)} - f_i^{(a)} f_i^{(b)} f_j^{(a)} f_{j,y}^{(a)} + f_i^{(a)} f_i^{(b)} f_j^{(a)} f_j^{(b)}]}  \nonumber\\
    &\quad \scalemath{0.9}{-\mathbb{E}[f_i^{(a)} f_{i,x}^{(a)} - f_i^{(a)} f_i^{(b)}]\mathbb{E}[f_j^{(a)} f_{j,y}^{(a)} - f_j^{(a)} f_j^{(b)}]} \nonumber\\
    &= \scalemath{0.9}{\mathbb{E}[f_i f_{i,x} f_j f_{j,y}] - \mathbb{E}[f_i f_{i,x} f_j]\mathbb{E}[f_j] - \mathbb{E}[f_i f_j f_{j,y}]\mathbb{E}[f_i] - \mathbb{E}[f_i f_{i,x}]\mathbb{E}[f_j f_{j,y}]}  \nonumber\\
    &\quad \scalemath{0.9}{+ \mathbb{E}[f_i f_{i,x}]\mathbb{E}[f_j]^2 + \mathbb{E}[f_i]^2\mathbb{E}[f_j f_{j,y}] + \mathbb{E}[f_i f_j]^2 - \mathbb{E}[f_i]^2\mathbb{E}[f_j]^2} \nonumber \\
    &= \mathbb{C}ov[f_i f_{i,x},f_j f_{j,y}] + \mathbb{C}ov[-f_i\mathbb{E}[f_i], f_j f_{j,y}] + \mathbb{C}ov[f_i f_{i,x},-f_j \mathbb{E}[f_j]]  \nonumber\\
    &\quad + \mathbb{E}[f_i f_j]^2 - \mathbb{E}[f_i]^2\mathbb{E}[f_j]^2 \nonumber
\end{align}
\begin{align}
    &= \mathbb{C}ov[f_i f_{i,x} - f_i \mathbb{E}[f_i], f_j f_{j,y} - f_j \mathbb{E}[f_j]]  - \mathbb{C}ov[f_i \mathbb{E}[f_i],f_j \mathbb{E}[f_j]] \nonumber\\
    &\quad + \mathbb{E}[f_i f_j]^2 - \mathbb{E}[f_i]^2\mathbb{E}[f_j]^2  \nonumber\\
    &= \mathbb{C}ov[f_i f_{i,x} - f_i \mathbb{E}[f_i], f_j f_{j,y} - f_j \mathbb{E}[f_j]] + \mathbb{E}[f_i f_j]^2 - \mathbb{E}[f_i f_j]\mathbb{E}[f_i]\mathbb{E}[f_j] \nonumber\\
    &= \mathbb{C}ov[f_i f_{i,x} - f_i \mathbb{E}[f_i], f_j f_{j,y} - f_j \mathbb{E}[f_j]] + \mathbb{C}ov[f_i,f_j]^2 + \mathbb{E}[f_i f_j]\mathbb{E}[f_i]\mathbb{E}[f_j] - \mathbb{E}[f_i]^2\mathbb{E}[f_j]^2 \nonumber\\
    &= \mathbb{C}ov[f_i f_{i,x} - f_i \mathbb{E}[f_i], f_j f_{j,y} - f_j \mathbb{E}[f_j]] + \mathbb{C}ov[f_i,f_j]^2 + \mathbb{C}ov[f_i \mathbb{E}[f_i],f_j \mathbb{E}[f_j]]. \nonumber \\
    &
    \label{eq:ac}
\end{align}
The final case to consider is when $a=d\neq b = c$
\begin{align}
    \mathbb{C}ov\left[ f_i^{(a)} f_{i,x}^{(a)} - \right.& \left. f_i^{(a)} f_i^{(b)} ,  f_j^{(c)} f_{j,y}^{(c)} - f_j^{(c)} f_j^{(d)}\right] =  \nonumber\\
    &\quad \mathbb{C}ov[f_i^{(a)} f_{i,x}^{(a)}, -f_j^{(c)} f_j^{(d)}] + \mathbb{C}ov[- f_i^{(a)} f_i^{(b)},f_j^{(c)} f_{j,y}^{(c)}] + \mathbb{C}ov[f_i^{(a)} f_i^{(b)},f_j^{(c)} f_j^{(d)}] \nonumber\\
    &= \mathbb{C}ov[f_i f_{i,x},-f_j\mathbb{E}[f_j]] + \mathbb{C}ov[-f_i\mathbb{E}[f_i],f_j f_{j,y}] + \mathbb{E}[f_i f_j]^2 - \mathbb{E}[f_i]^2\mathbb{E}[f_j]^2 \nonumber\\
    &= \mathbb{C}ov[f_i f_{i,x},-f_j\mathbb{E}[f_j]] + \mathbb{C}ov[-f_i\mathbb{E}[f_i],f_j f_{j,y}] + \mathbb{C}ov[f_i,f_j]^2\nonumber \\
    &\quad + 2\mathbb{E}[f_i f_j]\mathbb{E}[f_i]\mathbb{E}[f_j] - 2\mathbb{E}[f_i]^2\mathbb{E}[f_j]^2 \nonumber\\
    &= \scalemath{0.9}{\mathbb{C}ov[f_i f_{i,x},-f_j \mathbb{E}[f_j]] + \mathbb{C}ov[-f_i\mathbb{E}[f_i],f_j f_{j,y}] + \mathbb{C}ov[f_i,f_j]^2 + 2\mathbb{C}ov[f_i\mathbb{E}[f_i],f_j\mathbb{E}[f_j]]}.
    \label{eq:ad}
\end{align}
To combine all of these cases together,
\begin{align}
    \mathbb{C}ov &[\bold Q_{i,x}(\set{N}) ,
    \bold Q_{j,y}(\set{M})]] = \frac{1}{N^2 M^2}\sum_a^N\sum_b^N\sum_c^M\sum_d^M\mathbb{C}ov\left[ f_i^{(a)} f_{i,x}^{(a)} - f_i^{(a)} f_i^{(b)} ,  f_j^{(c)} f_{j,y}^{(c)} - f_j^{(c)} f_j^{(d)}\right] \\
    &= \frac{1}{N^2 M^2} [ (P(N-1)(M-1) - P(P-1))( \mathbb{C}ov[f_i f_{i,x} - f_i \mathbb{E}[f_i] ,  f_j f_{j,y} - f_j \mathbb{E}[f_j]]  \nonumber\\
    &\quad \scalemath{0.9}{+  \mathbb{C}ov[f_i f_{i,x} - f_i \mathbb{E}[f_i] , - f_j \mathbb{E}[f_j]] + \mathbb{C}ov[- f_i \mathbb{E}[f_i] ,  f_j f_{j,y} - f_j \mathbb{E}[f_j]] + \mathbb{C}ov[ - f_i \mathbb{E}[f_i] ,  - f_j \mathbb{E}[f_j]] )}  \nonumber\\
    &\quad + P(N-1)(\mathbb{C}ov[f_i f_{i,x} - f_i \mathbb{E}[f_i] ,  f_j f_{j,y} - f_j^2] + \mathbb{C}ov[ - f_i \mathbb{E}[f_i] ,  f_j f_{j,y} - f_j^2])  \nonumber\\
    &\quad + P(M-1)(\mathbb{C}ov[f_i f_{i,x} - f_i^2 ,  f_j f_{j,y} - f_j \mathbb{E}[f_j]] + \mathbb{C}ov[f_i f_{i,x} - f_i^2,  - f_j \mathbb{E}[f_j]])  \nonumber\\
    &\quad + P(P-1)( \mathbb{C}ov[f_i f_{i,x} - f_i\mathbb{E}[f_i], f_j f_{j,y} - f_j \mathbb{E}[f_j]] + \mathbb{C}ov[f_i,f_j]^2 + \mathbb{C}ov[f_i\mathbb{E}[f_i],f_j\mathbb{E}[f_j]]  \nonumber\\
    &\quad +  \mathbb{C}ov[f_i f_{i,x},-f_j\mathbb{E}[f_j]] + \mathbb{C}ov[-f_i\mathbb{E}[f_i],f_j f_{j,y}] + \mathbb{C}ov[f_i,f_j]^2 + 2\mathbb{C}ov[f_i\mathbb{E}[f_i],f_j\mathbb{E}[f_j]])  \nonumber\\
    &\quad + P \mathbb{C}ov[f_i f_{i,x} - f_i^2 ,  f_j f_{j,y} - f_j^2]] \nonumber \\
    &= \frac{1}{N^2 M^2} [ (P(N-1)(M-1) - P(P-1)) \mathbb{C}ov[f_i f_{i,x} - 2f_i \mathbb{E}[f_i] ,  f_j f_{j,y} - 2f_j \mathbb{E}[f_j]]  \nonumber\\
    &\quad \scalemath{0.9}{+ P(N-1)\mathbb{C}ov[f_i f_{i,x} - 2f_i \mathbb{E}[f_i] ,  f_j f_{j,y} - f_j^2]    + P(M-1)\mathbb{C}ov[f_i f_{i,x} - f_i^2 ,  f_j f_{j,y} - 2f_j \mathbb{E}[f_j]]}  \nonumber\\
    &\quad\quad\quad + P(P-1)( \mathbb{C}ov[f_i f_{i,x} - 2f_i\mathbb{E}[f_i], f_j f_{j,y} - 2f_j\mathbb{E}[f_j]] + 2\mathbb{C}ov[f_i,f_j]^2 )  \nonumber\\
    &\quad\quad\quad\quad + P \mathbb{C}ov[f_i f_{i,x} - f_i^2 ,  f_j f_{j,y} - f_j^2]] \nonumber \\
    &= \frac{P}{N^2 M^2} [ (N-1)(M-1)  \mathbb{C}ov[f_i f_{i,x} - 2f_i \mathbb{E}[f_i] ,  f_j f_{j,y} - 2f_j \mathbb{E}[f_j]]  \nonumber\\
    &\quad \scalemath{0.9}{+ (N-1)\mathbb{C}ov[f_i f_{i,x} - 2f_i \mathbb{E}[f_i] ,  f_j f_{j,y} - f_j^2]    + (M-1)\mathbb{C}ov[f_i f_{i,x} - f_i^2 ,  f_j f_{j,y} - 2f_j \mathbb{E}[f_j]]}  \nonumber\\
    &\quad\quad\quad + 2(P-1)\mathbb{C}ov[f_i,f_j]^2  + \mathbb{C}ov[f_i f_{i,x} - f_i^2 ,  f_j f_{j,y} - f_j^2]].\nonumber\\
    &
    \label{eq:me_result}
\end{align}
Equation \eqref{eq:me_result} is the covariance across two input variables, $x$ and $y$. We now stack the covariances to include all input indices, $x,y = 1,\ldots,I$, such that the covariances are matrices. Thus,
\begin{align}
    \mathbb{C}ov [\vec{Q}_{i}(\set{N}) ,
    \vec{Q}_{j}(\set{M})] &= \frac{P}{N^2 M^2} [ (N-1)(M-1)  \mat{R}_{ij}  + (N-1) \{\mat{S}^T \}_{ij}  \nonumber\\
    &\quad + (M-1)\mat{S}_{ij} + 2(P-1)\mat{U}_{ij}  + \mat{O}_{ij}].
\end{align}
Propositions \ref{MultiSobolVar} and \ref{MultiSobolHighLow} follow by using the above equation with different sets of input samples.

\section{Proof of Proposition \ref{VarSobolCov}, \ref{VarMEVarLow}, and \ref{VarMEVarHighLow}}
The covariance between the MC variance estimator and the main effect variance estimator is found
\begin{align}
    \mathbb{C}ov[\vec{Q}_{i,u}(\set{N}), \vec{Q}_{V,j}(\set{M})] &= \scalemath{0.9}{\mathbb{C}ov\left[\frac{1}{N^2}\sum_a^N \sum_b^N f_i^{(a)} f_{i,u}^{(a)} - f_i^{(a)} f_i^{(b)}, \frac{1}{2M(M-1)} \sum_c^M \sum_d^M (f_j^{(c)} - f_j^{(d)})^2 \right]} \\
    &\scalemath{0.9}{= \frac{1}{2M(M-1)N^2} \sum_a^N \sum_b^N \sum_c^M \sum_d^M \mathbb{C}ov[f_i^{(a)} f_{i,u}^{(a)} - f_i^{(a)} f_i^{(b)}, (f_j^{(c)} - f_j^{(d)})^2]}.
\end{align}
We consider 8 cases that result in nonzero covariance terms which can be seen in Table \ref{tab:SAVarcov},

\begin{table}[h]
    \begin{center}
    \caption{Covariances resulting from combinations of shared samples}
    \label{tab:SAVarcov}
    \begin{tabular}{|c | c | c |}
    \hline
    Shared Samples & Covariance & Occurrence Frequency \\
    \hline
         $a=c\neq b \neq d$ & $\mathbb{C}ov[f_i f_{i,u} - f_i \mathbb{E}[f_i], (f_j - \mathbb{E}[f_j])^2]$ & $P(N-1)(M-1) - P(P-1)$\\
         $a=d\neq b \neq c$ & $\mathbb{C}ov[f_i f_{i,u} - f_i \mathbb{E}[f_i], (f_j - \mathbb{E}[f_j])^2]$ & $P(N-1)(M-1) - P(P-1)$ \\
         $a \neq b = c \neq d$ & $\mathbb{C}ov[-f_i \mathbb{E}[f_i], (f_j - \mathbb{E}[f_j])^2]$ & $P(N-1)(M-1) - P(P-1)$\\
         $a \neq c \neq b = d$ & $\mathbb{C}ov[-f_i \mathbb{E}[f_i], (f_j - \mathbb{E}[f_j])^2]$ & $P(N-1)(M-1) - P(P-1)$ \\
         $a = b = c \neq d$ & $\mathbb{C}ov[f_i f_{i,u} - f_i^2, (f_j - \mathbb{E}[f_j])^2]$ & $P(M-1)$\\
         $a = b = d \neq c$ & $\mathbb{C}ov[f_i f_{i,u} - f_i^2, (f_j - \mathbb{E}[f_j])^2]$ & $P(M-1)$ \\
         $a = c \neq b = d$ & Eq. \eqref{eq:sa_ac} & $P(P-1)$\\
         $a = d \neq b = c$ & Eq. \eqref{eq:sa_ac} & $P(P-1)$ \\
    \hline
    \end{tabular}
\end{center}
\end{table}

Now, the following cases are not as straightforward as the results in the table. Consider $a=c \neq b = d$
\reqnomode
\begin{align}
    \mathbb{C}ov[f_i^{(a)} f_{i,u}^{(a)} &- f_i^{(a)} f_i^{(b)}, (f_j^{(a)} - f_j^{(b)})^2] = \mathbb{C}ov[f_i^{(a)} f_{i,u}^{(a)} - f_i^{(a)} f_i^{(b)}, f_j^{(a)2} + f_j^{(b)2} - 2f_j^{(a)}f_j^{(b)}] \nonumber\\
    &= \cov{f_i f_{i,u}, f_j^2} + \cov{f_i f_{i,u}, -2f_j \ee{f_j}} + 2\cov{-f_i \ee{f_i}, f_j^2} \nonumber \tag{A}\\
    &\quad  + \cov{-f_i^{(a)} f_i^{(b)}, -2f_j^{(a)}f_j^{(b)}}. \tag{B} \nonumber
\end{align}
\leqnomode
Now, consider (B)
\begin{align}
    \cov{-f_i^{(a)} f_i^{(b)}, -2f_j^{(a)}f_j^{(b)}} &= 2\ee{f_i^{(a)} f_i^{(b)} f_j^{(a)} f_j^{(b)}} - 2\ee{f_i^{(a)} f_i^{(b)}}\ee{f_j^{(a)} f_j^{(b)}} \nonumber \\
    &= 2\ee{f_i f_j}^2 - 2\ee{f_i}^2\ee{f_j}^2 \nonumber \\
    &= 2\cov{f_i,f_j}^2 + 4\cov{f_i\ee{f_i},f_j\ee{f_j}}.
\end{align}
By combining (A) and (B),
\begin{align}
    (A) + (B) &=  \cov{f_i f_{i,u}, f_j^2 - 2f_j\ee{f_j}} + \cov{-2f_i\ee{f_i}, f_j^2 - 2f_j\ee{f_j}} + 2\cov{f_i,f_j}^2 \nonumber \\
    &= \mathbb{C}ov[f_i f_{i,u}-2 f_i\mathbb{E}[f_i], (f_j-\mathbb{E}[f_j])^2] + 2\mathbb{C}ov[f_i,f_j]^2.
    \label{eq:sa_ac}
\end{align}
We note here that the argument is the same for $a=d \neq b = c$ with an identical result. These results are then combined
\begin{align}
    \mathbb{C}ov&[\vec{Q}_{i,u}(\set{N}), \vec{Q}_{V,j}(\set{M})] =  \scalemath{0.9}{\frac{1}{2M(M-1)N^2} \sum_a^N \sum_b^N \sum_c^M \sum_d^M \mathbb{C}ov[f_i^{(a)} f_{i,u}^{(a)} - f_i^{(a)} f_i^{(b)}, (f_j^{(c)} - f_j^{(d)})^2]} \nonumber\\
    &= \frac{1}{2M(M-1)N^2} \left[(2P(N-1)(M-1) - 2P(P-1))\mathbb{C}ov[f_i f_{i,u} - f_i \mathbb{E}[f_i], (f_j - \mathbb{E}[f_j])^2] \right. \nonumber \\
    &\quad - (2P(N-1)(M-1) - 2P(P-1))\mathbb{C}ov[f_i \mathbb{E}[f_i], (f_j - \mathbb{E}[f_j])^2]  \nonumber\\
    &\quad\quad + 2P(M-1) \mathbb{C}ov[f_i f_{i,u} - f_i^2, (f_j - \mathbb{E}[f_j])^2]  \nonumber\\
    &\quad\quad\quad + 2P(P-1)[ \mathbb{C}ov[f_i f_{i,u} - 2f_i \mathbb{E}[f_i], (f_j - \mathbb{E}[f_j])^2] + 2\mathbb{C}ov[f_i,f_j]^2]] \nonumber\\
    &= \frac{P(N-1)}{MN^2}\mathbb{C}ov[f_i f_{i,u} - 2 f_i \mathbb{E}[f_i], (f_j - \mathbb{E}[f_j])^2]  \nonumber\\
    &\quad + \frac{P}{MN^2}\mathbb{C}ov[f_i f_{i,u} - f_i^2, (f_j - \mathbb{E}[f_j])^2]   + \frac{2P(P-1)}{M(M-1)N^2}\mathbb{C}ov[f_i,f_j]^2. \nonumber \\
    &
\end{align}
Similarly, by taking its transpose,
\begin{align}
    \mathbb{C}ov[\vec{Q}_{V,j}(\set{M}), \vec{Q}_{i,u}(\set{N})] &= \frac{P(N-1)}{MN^2}\mathbb{C}ov[(f_j - \mathbb{E}[f_j])^2, f_i f_{i,u} - 2 f_i \mathbb{E}[f_i]]  \nonumber\\
    &\quad + \frac{P}{MN^2}\mathbb{C}ov[(f_j - \mathbb{E}[f_j])^2, f_i f_{i,u} - f_i^2]   + \frac{2P(P-1)}{M(M-1)N^2}\mathbb{C}ov[f_j,f_i]^2.
    \label{eq:mev_result}
\end{align}
Equation \eqref{eq:mev_result} defines the covariance for one input of interest. We now stack the covariances to include all input indices, $u = 1,\ldots,I$,
\begin{align}
    \mathbb{C}ov[\vec{Q}_{i,\vec{x}}(\set{N}), \vec{Q}_{V,j}(\set{M})] &=\frac{P(N-1)}{MN^2} {\vec{E}}_{ij} + \frac{P}{MN^2} {\vec{C}}_{ij} + \frac{2P(P-1)}{M(M-1)N^2} {\vec{U}}_{ij,0} \\
    \mathbb{C}ov[\vec{Q}_{V,j}(\set{M}), \vec{Q}_{i,\vec{x}}(\set{N})] &= \frac{P(N-1)}{MN^2} \{\vec{E}^T\}_{ji}  + \frac{P}{MN^2}\{\vec{C}^T\}_{ji}   + \frac{2P(P-1)}{M(M-1)N^2}\{\vec{U}_{j0,0}\}^T
\end{align}
where ${\vec{U}}_{ij,0} \in \reals^I$ is the first column of $\vec{U}_{ij}$. Now, Propositions \ref{VarMEVarLow} and \ref{VarMEVarHighLow} follow by using the above equations with different input sample sets.

\end{document}